%% file: wyner_ziv_v6.tex
\documentclass[onecolumn,10pt]{IEEEtran}
\input{preamble.tex}
\hypersetup{linkcolor=black,citecolor=black,urlcolor=black,anchorcolor=black}
\providecommand{\Pcn}{P_{\mathrm{c}}^{(n)}}

\renewcommand{\jW}[1]{\jmath_{\mathrm{WZ}}(#1|D)}
\renewcommand{\labelenumi}{(\roman{enumi})}
\renewcommand{\theenumi}{\roman{enumi}}
\allowdisplaybreaks[2]
\begin{document}
\title{Exact Second-Order Asymptotics for the Discrete Wyner--Ziv Problem}

\author{Lin Zhou}

\maketitle
\begin{abstract}
We revisit the discrete Wyner--Ziv problem and establish exact second-order asymptotics. Our main contribution is a second-order converse bound, which is derived using posterior decomposition, concentration inequalities, and a conditional Gaussian approximation for martingales. Furthermore, we extend the previous best known achievability result of Li and Li (arXiv 2025) by allowing the test channel to depend on the type of the observed source sequence. Exact second-order asymptotics are established by combining our achievability and converse bounds. In particular, we show that the achievability bound of Li and Li (arXiv 2025) is not optimal in general. Specifically, we provide two numerical examples to illustrate our results: a binary asymmetric source and a quaternary source. For the first example, the achievability bound of Li and Li (arXiv 2025) achieves the optimal second-order asymptotics. However, in the second example, we show that our achievability result is optimal, which reduces the second-order coding rate of Li and Li (arXiv 2025) by $10.77\%$ at the excess-distortion probability of $0.1$ and by $21.35\%$ at the excess-distortion probability of $0.2$.
\end{abstract}
\begin{IEEEkeywords}
Lossy source coding, Rate-distortion, Dispersion, Side information, Discrete memoryless sources.
\end{IEEEkeywords}

\section{Introduction}
\label{sec:intro}
Shannon~\cite{shannon1959coding} pioneered the study of lossy source coding, a.k.a., rate-distortion, to allow compression with imperfection, which is valid for compression of continuous sequences and could further reduce the transmission rate of lossless compression of discrete sequences. In particular, Shannon's rate-distortion theorem established the minimal asymptotic rate achieved by an optimal code under the expected distortion constraint as the blocklength tends to infinity, which is termed the rate-distortion function. The same asymptotic theoretical benchmark holds when one requires a vanishing excess-distortion probability. Shannon's result is known as first-order asymptotics for lossy source coding since it characterizes the first dominant term for the theoretical benchmark. Several decades later, in 2010s, the finite-blocklength performance and its second-order asymptotic approximation was derived by Ingber and Kochman~\cite{ingber2011dispersion} and by Kostina and Verd\'u~\cite{kostina2012}, which refined Shannon's result by further characterizing the second dominant term of the theoretical benchmark. In particular, second-order asymptotics was derived under the excess-distortion probability criterion, showing how the performance of optimal codes depends on the blocklength at a lead to n non-vanishing excess-distortion probability. It is well understood that second-order asymptotics provides good approximation for the finite-blocklength performance of optimal codes for lossy source coding when the blocklength is moderately large.

Motivated by practical applications including compression of correlated sensor measurements~\cite{kashyap}, multi-view imaging, and distributed video coding~\cite{girodvideo,morbee}, Wyner and Ziv~\cite{wyner1976rate} generalized the point-to-point setting of Shannon to the case with decoder side information. Specifically, the encoder compresses the observed source sequence and the decoder generates a reconstruction sequence using both compressed information and a side-information sequence that is correlated with the source sequence. Wyner and Ziv derived the first-order asymptotics under the expected distortion criterion, showing that the asymptotic minimal compression rate is lead to n by an optimization problem that involves an auxiliary random variable that forms a Markov chain with the source sequence and the side-information sequence. The same result also holds if one requires a vanishing excess-distortion probability. We remark that the Wyner--Ziv problem has a first-order characterization involving an auxiliary random variable that forms a Markov chain with the source and side information. Analogous to the Wyner--Ziv problem, there is a lossless source coding problem with a rate-limited helper named the Wyner--Ahlswede--K\"orner problem~\cite{Ahl75}, where the first-order asymptotics also involves an auxiliary random variable and a Markov chain. Recently, the second-order asymptotics for the Wyner--Ahlswede--K\"orner problem was established with the achievability result of Liu~\cite[Theorem~2]{liu2021wak} and the converse result of Cao~\cite{cao2026wak}.

In the past decades, the result of Wyner and Ziv was refined in several aspects, including a strong converse result~\cite{oohama,tyagiwatanabe}, a strong converse exponent result~\cite{watanabe2025}, and finite-blocklength achievability results and corresponding second-order asymptotic results~\cite[Theorems~8,~18, and~19]{watanabe2015}, \cite[Theorem~3]{poisson}, \cite[Theorem~15]{lili}. Specifically, the strong converse result shows that the excess-distortion probability tends to one when the compression rate is below the first-order asymptotics while the strong converse exponent result further derives the exact speed at which the excess-distortion probability tends to one. The best known second-order achievability result was~\cite[Theorem~15]{lili}, which was conjectured to be optimal by the authors.

However, a matching second-order converse and exact second-order asymptotics for the discrete Wyner--Ziv problem remain open. A critical question is whether the achievability bound of Li and Li~\cite[Theorem~15]{lili} is optimal. In this paper, we answer this question negatively: the result in~\cite[Theorem~15]{lili} is not second-order asymptotically optimal in general.

\subsection{Main Contributions}
\label{sec:contributions}
Our main contribution is the characterization of the exact second-order asymptotics for the discrete Wyner--Ziv problem, including a converse bound and an achievability bound that refines Li and Li~\cite[Theorem~15]{lili}.

Our first contribution is a second-order converse bound, whose proof proceeds in three steps. Firstly, we use a posterior decomposition and an optimality-gap bound to identify the optimal test channels for the first-order asymptotics and upper bound the non-excess-distortion probability of any code by the probability of a joint event for the compression rate and distortion level, together with a vanishing term. Secondly, we use concentration inequalities and a modification of the code to replace the thresholds and conditional covariance in the above event by quantities determined by the source sequence. Finally, conditioned on the source sequence, we use Gaussian approximation for martingales to bound the dominant term of the non-excess-distortion probability by the probability of an event involving a two-dimensional Gaussian random vector and subsequently average over the distribution of the source sequence to obtain the desired converse bound.

Our second contribution is a generalized achievability result of~\cite[Theorem~15]{lili}. In contrast to the code construction in~\cite[Theorem~15]{lili} that fixes a first-order optimal test channel while adapting its local perturbation to the type of the source sequence, our construction selects the optimal test channel itself using the type of the source sequence. Combining our achievability and converse results leads to the exact second-order asymptotics for the discrete Wyner--Ziv problem. Furthermore, we provide two numerical examples to illustrate our results and answer the question on the optimality of the result in~\cite[Theorem~15]{lili}. In our first example, we consider a binary source and a binary asymmetric channel that describes the correlation of the source and side information; for this example, we show that the achievability bound of~\cite[Theorem~15]{lili} is second-order asymptotically optimal. However, in our second example of a quaternary source, we show that the achievability result of~\cite[Theorem~15]{lili} is strictly suboptimal. In particular, our optimal result reduces the second-order coding rate of~\cite[Theorem~15]{lili} by $10.77\%$ at the excess-distortion probability of $0.10$ and by $21.35\%$ at the excess-distortion probability of $0.20$.

\subsection{Organization for the Rest of the Paper}
\label{sec:organization}
The rest of the paper is organized as follows. In Section~\ref{sec:formulation}, we formulate the problem, define codes and theoretical benchmarks and present existing results on first-order asymptotics and second-order achievability result. In Section~\ref{sec:results}, we present our main results with discussions and provide two illustrating numerical examples. The converse and achievability proofs of our results are presented in Sections~\ref{sec:converse} and~\ref{sec:achievability}, respectively. Finally, we conclude the paper and discuss future directions in Section~\ref{sec:conclusion}. For smooth presentation of our main results and discussions, the proofs of supporting lemmas and calculations for numerical examples are deferred to appendices.

\section{Problem Formulation, Assumptions, and Existing Results}
\label{sec:formulation}
\subsection*{Notation}\label{sec:notation}
Random variables and realizations use uppercase and lowercase letters. Bold symbols denote vectors, while distributions use ordinary uppercase symbols. Logarithms and $\exp$ use base $e$. The standard Gaussian cumulative distribution function is $\Phi$. We use calligraphic letters to denote sets. We use $\bbR$, $\bbR_+$ and $\bbN$ to denote the sets of real numbers, nonnegative real numbers, and positive integers, respectively. For any positive integer $n\in\bbN$, we use $[n]$ to denote the set of integers $\{1,\ldots,n\}$ and use $X^n:=(X_1,\ldots,X_n)$ to denote a collection of random variables; furthermore, lead to n any two positive integers $(i,j)\in[n]^2$ such that $i\leq j$, we use $X_i^j$ to denote $(X_i,\ldots,X_j)$. For any real number $a\in(0,1)$, we use $h_2(a):=-a\log a-(1-a)\log(1-a)$ to denote binary entropy, and use $\Ber(a)$ to denote the Bernoulli distribution with probability of one being $a$. For any positive integer $k\in\bbN$, we use $\bzero_k:=(0,\ldots,0)^{\mathsf T}\in\bbR^k$ to denote the all-zero vector of length $k$. For any nonnegative real number ${{\sigma^2}}\in\bbR_+$, we use ${\calN}(0,{{\sigma^2}})$ to denote the Gaussian probability distribution function with mean zero and variance ${\sigma^2}$; for any positive integer $k\in\bbN$ and a positive semidefinite matrix $\Sigma\in\bbR^{k\times k}$, we use ${\calN}(\bzero_k,\Sigma)$ to denote the $k$-dimensional Gaussian distribution with zero mean vector and covariance matrix $\Sigma$. We use $\bbI\{\cdot\}$ for indicators, $H$ for entropy, $I$ to denote mutual information, and ${\operatorname{KL}}(\cdot\|\cdot)$ for relative entropy (KL divergence). We use $\oplus$ to denote addition modulo two. We use $\langle\cdot,\cdot\rangle$, $\|\cdot\|$, and $\|\cdot\|_2$ to denote Euclidean inner products, Euclidean norms, and matrix spectral norms, respectively. We use ${{\mathrm{Id}_2}}$ to denote the $2\times2$ identity matrix, and $\Rightarrow$ to denote weak convergence. lead to n finite alphabets $\calX$ and $\calY$, we use $\calP(\calX)$ and $\calP(\calY|\calX)$ to denote the sets of all probability mass functions and conditional distributions, respectively. The symbols $P\{\cdot\}$ and $\bbP\{\cdot\}$ denote probability under the distribution specified locally; $\bbE_P$, $\var_P$, and $\cov_P$ denote expectation, variance, and covariance, with the subscript distribution omitted when it is clear. We write ${\calS}(\cdot)$ for the sigma-field generated by the specified random variables.

\begin{figure}[!t]
\centering
\begin{tikzpicture}[>=Latex,thick,block/.style={draw,rounded corners=2pt,minimum width=2.0cm,minimum height=1.05cm,align=center}]
\node (x) at (-1.6,0) {$X^n$};
\node[block] (enc) at (0.6,0) {Encoder\\$f_n$};
\node[block] (dec) at (5.0,0) {Decoder\\$\phi_n$};
\node (z) at (7.6,0) {$Z^n$};
\node (y) at (-1.6,-1.55) {$Y^n$};
\node[font=\small] at (-1.6,-.775) {$P_{XY}^{n}$};
\draw[->] (x) -- (enc);
\draw[->] (enc) -- node[above,font=\small] {$f_n(X^n)$} (dec);
\draw[->] (dec) -- (z);
\draw[->] (y) -| (dec.south);
\end{tikzpicture}
\caption{System model for the Wyner--Ziv problem of lossy source coding with decoder side information.}\label{fig:system}

\end{figure}
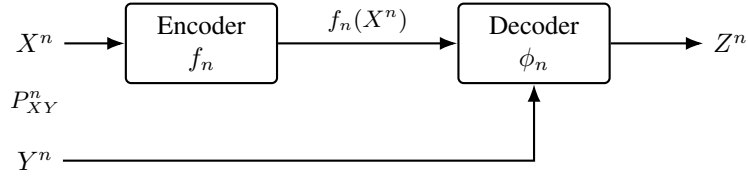

\subsection{Problem Formulation}
\label{sec:problem}
Let $\calX$, $\calY$, and $\calZ$ be three finite alphabets, from which the source sequence, side-information sequence, and reproduction sequence take values, respectively. We consider a discrete memoryless source with joint distribution $P_{XY}$ such that for each $n\in\bbN$ and $(x^n,y^n)\in\calX^n\times\calY^n$,
\begin{align}
P_{XY}^n(x^n,y^n)=\prod_{i\in[n]} P_{XY}(x_i,y_i).\label{eq:product}
\end{align}
Fix a positive real number $0<d_{\max}<\infty$. Let $d:\calX\times\calZ\to[0,d_{\max}]$ be a bounded single-letter distortion measure. For any $(x^n,z^n)\in\calX^n\times\calZ^n$, define the following additive distortion function:
\begin{align}
{d^{(n)}}(x^n,z^n):=\frac1n\sum_{i\in[n]}d(x_i,z_i).\label{eq:blockdist}
\end{align}
As shown in Fig.~\ref{fig:system}, in the Wyner--Ziv problem, the encoder $f_n$ compresses the source sequence $X^n$ to one of $M$ messages, and the decoder $\phi_n$ generates a source sequence reconstruction $Z^n$ using the compressed message and the side-information sequence $Y^n$.

Fix two integers $(n,M)\in\bbN^2$.
\begin{definition}\label{def:code}
An $(n,M)$-code for the Wyner--Ziv problem consists of an encoder $f_n$ and a decoder $\phi_n$, defined as follows:
\begin{align}
f_n&:\calX^n\to[M],\\
\phi_n&:[M]\times\calY^n\to\calZ^n.
\end{align}
The reconstruction sequence is $Z^n=\phi_n(f_n(X^n),Y^n)$. lead to n any distortion level $D>0$, the excess-distortion probability is defined as
\begin{align}
\Pen(D):=\bbP\{{d^{(n)}}(X^n,Z^n)>D\}.\label{eq:error}
\end{align}
\end{definition}

With Definition~\ref{def:code}, we can now define the finite-blocklength theoretical benchmarks. Fix any distortion level $D>0$ and tolerable excess-distortion probability $\veps\in(0,1)$. Let
\begin{align}
M^*(D,\veps)&:=\min\{M\in\bbN:\;\text{an $(n,M)$-code has }\Pen(D)\leq\veps\},\label{eq:fundamental}\\
R_n^*(D,\veps)&:=\frac1n\log M^*(D,\veps).
\end{align}
Note that $M^*(D,\veps)$ is the minimal message size to compress a length-$n$ sequence $X^n$ with decoder side information $Y^n$ subject to an excess-distortion constraint of at most $\veps$ with respect to the distortion level $D$. Correspondingly, $R_n^*(D,\veps)$ is the corresponding compression rate per source symbol.

To define the second-order asymptotics, let $R^*(D,\veps)$ denote the first-order asymptotic rate with excess-distortion probability $\veps$, i.e.,
\begin{align}
R^*(D,\veps):=\lim_{n\to\infty}R_n^*(D,\veps).\label{eq:first-order-limit}
\end{align}

\begin{definition}\label{def:second}
A real number $L\in\bbR$ is said to be second-order $(D,\veps)$-achievable if there exists a sequence of $(n,M)$-codes such that
\begin{align}
\limsup_{n\to\infty}\frac{\log M-nR^*(D,\veps)}{\sqrt n}&\leq L,\\
\limsup_{n\to\infty}\Pen(D)&\leq\veps.
\end{align}
The infimum of all such numbers is denoted by $L_{\mathrm{WZ}}^*(D,\veps)$.
\end{definition}
The quantity $L_{\mathrm{WZ}}^*(D,\veps)$ is the smallest backoff from the first-order asymptotics $R^*(D,\veps)$. In this paper, our main contribution is the exact characterization of the second-order asymptotics $L_{\mathrm{WZ}}^*(D,\veps)$.

\subsection{Assumptions}
\label{sec:assumptions}
To present our results, we need the following assumptions on the distortion level, the source distribution and the transition matrix (channel) that describes the correlation between the source sequence and the side information.

The side-information channel is $P_{Y|X}$. Define the following two distortion values
\begin{align}
D_{\min}&:=\sum_{x\in\calX}P_X(x)\min_{z\in\calZ} d(x,z),\label{eq:Dmin}\\
D_0&:=\sum_{y\in\calY}\min_{z\in\calZ}\sum_{x\in\calX}P_X(x)P_{Y|X}(y|x)d(x,z).\label{eq:D0}
\end{align}
The quantity $D_{\min}$ is the minimum possible expected distortion while $D_0$ is the minimum expected distortion attainable using only the side information. Thus, any targeted distortion level should satisfy that $D_{\min}<D<D_0$.

The following assumptions are needed.
\begin{enumerate}
\renewcommand{\theenumi}{A\arabic{enumi}}
\renewcommand{\labelenumi}{(\theenumi)}
\item\label{ass:finite} The alphabets $\calX$, $\calY$, and $\calZ$ are finite and the joint distribution $P_{XY}$ satisfies that $\min_{x\in\calX}P_X(x)>0$ and $\min_{y\in\calY}P_Y(y)>0$, where $P_Y$ is the induced marginal distribution of $P_{XY}$
\item\label{ass:overlap} The channel $P_{Y|X}$ satisfies that
\begin{align}
\kappa_W&:=\min_{(x,x')\in\calX^2}\sum_{y\in\calY}P_{Y|X}(y|x)P_{Y|X}(y|x'),>0.\label{eq:overlap}
\end{align}
Equivalently, no two source symbols can be distinguished without error from one side information symbol.
\item\label{ass:interior} The distortion threshold and error probability are fixed independently of $n$ and satisfy $D_{\min}<D<D_0$.
\end{enumerate}

We provide two examples of the channel and distortion function that satisfy the above conditions, which are also used to illustrate our second-order asymptotics in subsequent sections.
\begin{enumerate}
\item \emph{Binary asymmetric source.} Let $\calX=\calY=\calZ=\{0,1\}$. Fix $(q,{{p_{01}}},{{p_{10}}})\in(0,1)^3$. Let $P_X=\Ber(q)$ be the Bernoulli distribution with parameter $q$, and let $d(x,z)=\bbI\{x\ne z\}$. The channel matrix $P_{Y|X}$ is lead to n by
\begin{align}
W_{\rm BAC}=\begin{pmatrix}1-{{p_{01}}}&{{p_{01}}}\\{{p_{10}}}&1-{{p_{10}}}\end{pmatrix}.\label{eq:bac-channel-general}
\end{align}
The induced marginal distribution $P_Y$ satisfies
\begin{align}
P_Y(0)&=(1-q)(1-{{p_{01}}})+q{{p_{10}}}>0,\label{eq:bac-output}\\
P_Y(1)&=(1-q){{p_{01}}}+q(1-{{p_{10}}})>0.\label{eq:bac-output-one}
\end{align}
Thus, Assumption~(\ref{ass:finite}) holds. It follows from the definition of $\kappa_W$ in~\eqref{eq:overlap} that
\begin{align}
\kappa_{W_{\rm BAC}}&=\min\{(1-{{p_{01}}})^2+{{p_{01}}}^2,\ {{p_{10}}}^2+(1-{{p_{10}}})^2,\ (1-{{p_{01}}}){{p_{10}}}+{{p_{01}}}(1-{{p_{10}}})\},>0.\label{eq:bac-overlap-general}
\end{align}
Thus, Assumption~(\ref{ass:overlap}) holds. It follows from the definitions of $D_{\min}$ and $D_0$ in~\eqref{eq:Dmin}--\eqref{eq:D0} that
\begin{align}
D_{\min}&=0,\label{eq:bac-distortion-endpoints}\\
D_0&=\min\{(1-q)(1-{{p_{01}}}),q{{p_{10}}}\}+\min\{(1-q){{p_{01}}},q(1-{{p_{10}}})\}>0.
\end{align}
Thus, any $D\in(0,D_0)$ satisfies Assumption~(\ref{ass:interior}).

\item \emph{Quaternary source.} Let $X=({{X_{\rm p}}},{{X_{\rm b}}})\in\{0,1\}^2$ be a quaternary source, where ${{X_{\rm p}}}\sim\Ber(1/3)$ and ${{X_{\rm b}}}\sim\Ber(1/2)$ are independent Bernoulli random variables. Let $Y={{X_{\rm b}}}\oplus {{N}}$, where ${{N}}\sim\Ber(2/5)$ is another Bernoulli random variable independent of $({{X_{\rm b}}},{{X_{\rm p}}})$. The reproduction alphabet is $\calZ=\{0,1,*\}\times\{0,1\}$. With source symbols ordered as $(0,0),(0,1),(1,0),(1,1)$ and side-information symbols as $0,1$, the channel $P_{Y|X}$ is lead to n by
\begin{align}
W_4=\begin{pmatrix}3/5&2/5\\2/5&3/5\\3/5&2/5\\2/5&3/5\end{pmatrix}.\label{eq:four-symbol-channel}
\end{align}
As a result, the source distribution is $P_X=(1/3,1/3,1/6,1/6)$ and the side-information distribution is $P_Y=(1/2,1/2)$. Thus, Assumption~(\ref{ass:finite}) holds. Furthermore, the definition in~\eqref{eq:overlap} leads to
\begin{align}
\kappa_{W_4}&=\min\{13/25,12/25\},=12/25>0,\label{eq:four-symbol-overlap}
\end{align}
verifying Assumption~(\ref{ass:overlap}). Assumption~(\ref{ass:interior}) is verified in Appendix~\ref{app:four-symbol-model}.
\end{enumerate}

\subsection{Existing Results}
\label{sec:first-order}
Fix an integer $m_0\in[|\calX|+1]$. Let $U$ be an auxiliary random variable defined on the set $\calU=[m_0]$. Define the set of reconstruction functions by
\begin{align}
{{\calG}}_{{m_0}}:=\{{{g}}:[{m_0}]\times\calY\to\calZ\},\label{eq:reconstruction-map-set}
\end{align}
For any function $g\in\calG_{m_0}$, $g(u,y)$ is the reproduction associated with $u\in\calU$ and $y\in\calY$. Fix any distortion level $D\in(D_{\min},D_0)$. Define the following rate-distortion function
\begin{align}
R_{\mathrm{WZ}}(D):=\min_{\substack{{m_0}\in[|\calX|+1],\;P_{{{U}}|X}\in\calP([{m_0}]|\calX),\;{{g}}\in{{\calG}}_{{m_0}}\\
\bbE[d(X,{{g}}({{U}},Y))]\leq D}}I(X;{{U}}|Y),\label{eq:wzclassic}
\end{align}
where the expectation and mutual information are calculated with respect to the joint distribution $P_{XY}P_{U|X}$ and $U-X-Y$ forms a Markov chain.

Recall the definition of the first-order asymptotic rate $R^*(D,\veps)$ in~\eqref{eq:first-order-limit}. Wyner and Ziv~\cite{wyner1976rate} proved that
\begin{align}
\lim_{\veps\to0}R^*(D,\veps)=R_{\mathrm{WZ}}(D).\label{eq:classical-first-order-limit}
\end{align}
The strong converse result further shows that for any $\veps\in(0,1)$,
\begin{align}
R^*(D,\veps)=R_{\mathrm{WZ}}(D).\label{eq:first-order-strong-converse}
\end{align}
Thus, the first-order asymptotic rate remains the same for any excess-distortion probability $\veps\in(0,1)$.

We next recall the second-order achievability result of~\cite[Theorem~15]{lili}. Let $\widetilde{\calU}:=[|\calZ|^{|\calY|}]$ index all functions from $\calY$ to $\calZ$, and fix an enumeration $\widetilde g:\widetilde{\calU}\times\calY\to\calZ$. The auxiliary $\widetilde U$ is the reconstruction-rule index, so $\widetilde g(\widetilde u,\cdot)$ is the rule indexed by $\widetilde u\in\widetilde{\calU}$. The rate-distortion function $R_{\mathrm{WZ}}(D)$ in~\eqref{eq:wzclassic} has the equivalent expression
\begin{align}
R_{\mathrm{WZ}}(D)=\min_{\substack{P_{{\widetilde U}|X}\in\calP({\widetilde{\calU}}|\calX)\\
\bbE_P[d(X,{\widetilde g}({\widetilde U},Y))]\leq D}}I_P(X;{\widetilde U}|Y),\label{eq:fixed-alphabet}
\end{align}
where $P=P_{XY}P_{{\widetilde U}|X}$ denotes the joint distribution where the expectation and mutual information are calculated with respect to, which is proved in Appendix~\ref{app:rule-alphabet}. Define the set of optimizers for $R_{\mathrm{WZ}}(D)$ by
\begin{align}
\calP^*(D):=\{P=P_{XY}P_{{\widetilde U}|X}:\;P_{{\widetilde U}|X}\in\calP({\widetilde{\calU}}|\calX),\;
I_P(X;{\widetilde U}|Y)=R_{\mathrm{WZ}}(D),\;\bbE_P[d(X,{\widetilde g}({\widetilde U},Y))]\leq D\}.\label{eq:optimal-set}
\end{align}
Define the optimal dual parameter as
\begin{align}
\lambda^*:=-R_{\mathrm{WZ}}'(D).\label{eq:multiplier-definition}
\end{align}
Proposition~\ref{prop:optimal-properties} below shows that $\lambda^*$ is well defined and strictly positive under Assumption~\ref{ass:interior}.

Fix any $P\in\calP^*(D)$, let $P_{\widetilde U}$ be the induced marginal distribution on ${\widetilde{\calU}}$ by $P_{{\widetilde U}|X}$ and $P_{XY}$. For any $(x,y,{\widetilde u})\in\calX\times\calY\times{\widetilde{\calU}}$ such that $P_{XY}(x,y)P_{{\widetilde U}|X}({\widetilde u}|x)>0$, define the following quantities:
\begin{align}
P_{{\widetilde U}|Y}({\widetilde u}|y)&:=\frac{\sum_{x'\in\calX}P_X(x')P_{{\widetilde U}|X}({\widetilde u}|x')P_{Y|X}(y|x')}{P_Y(y)},\label{eq:optimal-decoder-law}\\
\imath_P(x,{\widetilde u},y)&:=\log\frac{P_{{\widetilde U}|X}({\widetilde u}|x)}{P_{{\widetilde U}|Y}({\widetilde u}|y)},\label{eq:optimal-information}\\
\jmath_P(x|D)&:=\bbE_P[\imath_P(X,{\widetilde U},Y)+\lambda^*(d(X,{\widetilde g}({\widetilde U},Y))-D)|X=x],\label{eq:existing-tilted}\\
V(P)&:=\var_{P_X}[\jmath_P(X|D)],\label{eq:existing-source-variance}\\
\Sigma(P)&:=\bbE_P\left[\cov_P\left(
\begin{pmatrix}\imath_P(X,{\widetilde U},Y)\\d(X,{\widetilde g}({\widetilde U},Y))\end{pmatrix}\middle|X,{\widetilde U}\right)\right].\label{eq:cov}
\end{align}
Consistent with~\cite[Definition~6]{kostina2012}, $\jmath_P(\cdot|D)$ can be understood as the distortion-tilted information density of the Wyner-Ziv problem. The quantity $V(P)$ is the variance of the distortion-tilted information density, which results from the randomness of the source sequences at distortion level $D$. Furthermore, the covariance matrix $\Sigma(P)$ reflects the expected conditional covariance matrix of the compression rate and the achievable distortion level as a function of side information, reflecting the uncertainty caused by side information lead to n the source sequence and the compressed message defined via ${\widetilde U}$.

Let ${{\Lambda}}\sim {\calN}(0,V(P))$ be a Gaussian random variable with mean zero and variance $V(P)$. lead to n a positive semidefinite $2\times2$ matrix $\Sigma$, let $(J_1,J_2)\sim {\calN}(\bzero_2,\Sigma)$ be a two-dimensional Gaussian vector and define the following Gaussian cdf for any $({{\ell_1}},{{\ell_2}})\in\bbR^2$:
\begin{align}
\Phi_\Sigma({{\ell_1}},{{\ell_2}}):=\bbP\{J_1\leq {{\ell_1}},J_2\leq {{\ell_2}}\}.\label{eq:gaussian-cdf}
\end{align}
Furthermore, for any $L\in\bbR$ and $\veps\in(0,1)$, let
\begin{align}
F_{\mathrm{f}}(L)&:=\sup_{P\in\calP^*(D)}\bbE_{{\Lambda}}\Bigl[\sup_{t\in\bbR}
\Phi_{\Sigma(P)}(L-{{\Lambda}}-\lambda^*t,t)\Bigr],\label{eq:fixed-versus-adaptive}\\
L_{\mathrm{f}}(\veps)&:=\inf\{L\in\bbR:F_{\mathrm{f}}(L)\geq1-\veps\}.\label{eq:fixed-quantile-definition}
\end{align}
In~\eqref{eq:fixed-versus-adaptive}, the parameter $t$ describes the tradeoff between the compression rate and the achievable distortion level. As presented in the next lemma, $L_{\mathrm{f}}(\veps)$ was proved to be the achievable second-order asymptotics in~\cite[Theorem~15]{lili} while $F_{\mathrm{f}}(L)$ is the asymptotic non-excess-distortion probability with a second-order coding rate of $L$.

\begin{lemma}\label{lem:lili}
Under Assumptions~(\ref{ass:finite})--(\ref{ass:interior}), for every $\veps\in(0,1)$ such that $L_{\mathrm{f}}(\veps)\geq0$,
\begin{align}
L_{\mathrm{WZ}}^*(D,\veps)\leq L_{\mathrm{f}}(\veps).\label{eq:lili-coefficient}
\end{align}
\end{lemma}
Lemma~\ref{lem:lili} was adapted from~\cite[Theorem~15]{lili}. Too see this, for any $P\in\calP^*(D)$ and $(x,{\widetilde u})\in\calX\times{\widetilde{\calU}}$, define
\begin{align}
\bard(x,{\widetilde u})&:=\sum_{y\in\calY}P_{Y|X}(y|x)d(x,{\widetilde g}({\widetilde u},y)),\label{eq:mean-distortion-rule}\\
\bard_P(x)&:=\sum_{{\widetilde u}\in{\widetilde{\calU}}}P_{{\widetilde U}|X}({\widetilde u}|x)\bard(x,{\widetilde u}).\label{eq:mean-distortion-source}
\end{align}
The conditions of~\cite[Theorem~15]{lili} include i) the parameter $\lambda^*$ is well defined, strictly positive, and finite and ii) the following expected variance is strictly positive:
\begin{align}
{{{\widetilde V}}}(P)&:=\bbE_{P_X}\left[\var_{P_{{\widetilde U}|X}(\cdot|X)}[\bard(X,{\widetilde U})]\right]
=\sum_{\substack{x\in\calX\\{\widetilde u}\in{\widetilde{\calU}}}}P_X(x)P_{{\widetilde U}|X}({\widetilde u}|x)\bigl(\bard(x,{\widetilde u})-\bard_P(x)\bigr)^2.\label{eq:existing-direction}
\end{align}
Under Assumptions~(\ref{ass:finite})--(\ref{ass:interior}), both conditions are shown satisfied in our Proposition~\ref{prop:optimal-properties} below. One might find that the covariance matrix in~\cite[Theorem~15]{lili} takes a different form from $\Sigma(P)$ in~\eqref{eq:cov}. However, note that conditioned on $(X,{\widetilde U})$, the variance of $-\log(P_{{\widetilde U}|Y}({\widetilde U}|Y)/P_{\widetilde U}({\widetilde U}))$ is the same as $\imath_P(X,{\widetilde U},Y)-\log(P_{{\widetilde U}|X}({\widetilde U}|X)/P_{\widetilde U}({\widetilde U}))$, which shows that the covariance matrix of~\cite[Theorem~15]{lili} is exactly $\Sigma(P)$.

Furthermore, we remark that Lemma~\ref{lem:lili} requires that $L_{\mathrm{f}}(\veps)\geq0$, whose corresponding range of excess-distortion probabilities $\veps$ depends on $F_{\mathrm{f}}$ and may not cover the whole range $(0,1)$. In this paper, our derived achievability result is valid for any $\veps\in(0,1)$.

\section{Main Results and Discussions}
\label{sec:results}
\subsection{Properties of Optimal Distributions and Distortion-Tilted Information Density}
\label{sec:optimal}
Recall the definitions of $R_{\mathrm{WZ}}(D)$, $\calP^*(D)$, $\lambda^*$, $\imath_P(\cdot)$, $\jmath_P(\cdot)$, and $\Sigma(P)$ in~\eqref{eq:fixed-alphabet},~\eqref{eq:optimal-set},~\eqref{eq:multiplier-definition},~\eqref{eq:optimal-information},~\eqref{eq:existing-tilted}, and~\eqref{eq:cov}, respectively.

\begin{proposition}\label{prop:optimal-properties}
Under Assumptions~(\ref{ass:finite})--(\ref{ass:interior}), the following results hold.
\begin{enumerate}
\item The rate-distortion function $R_{\mathrm{WZ}}(D)$ is differentiable at $D$ and $\lambda^*=-R_{\mathrm{WZ}}'(D)$ is strictly positive and finite.
\item The set $\calP^*(D)$ is nonempty and compact, and for each $P\in\calP^*(D)$,
\begin{align}
\bbE_P[d(X,{\widetilde g}({\widetilde U},Y))]=D.\label{eq:optimal-active-constraint}
\end{align}
\item The distortion-tilted information density $\jmath_\cdot(\cdot|D)$ in~\eqref{eq:existing-tilted} has a common value for any $P\in\calP^*(D)$, i.e., for any $(P,P')\in\calP^*(D)^2$ and $x\in\calX$,
\begin{align}
\jmath_P(x|D)=\jmath_{P'}(x|D){=:\jW{x}}.\label{eq:common-optimal-density}
\end{align}
\item The expected value of the distortion-tilted information density equals the first-order asymptotics, i.e., for any $P\in\calP^*(D)$,
\begin{align}
\bbE_P[\jW{X}]=R_{\mathrm{WZ}}(D).\label{eq:tilted-mean}
\end{align}
\item For any $P\in\calP^*(D)$, the variance ${{{\widetilde V}}}(P)$ in~\eqref{eq:existing-direction} is strictly positive.
\end{enumerate}
\end{proposition}
The proof of Proposition~\ref{prop:optimal-properties} is provided in Appendix~\ref{app:optimal-properties}. As a result of Claim (iii), we can write $\jW{x}$ to denote the distortion-tilted information density by dropping the subscript denoting the optimal distribution $P\in\calP^*(D)$, and the source dispersion $V(P)$ has the same value for any $P\in\calP^*(D)$, which can be written as
\begin{align}
V=\sum_{x\in\calX}P_X(x)\bigl(\jW{x}-R_{\mathrm{WZ}}(D)\bigr)^2.\label{eq:source-variance}
\end{align}
As a result of Claim (v), we have
\begin{align}
{{\underline{\widetilde V}}}:=\min_{P\in\calP^*(D)}{{{\widetilde V}}}(P)>0,\label{eq:uniform-direction}
\end{align}
which is used in our achievability proof in Section~\ref{sec:ach-code}.

Furthermore, define the set of covariance matrices
\begin{align}
\calK_D:=\{\Sigma(P):P\in\calP^*(D)\}.\label{eq:covariance-family-definition}
\end{align}
The following proposition leads to an alternative form of $F_{\mathrm{f}}(L)$ in~\eqref{eq:fixed-versus-adaptive}.
\begin{proposition}\label{prop:covariance-family}
For any real number every $L\in\bbR$,
\begin{align}
F_{\mathrm{f}}(L)=\sup_{\Sigma\in\calK_D}\bbE_{{\Lambda}}\Bigl[\sup_{t\in\bbR}
\Phi_\Sigma(L-{{\Lambda}}-\lambda^*t,t)\Bigr].\label{eq:fixed-covariance-representation}
\end{align}
\end{proposition}
The proof of Proposition~\ref{prop:covariance-family} is provided in Appendix~\ref{app:geometry}. The set $\calK_D$ is the image of the same optimal set $\calP^*(D)$ used in~\eqref{eq:fixed-versus-adaptive} while~\eqref{eq:fixed-covariance-representation} expresses that supremum in terms of covariance matrices. As shown below, Proposition~\ref{prop:covariance-family} allows us to compare the achievability result of~\cite[Theorem~15]{lili} with the exact second-order asymptotics in this paper.

\subsection{Main Results and Discussions}
\label{sec:main-result}
Recall the definitions of $V(P)$ and $\Phi_\Sigma$ in~\eqref{eq:existing-source-variance} and~\eqref{eq:gaussian-cdf}, respectively. By Claim~(iii) of Proposition~\ref{prop:optimal-properties}, $V(P)=V$ for every $P\in\calP^*(D)$. Let ${{\Lambda}}\sim {\calN}(0,V)$ be the Gaussian random variable with mean zero and variance $V$. For any real number $L\in\bbR$, define the following two quantities:
\begin{align}
F_*(L)&:=\bbE\Bigl[\sup_{(P,t)\in\calP^*(D)\times\bbR}
\Phi_{\Sigma(P)}(L-{{\Lambda}}-{\lambda^*} t,t)\Bigr],\label{eq:Fstar}\\
L_*(\veps)&:=\inf\{L\in\bbR:\;F_*(L)\geq1-\veps\}.\label{eq:quantile}
\end{align}
The parameter $t$ here has the same meaning as in~\eqref{eq:fixed-versus-adaptive}.

Our first result is the following second-order converse bound.
\begin{theorem}\label{thm:main}
Under Assumptions~(\ref{ass:finite})--(\ref{ass:interior}),
\begin{align}
L_{\mathrm{WZ}}^*(D,\veps)\geq L_*(\veps).\label{eq:main-lower}
\end{align}
\end{theorem}
The proof of Theorem~\ref{thm:main} is provided in Section~\ref{sec:converse} and proceeds in three steps. Firstly, we use posterior decomposition and an optimality-gap bound to identify the first-order optimal test channels and upper bound the non-excess-distortion probability by the probability of a joint rate and distortion event, together with a vanishing term. Secondly, we use concentration inequalities and a source-only code modification to replace the thresholds and conditional covariance in the above event by quantities depending on the source sequence. Thirdly, we use a conditional Gaussian approximation for martingales and average over the source distribution to show that $F_*(L)$ is the asymptotic upper bound of the non-excess-distortion probability.

Our second result is an achievability result.
\begin{theorem}\label{thm:achievability}
Under Assumptions~(\ref{ass:finite})--(\ref{ass:interior}),
\begin{align}
L_{\mathrm{WZ}}^*(D,\veps)\leq L_*(\veps).\label{eq:main-upper}
\end{align}
\end{theorem}
The proof of Theorem~\ref{thm:achievability} is provided in Section~\ref{sec:achievability} and proceeds in three steps. Firstly, we derive a finite-blocklength upper bound on the excess-distortion probability via the one-shot coding method (cf. Lemma~\ref{lem:oneshot}) by choosing an optimal test channel for each source type and perturbing its mean distortion. Subsequently, we use the conditional means and Gaussian approximation (cf. Lemma~\ref{lem:types}) to bound lower bound the non-excess-distortion probability by the probability of a joint rate and distortion event. Finally, we use finite approximations and average over source types to show that $F_*(L)$ is the asymptotic lower bound for the non-excess distortion probability, which yield the desired result of Theorem~\ref{thm:achievability} by using the quantile properties (cf. Lemma~\ref{lem:gaussian}).

Combining Theorems~\ref{thm:main} and~\ref{thm:achievability} yields the following exact second-order asymptotics for the discrete Wyner-Ziv problem.
\begin{corollary}\label{cor:secondorder}
Under Assumptions~(\ref{ass:finite})--(\ref{ass:interior}),
\begin{align}
L_{\mathrm{WZ}}^*(D,\veps)=L_*(\veps).\label{eq:main-coefficient}
\end{align}
\end{corollary}
The authors of~\cite{lili} conjectured that the achievability result in~\cite[Theorem~15]{lili} is second-order asymptotically optimal for the discrete Wyner-Ziv problem. In the following lemma, we show that the conjecture is generally not true and also provide sufficient conditions for the result in~\cite[Theorem~15]{lili} to be optimal for second-order asymptotics.

\begin{lemma}\label{lem:comparison}
Under Assumptions~(\ref{ass:finite})--(\ref{ass:interior}), for every $L\in\bbR$,
\begin{align}
F_{\mathrm{f}}(L)\leq F_*(L),\label{eq:comparison-bound}
\end{align}
with equality if and only if there exists a matrix $\Sigma_0\in\calK_D$ such that
\begin{align}
\sup_{t\in\bbR}\Phi_{\Sigma_0}(L-{{\Lambda}}-{\lambda^*}t,t)
=\sup_{(\Sigma,t)\in\calK_D\times\bbR}\Phi_\Sigma(L-{{\Lambda}}-{\lambda^*}t,t)
\quad\text{almost surely}.\label{eq:comparison-equality}
\end{align}
In particular, the equality holds under the following two conditions.
\begin{enumerate}
\item If $\calK_D$ is a singleton with the element $\Sigma^*$, i.e., $\Sigma(P):=\Sigma^*$ for any $P\in\calP^*(D)$,
\begin{align}
F_{\mathrm{f}}(L)=F_*(L)=\bbE_{{{\Lambda}}}\Bigl[\sup_{t\in\bbR}\Phi_{\Sigma^*}(L-{{\Lambda}}-{\lambda^*}t,t)\Bigr].\label{eq:common-covariance-case}
\end{align}
\item If the source dispersion equals zero, i.e., $V=0$,
\begin{align}
F_{\mathrm{f}}(L)=F_*(L)=\sup_{(\Sigma,t)\in\calK_D\times\bbR}\Phi_\Sigma(L-{\lambda^*}t,t).\label{eq:zero-source-variance-case}
\end{align}
\end{enumerate}
\end{lemma}
\begin{proof}
Fix any real number $L\in\bbR$. It follows from the definitions of $F_{\mathrm{f}}(L)$ and $F_*(L)$ in~\eqref{eq:fixed-versus-adaptive} and~\eqref{eq:Fstar}, and Proposition~\ref{prop:covariance-family} that
\begin{align}
F_{\mathrm{f}}(L)
&=\sup_{\Sigma\in\calK_D}\bbE_{{{\Lambda}}}\Bigl[\sup_{t\in\bbR}
\Phi_\Sigma(L-{{\Lambda}}-{\lambda^*}t,t)\Bigr],\label{eq:comparison-fixed-family}\\
&\leq\bbE_{{{\Lambda}}}\Bigl[\sup_{\Sigma\in\calK_D}\sup_{t\in\bbR}
\Phi_\Sigma(L-{{\Lambda}}-{\lambda^*}t,t)\Bigr],\label{eq:comparison-expectation}\\
&=F_*(L),\label{eq:comparison-full-family}
\end{align}
where~\eqref{eq:comparison-fixed-family} follows from~\eqref{eq:fixed-versus-adaptive} and Proposition~\ref{prop:covariance-family}, and~\eqref{eq:comparison-full-family} follows from~\eqref{eq:Fstar}.

To prove the equality condition, note that $\calK_D$ is compact since $\calP^*(D)$ is compact and $\Sigma(P)$ is continuous. For any ${{\ell}}\in\bbR$, the function $\Sigma\mapsto\sup_{t\in\bbR}\Phi_\Sigma({{\ell}}-\lambda^*t,t)$ is upper semi-continuous. For a maximizer $\Sigma_0$ of the covariance supremum in~\eqref{eq:fixed-covariance-representation}, lead to n any ${{\ell_{\rm s}}}\in\bbR$, define
\begin{align}
{{{\Theta}}}_L({{\ell_{\rm s}}}):=\sup_{(\Sigma,t)\in\calK_D\times\bbR}\Phi_\Sigma(L-{{\ell_{\rm s}}}-\lambda^*t,t)
-\sup_{t\in\bbR}\Phi_{\Sigma_0}(L-{{\ell_{\rm s}}}-\lambda^*t,t)\label{eq:comparison-difference}
\end{align}
It follows from the definition of $F_{\mathrm{f}}(L)$ and $F_*(L)$ in~\eqref{eq:fixed-versus-adaptive} and~\eqref{eq:Fstar} that
\begin{align}
\bbE_{{\Lambda}}[{{{\Theta}}}_L({{\Lambda}})]=F_*(L)-F_{\mathrm{f}}(L).\label{eq:comparison-difference-mean}
\end{align}
Thus, $F_*(L)=F_{\mathrm{f}}(L)$ when ${{{\Theta}}}_L({{\Lambda}})=0$ almost surely, which proves~\eqref{eq:comparison-equality}. Furthermore, when singleton $\calK_D$ is the singleton, the optimization over the covariance disappears in both $F_*(L)$ and $F_{\mathrm{f}}(L)$, which proves~\eqref{eq:common-covariance-case}; if $V=0$, then ${{\Lambda}}=0$ almost surely, which implies the result in~\eqref{eq:zero-source-variance-case}.
\end{proof}

It follows from Lemma~\ref{lem:comparison} and the definitions of $L_{\mathrm{f}}(\cdot)$ and $L_*(\cdot)$ in~\eqref{eq:fixed-quantile-definition} and~\eqref{eq:quantile} that for any $\veps\in(0,1)$ such that $L_{\mathrm{f}}(\veps)\geq0$,
\begin{align}
L_{\mathrm{f}}(\veps)\geq L_*(\veps),\label{eq:comparison-quantiles}
\end{align}
with equality under the same conditions of Lemma~\ref{lem:comparison}. Furthermore, we provide two numerical examples in Section~\ref{sec:numerical}, where the binary asymmetric example numerically demonstrates the second-order asymptotic optimality of the achievability bound in~\cite[Theorem~15]{lili}, while the other quaternary example shows that~\cite[Theorem~15]{lili} is strictly suboptimal, i.e., $L_{\mathrm{f}}(\veps)>L_*(\veps)$.

\subsection{Numerical Examples}
\label{sec:numerical}
In this section, we provide two numerical examples: a binary asymmetric source and a quaternary source. For the first example, our numerical evaluation shows that the achievability bound in~\cite[Theorem~15]{lili} is second-order asymptotically optimal. The numerical optimization identifies a singleton optimal covariance family, which is a sufficient condition for optimality by Lemma~\ref{lem:comparison}. For the second example, the set $\calK_D$ is not singleton and the achievability result in~\cite[Theorem~15]{lili} is strictly suboptimal compared to our exact second-order asymptotic result in Corollary~\ref{cor:secondorder}.

\begin{table}[!htbp]
\centering
\caption{Comparison of second-order asymptotics}\label{tab:nonunique}
\begin{tabular}{cccc}
\toprule
$\veps$ & \shortstack{Ach result~\cite[Theorem~15]{lili}\\$L_{\mathrm{f}}(\veps)$} & \shortstack{Exact Second-Order\\(Corollary~\ref{cor:secondorder}): $L_*(\veps)$} & \shortstack{Second-order rate reduction\\$\Delta L$ (percentage)}\\
\midrule
$0.05$ & $0.55$ & $0.51$ & $0.04$ ($6.78\%$)\\
$0.10$ & $0.43$ & $0.38$ & $0.05$ ($10.77\%$)\\
$0.20$ & $0.28$ & $0.22$ & $0.06$ ($21.35\%$)\\
\bottomrule
\end{tabular}
\end{table}

\subsubsection{Binary Asymmetric Source}
\label{sec:bac-example}
Let $\calX=\calY=\calZ=\{0,1\}$, $P_X=(0.65,0.35)$, and $d(x,z)=\bbI\{x\ne z\}$ for $(x,z)\in\calX\times\calZ$. Set $D=0.05$ and $\veps=0.1$. The side-information channel, with rows and columns indexed by $(x,y)\in\{0,1\}^2$, is
\begin{align}
P_{Y|X}=\begin{pmatrix}0.92&0.08\\0.20&0.80\end{pmatrix}.\label{eq:bac-channel}
\end{align}
It follows from the definitions of $R_{\mathrm{WZ}}(\cdot)$, $L_{\mathrm{f}}(\cdot)$, and $L_*(\cdot)$ from~\eqref{eq:fixed-alphabet},~\eqref{eq:fixed-quantile-definition}, and~\eqref{eq:quantile} that
\begin{align}
R_{\mathrm{WZ}}(0.05)&\approx0.21,\label{eq:bac-main-values}\\
L_{\mathrm{f}}(0.1)\approx L_*(0.1)&\approx1.69,
\end{align}
where calculation details are provided in Appendix~\ref{app:bac}.

\subsubsection{A Quaternary Source}
\label{sec:nonunique-example}
Let $X=({{X_{\rm p}}},{{X_{\rm b}}})\in\calX=\{0,1\}^2$, $Y={{X_{\rm b}}}\oplus {{N}}\in\calY=\{0,1\}$, and $Z=({{Z_{\rm p}}},{{Z_{\rm b}}})\in\calZ=\{0,1,*\}\times\{0,1\}$, where ${{X_{\rm p}}}\sim\Ber(r)$, ${{X_{\rm b}}}\sim\Ber(1/2)$, and ${{N}}\sim\Ber(p)$ are mutually independent Bernoulli random variables. Let $(p,r,{q_{\rm p}})=(2/5,1/3,1/10)$ be three real numbers.
For $x=({{x_{\rm p}}},{{x_{\rm b}}})\in\calX$, $y\in\calY$, and $z=({{z_{\rm p}}},{{z_{\rm b}}})\in\calZ$, define the matrix
\begin{align}
\bigl(d_{{\rm p}}({{x_{\rm p}}},{{z_{\rm p}}})\bigr)_{{{x_{\rm p}}},{{z_{\rm p}}}}:=\begin{pmatrix}
0&\log9&\log(27/20)\\
\log9&0&\log(27/10)
\end{pmatrix},\label{eq:prediction-cost-table}
\end{align}
and the distortion function
\begin{align}
d(x,z):=d_{{\rm p}}({{x_{\rm p}}},{{z_{\rm p}}})+{{w}}\bbI\{{{x_{\rm b}}}\ne{{z_{\rm b}}}\},\label{eq:example-distortion}
\end{align}
where
\begin{align}
{{w}}:=\frac{{{{\omega}}}_p(\alpha)}{p-\alpha},\label{eq:example-tangent}
\end{align}
and $\alpha\in(0,p)$ is the unique solution to ${{{\omega}}}_p(\alpha)+(p-\alpha){{{\omega}}}_p'(\alpha)=0$ with
\begin{align}
{{{\omega}}}_p(a):=h_2(p+(1-2p)a)-h_2(a),\ ,\ a\in(0,1).\label{eq:example-phi}
\end{align}
Appendix~\ref{app:four-symbol-proof} calculates that $\alpha\approx0.31$ and ${{w}}\approx0.79$, so $0\leq d(x,z)\leq\log9+{{w}}$. Set the distortion level
\begin{align}
D:=q_{\rm p}\log\frac{1-q_{\rm p}}{q_{\rm p}}+{w}p,=h_2({q_{\rm p}})+\log(1-{q_{\rm p}})+{{w}}p\approx0.54.\label{eq:example-threshold}
\end{align}
Appendix~\ref{app:four-symbol-model} verifies $D_{\min}=0<D<D_0$.

Appendices~\ref{app:four-symbol-model} and~\ref{app:four-symbol-proof} show that $R_{\mathrm{WZ}}(D)\approx0.31$ and calculate the second-order rates $(L_{\mathrm{f}}(\veps),L_*(\veps))$. The calculated results are summarized in Table~\ref{tab:nonunique}, where $\Delta L:=L_{\mathrm{f}}(\veps)-L_*(\veps)$. In particular, when $\veps=0.1$, our exact second-order coding rate is $10.77\%$ smaller than the achievability result of~\cite[Theorem~15]{lili} and when $\veps=0.2$, the reduction is more than $20\%$.

\FloatBarrier

\section{Converse Proof (Proof of Theorem~\ref{thm:main})}
\label{sec:converse}
Fix $D\in(D_{\min},D_0)$ and any sequence of $(n,M)$-codes per Def.~\ref{def:code}. We prove Theorem~\ref{thm:main} by upper bounding its non-excess-distortion probability $\Pcn(D)$ when
\begin{align}
\log M\leq nR_{\mathrm{WZ}}(D)+\sqrt nL+o(\sqrt n),\label{eq:rateconverse}
\end{align}
where
\begin{align}
\Pcn(D)&=1-\Pen(D),=\Pr\{{d^{(n)}}(X^n,Z^n)\leq D\}.\label{eq:correct-probability}
\end{align}
In subsequent analysis, for ease of notation, let
\begin{align}
{\calE_n}:=\{{d^{(n)}}(X^n,Z^n)\leq D\}\label{eq:success-event}
\end{align}
denote the non-excess-distortion event. In this section, we show that
\begin{align}
\limsup_{n\to\infty}\Pr\{{\calE_n}\}\leq F_*(L).\label{eq:main-converse}
\end{align}
Thus, if $L<L_*(\veps)$, the right side is below $1-\veps$, which implies the desired converse bound $L_{\mathrm{WZ}}^*(D,\veps)\geq L_*(\veps)$ in Theorem~\ref{thm:main}.

\subsection{First-Order Optimal Posterior Distributions}
\label{sec:preliminaries}
We first rewrite the first-order optimization in terms of posterior mixtures using Lemma~\ref{lem:posterior} in this subsection. The resulted supporting gap and contact distributions are used in Lemma~\ref{lem:localization} in the next subsection and the source-measurable covariance in Lemma~\ref{lem:selectors} in the next next subsection.

Recall from~\eqref{eq:correct-probability}--\eqref{eq:success-event} that the non-excess-distortion probability is $\Pcn(D)=\Pr\{{\calE_n}\}$. To bound $\Pcn(D)$, we first identify the optimal posteriors for the first-order rate $R_{\mathrm{WZ}}(D)$ (cf. Eq.~\eqref{eq:wzclassic}). Note that the source distribution is $P_X$. Fix any test channel $P_{{U}|X}\in\calP({\calU}|\calX)$. For each ${u}\in{\calU}$, define the following marginal distribution
\begin{align}
{P_{U}}({u}):=\sum_{x\in\calX}P_X(x)P_{{U}|X}({u}|x),\label{eq:auxiliary-marginal}
\end{align}
and for each ${u}\in{\calU}$ such that ${P_{U}}({u})>0$, let
\begin{align}
{Q_{u}}(x):=P_{X|{U}}(x|{u})=\frac{P_X(x)P_{{U}|X}({u}|x)}{{P_{U}}({u})}.\label{eq:auxiliary-posterior}
\end{align}
In Lemma~\ref{lem:posterior} below, ${Q_{u}}$ is a posterior vector optimized together with its weight and $P_{X|{U}}(\cdot|{u})$ is the conditional distribution induced by an auxiliary test channel.

Recall that $P_{Y|X}$ is the induced channel by the joint source and side information distribution $P_{XY}$. For any distribution $Q\in\calP(\calX)$ and $(x,y)\in\calX\times\calY$, define
\begin{align}
Q_{XY}(x,y)&:=Q(x)P_{Y|X}(y|x),\\
Q_Y(y)&:=\sum_{x\in\calX}Q(x)P_{Y|X}(y|x).\label{eq:posterior-output}
\end{align}
For $(x,y)\in\calX\times\calY$ with $Q_Y(y)>0$, define
\begin{align}
Q_{X|Y}(x|y):=\frac{Q(x)P_{Y|X}(y|x)}{Q_Y(y)}.\label{eq:posterior-conditional}
\end{align}
Furthermore, define the following two quantities:
\begin{align}
\rho(Q)&:=-H_Q(X|Y),\label{eq:posterior-entropy}\\
\delta(Q)&:=\sum_{y\in\calY}\min_{z\in\calZ}\sum_{x\in\calX} Q(x)P_{Y|X}(y|x)d(x,z),\label{eq:posterior-risk}
\end{align}
where $H_Q(X|Y)$ is the conditional entropy under $Q_{XY}$, and $\delta(Q)$ is the minimum expected distortion when the source has distribution $Q$ and the decoder observes $Y$. It follows from the definition of $D_0$ in~\eqref{eq:D0} that $\delta(P_X)=D_0$.

The following lemma gives an alternative form for $R_{\mathrm{WZ}}(D)$ in Eq.~\eqref{eq:wzclassic}.

\begin{lemma}\label{lem:posterior}
The first-order coding rate $R_{\mathrm{WZ}}(D)$ has the following representation
\begin{align}
R_{\mathrm{WZ}}(D)=\min_{\substack{{m_0}\in\bbN,\ (\btheta,({Q_{u}})_{{u}\in[{m_0}]})\in\bbR_+^{{m_0}}\times\calP(\calX)^{{m_0}}\\
\sum_{{u}\in[{m_0}]}{\theta_{u}}{Q_{u}}=P_X,\;\sum_{{u}\in[{m_0}]}{\theta_{u}}\delta({Q_{u}})\leq D}}
\Bigl\{\sum_{{u}\in[{m_0}]}{\theta_{u}}\rho({Q_{u}})-\rho(P_X)\Bigr\}.\label{eq:firstorder}
\end{align}
\end{lemma}
\begin{proof}
Fix any feasible $(P_{U|X},g)$ in~\eqref{eq:wzclassic}. Discard zero-probability auxiliary labels and relabel the remaining labels by $[m_0]$. Set $\theta_u=P_U(u)$ and use $Q_u$ in~\eqref{eq:auxiliary-posterior}. Marginalizing $U$ gives $\sum_{u\in[m_0]}\theta_uQ_u=P_X$. Given $U=u$, the joint distribution of $(X,Y)$ is $Q_u(x)P_{Y|X}(y|x)$. Thus, the minimizing reconstructions in~\eqref{eq:posterior-risk} yield $\sum_{u\in[m_0]}\theta_u\delta(Q_u)\leq\bbE[d(X,g(U,Y))]\leq D$. The objective satisfies
\begin{align}
I(X;{U}|Y)
&=H(X|Y)-\sum_{{u}\in[{m_0}]}{\theta_{u}}H(X|Y,{U}={u}),\label{eq:v20-posterior-1}\\
&=\sum_{{u}\in[{m_0}]}{\theta_{u}}\rho({Q_{u}})-\rho(P_X),\label{eq:v20-posterior-2}
\end{align}
where~\eqref{eq:v20-posterior-1} is the conditional mutual information identity, and~\eqref{eq:v20-posterior-2} follows from ${U}-X-Y$ and $\rho(Q)=-H_Q(X|Y)$.

Conversely, discard zero-weight components of a feasible mixture in~\eqref{eq:firstorder} and relabel the remaining components by $[m_0]$. Its mean constraint implies $\sum_{u\in[m_0]}\theta_u=1$. Since $P_X$ has full support, $P_{U|X}(u|x)=\theta_uQ_u(x)/P_X(x)$ defines a test channel whose row sums equal one by the same mean constraint. Under the joint distribution $P_XP_{U|X}P_{Y|X}$, we have $U-X-Y$, $P_U(u)=\theta_u$, and $P_{X|U}(\cdot|u)=Q_u$. Choosing the minimizing reconstruction for each $(u,y)$ in~\eqref{eq:posterior-risk} gives $\bbE[d(X,g(U,Y))]=\sum_{u\in[m_0]}\theta_u\delta(Q_u)\leq D$, with the objective in~\eqref{eq:v20-posterior-2}. Finally, Appendix~\ref{app:posterior-cardinality} reduces every finite mixture to at most $|\calX|+1$ positive weights while preserving its mean and distortion and without increasing its objective. Thus the unrestricted $m_0$ in~\eqref{eq:firstorder} gives the same minimum as the cardinality constraint in~\eqref{eq:wzclassic}.
\end{proof}

The mean constraint in~\eqref{eq:firstorder} implies that the weights sum to one, so this is a linear program over probability measures on the posterior simplex. Its dual is obtained by assigning multipliers to distortion, mean, and unit mass; Appendix~\ref{app:optimal-properties} gives the derivation and proves strong duality and attainment. In the following maximization, $\lambda$, $\bbeta=(\beta(x))_{x\in\calX}$, and $\beta_0$ are optimization variables:
\begin{align}
R_{\mathrm{WZ}}(D)
=\max_{\substack{\lambda\geq0,\ {\bbeta}\in\bbR^{|\calX|},\ {\beta_0}\in\bbR\\
{\beta_0}+\langle{\bbeta},Q\rangle\leq\rho(Q)+\lambda\delta(Q)\ \text{for all }Q\in\calP(\calX)}}
\{{\beta_0}+\langle{\bbeta},P_X\rangle-\rho(P_X)-\lambda D\}.\label{eq:posterior-dual}
\end{align}
Now, fix a maximizing triple $(\lambda^*,\bbeta,\beta_0)$ in~\eqref{eq:posterior-dual}. Appendix~\ref{app:optimal-properties} proves existence and identifies its distortion multiplier as $\lambda^*$~\eqref{eq:multiplier-definition}. Dual feasibility and equality of the optimal objectives imply that for every $Q\in\calP(\calX)$,
\begin{align}
\Delta(Q)&:=\rho(Q)+{\lambda^*}\delta(Q)-{\beta_0}-\langle {\bbeta},Q\rangle\geq0,\label{eq:deficit}\\
R_{\mathrm{WZ}}(D)+{\lambda^*} D&={\beta_0}+\langle {\bbeta},P_X\rangle-\rho(P_X).\label{eq:dualvalue}
\end{align}
The function $\Delta$ is the nonnegative Karush--Kuhn--Tucker (KKT) gap above the supporting affine function $\beta_0+\langle\bbeta,Q\rangle$, which vanishes at every posterior used with positive weight by an optimal mixture. We use the normalization $\langle\bbeta,P_X\rangle=0$, which is justified in Appendix~\ref{app:optimal-properties} that also proves uniqueness of the normalized maximizer.

Define $\calC$ as the zero set of the optimal gap. Appendix~\ref{app:optimal-properties} proves that it is nonempty and finite and that all its elements have full support. Enumerate these contact distributions as $Q_1,\ldots,Q_m$, i.e.,
\begin{align}
\calC&:=\{Q\in\calP(\calX):\;\Delta(Q)=0\}=\{Q_1,\ldots,Q_m\}.\label{eq:zero-gap-posteriors}
\end{align}
Define the associated set of contact-mixture weights by
\begin{align}
\calA:=\Bigl\{\btheta\in\bbR_+^m:
\sum\nolimits_{j\in[m]}\theta_jQ_j=P_X,\quad
\sum\nolimits_{j\in[m]}\theta_j\delta(Q_j)=D\Bigr\}.\label{eq:weightset}
\end{align}
The set $\calA$ is nonempty, compact, and convex. Its mean constraint implies $\sum_{j\in[m]}\theta_j=1$. After merging repeated posteriors and omitting zero weights, the optimal mixtures in Lemma~\ref{lem:posterior} are exactly those with weights in $\calA$: their positive-weight posteriors lie in $\calC$ and their distortion equals $D$, since $\lambda^*>0$. Conversely, every $\btheta\in\calA$ gives an optimal mixture. Appendix~\ref{app:optimal-properties} proves this equivalence. Thus, $\calC$ supplies the reference posteriors for localization in Lemma~\ref{lem:localization}, while $\calA$ parametrizes the contact mixtures and covariance choices, which will be used in Lemmas~\ref{lem:localization} and~\ref{lem:selectors} below.

For any positive integer $m_0\in\bbN$, nonnegative weights $(\vartheta_{u})_{u\in[m_0]}$, and posteriors $Q_{u}\in\calP(\calX)$ with $\sum_{{u}\in[{m_0}]}{\vartheta_{u}}{Q_{u}}=P_X$, we have
\begin{align}
\sum_{u\in[m_0]}\vartheta_u\Delta(Q_u)
&=\sum_{u\in[m_0]}\vartheta_u\rho(Q_u)
+\lambda^*\sum_{u\in[m_0]}\vartheta_u\delta(Q_u)
-\beta_0-\langle\bbeta,P_X\rangle\notag\\
&=\left[\sum_{u\in[m_0]}\vartheta_u\rho(Q_u)-\rho(P_X)-R_{\mathrm{WZ}}(D)\right]
+\lambda^*\left[\sum_{u\in[m_0]}\vartheta_u\delta(Q_u)-D\right],\label{eq:mixture-gap-identity}
\end{align}
where the second equality follows from~\eqref{eq:dualvalue}, and the first equality follows from~\eqref{eq:deficit} and uses $\sum_{u\in[m_0]}\vartheta_u=1$, which is obtained by summing the mean constraint over $x\in\calX$, and $\sum_{u\in[m_0]}\vartheta_uQ_u=P_X$.

If $\sum_{{u}\in[{m_0}]}{\vartheta_{u}}\delta({Q_{u}})\leq D$, this identity implies
\begin{align}
\sum_{{u}\in[{m_0}]}{\vartheta_{u}}\rho({Q_{u}})-\rho(P_X)-R_{\mathrm{WZ}}(D)
&\geq\sum_{{u}\in[{m_0}]}{\vartheta_{u}}\Delta({Q_{u}})\geq0.\label{eq:mixture-gap-rate-slack}
\end{align}
Equation~\eqref{eq:mixture-gap-rate-slack} shows that a positive average posterior gap increases the rate required by a distortion-feasible mixture. To use this fact for an arbitrary code, Lemma~\ref{lem:localization} applies the supporting gap to its sequential posteriors and combines it with the nonnegative reconstruction regret. The resulting threshold bound retains a code-dependent allocation. Lemma~\ref{lem:selectors} then replaces that allocation by a source-measurable one.

\subsection{From Non-Excess-Distortion to Posterior Localization}
\label{sec:localization}
We next apply the posterior gap $\Delta$ in~\eqref{eq:deficit} to an arbitrary $(n,M)$-code. Lemma~\ref{lem:localization} relates its non-excess-distortion probability to a joint rate and distortion event with a code-dependent allocation, and approximates the predictable covariance by a mixture of contact covariance matrices. Here, a contact covariance describes the conditional fluctuations of the centered information and distortion scores when the posterior is a contact distribution and the side information is drawn from $P_{Y|X}(\cdot|x)$ for a fixed source symbol $x$. We define these scores and matrices in~\eqref{eq:centered-scores}--\eqref{eq:contact-covariance}. Recall ${\calE_n}$ from~\eqref{eq:success-event}. We now define the scores and prove the localization estimates. For $Q\in\calP(\calX)$ and every $y\in\calY$, choose a Bayes rule $g_Q:\calY\to\calZ$ by
\begin{align}
g_Q(y)\in\argmin_{z\in\calZ}\sum_{x\in\calX}Q(x)P_{Y|X}(y|x)d(x,z),\label{eq:bayes-rule-local}
\end{align}
where tiles are resolved using a fixed ordering of $\calZ$. For the distributions in~\eqref{eq:zero-gap-posteriors} and any $(j,y)\in[m]\times\calY$, define
\begin{align}
g_j&:=g_{Q_j},\label{eq:contact-rule-output}\\
Q_{j,Y}(y)&:=\sum_{x\in\calX}Q_j(x)P_{Y|X}(y|x).
\end{align}
For any $(j,x,y)\in[m]\times\calX\times\calY$ such that $P_{Y|X}(y|x)>0$, define
\begin{align}
\imath_j(x,y)&:=\log\frac{Q_j(x)}{P_X(x)}+\log\frac{P_Y(y)}{Q_{j,Y}(y)},\label{eq:contact-information}\\
\imath_j^{\rm c}(x,y)&:=\imath_j(x,y)-\sum_{\substack{y'\in\calY:\ P_{Y|X}(y'|x)>0}}P_{Y|X}(y'|x)\imath_j(x,y'),\label{eq:centered-scores}\\
d_j^{\rm c}(x,y)&:=d(x,g_j(y))-\sum_{y'\in\calY}P_{Y|X}(y'|x)d(x,g_j(y')),\label{eq:centered-distortion-score}\\
{\Sigma_j}(x)&:=\bbE_{Y\sim P_{Y|X}(\cdot|x)}\left[\begin{pmatrix}\imath_j^{\rm c}(x,Y)\\d_j^{\rm c}(x,Y)\end{pmatrix}
\begin{pmatrix}\imath_j^{\rm c}(x,Y)&d_j^{\rm c}(x,Y)\end{pmatrix}\right].\label{eq:contact-covariance}
\end{align}
These centered scores describe the conditional information and distortion fluctuations. The matrix ${\Sigma_j}(x)$ is their conditional covariance that results from the randomness of side information $Y$ given a source symbol $x\in\calX$. For any nonnegative $\btheta=(\theta_1,\ldots,\theta_m)$ with $\sum_{j\in[m]} \theta_j=1$, define
\begin{align}
\Sigma(\btheta):=\sum_{j\in[m]} \theta_j\sum_{x\in\calX}Q_j(x){\Sigma_j}(x).\label{eq:mixture-covariance}
\end{align}
Fix a sequence of $(n,M)$-codes satisfying~\eqref{eq:rateconverse} under Assumptions~(\ref{ass:finite})--(\ref{ass:interior}).
Denote the source and observation sequences in their original coordinate order by $(X^n,Y^n)$. Set ${S_0}:=f_n(X^n)$ and $Z^{n}:=\phi_n({S_0},Y^n)$. In these original coordinate variables, the event in~\eqref{eq:success-event} is
\begin{align}
{\calE_n}&=\{{d^{(n)}}(X^n,Z^n)\leq D\},\\
\Pcn(D)&=\Pr\{{\calE_n}\},
\end{align}
where in~\eqref{eq:main-converse}, it suffices to consider a subsequence with $\Pcn(D)$ bounded away from zero; the bound is immediate on a subsequence where this probability tends to zero. All later conditioning and code modifications are transferred back to this original event.

Introduce a uniform random permutation ${{\Upsilon}}$, independently of the original pair $(X^n,Y^n)$, and, for each $i\in[n]$, define
\begin{align}
\tilX_i&:=X_{{{\Upsilon}}(i)},\label{eq:permutation-convention}\\
\tilY_i&:=Y_{{{\Upsilon}}(i)},\\
\tilZ_i&:=Z_{{{\Upsilon}}(i)}.
\end{align}
Let $P_n^{\rm o}$ denote the joint distribution of $(X^n,Y^n,S_0,Z^n,{\Upsilon},\tilX^n,\tilY^n,\tilZ^n)$ induced by the product source $P_{XY}^n$, the code $(f_n,\phi_n)$, and the independent uniform permutation ${\Upsilon}$. All subsequent code posteriors and probabilities initially refer to this distribution. The additive distortion is invariant under permutation, so ${\calE_n}=\{{d^{(n)}}(\tilX^n,\tilZ^n)\leq D\}$. Written in the original coordinates, this event is determined by $(X^n,Y^n)$ and the fixed code, and does not depend on ${\Upsilon}$.

For every permutation ${\upsilon}$ of $[n]$ and $(\tilx^n,\tily^n)\in\calX^n\times\calY^n$, independence of ${\Upsilon}$ and the original iid source gives the following formulas. The channel conditioning is well defined for every $\tilx^n$ because $P_X$ has full support.
\begin{align}
P^{\rm o}_{n,\tilX^n\tilY^n|{\Upsilon}}(\tilx^n,\tily^n|{\upsilon})
&=\prod_{i\in[n]}P_{XY}(\tilx_i,\tily_i),\label{eq:permuted-source-product}\\
P^{\rm o}_{n,\tilY^n|\tilX^n,{\Upsilon}}(\tily^n|\tilx^n,{\upsilon})
&=\prod_{i\in[n]}P_{Y|X}(\tily_i|\tilx_i).\label{eq:permuted-channel-product}
\end{align}
The original message is still $S_0=f_n(X^n)$. It is determined by $(\tilX^n,{\Upsilon})$, since $X_j=\tilX_{{\Upsilon}^{-1}(j)}$ for every $j\in[n]$.

Recall the definitions of $\Delta$ and $g_Q$ in~\eqref{eq:deficit} and~\eqref{eq:bayes-rule-local}, respectively. Empty index ranges denote empty tuples. For any source-only description $S$ determined by $(\tilX^n,{\Upsilon})$, initially ${S}={S_0}$, and each $i\in[n]$, define
\begin{align}
{Q_i^{\rm s}}&:={P^{\rm o}}_{{n},\tilX_i|{S},\tilX^{i-1},\tilY_{i+1}^n,{{\Upsilon}}},\label{eq:sequential-posterior}\\
H_n&:=-\log {P_n^{\rm o}}({S}|\tilY^n,{{\Upsilon}}),\label{eq:message-likelihood}\\
d_i&:=d(\tilX_i,\tilZ_i),\label{eq:coordinate-distortion}\\
d_i^*&:=d(\tilX_i,g_{{Q_i^{\rm s}}}(\tilY_i)),\label{eq:bayes-coordinate-distortion}\\
{{\bar\mu_i^{\mathrm d}}}&:={\lambda^*}\bbE_{{P_n^{\rm o}}}[d_i-d_i^*|{S},\tilX^{i-1},\tilY^n,{{\Upsilon}}],\label{eq:reconstruction-regret}\\
{\Delta_n}&:=\sum_{i\in[n]}[\Delta({Q_i^{\rm s}})+{{\bar\mu_i^{\mathrm d}}}],\label{eq:accumulated-gap}\\
{\Delta_n^{\rm p}}&:=\sum_{i\in[n]}\Delta({Q_i^{\rm s}}).\label{eq:posterior-gap-sum}
\end{align}
The posterior ${Q_i^{\rm s}}$ need not be a contact distribution. It follows from conditional Bayes optimality that ${{\bar\mu_i^{\mathrm d}}}\geq0$, and ${\Delta_n^{\rm p}}$ is the accumulated KKT gap. The message likelihood $H_n$ and coordinate distortions $d_i$ encode the original constraints; Appendix~\ref{app:localization} decomposes them using these posteriors and gaps.

Recall the definitions of $\jW{\cdot}$ and $V$ in~\eqref{eq:common-optimal-density} and~\eqref{eq:source-variance}, respectively; centering $\jW{X}$ by $R_{\mathrm{WZ}}(D)$ leads to mean zero. For ${Q}\in\calP(\calX)$, define the nearest contact by
\begin{align}
{{\Pi}_{\calC}}({Q})&:=Q_{j({Q})},\label{eq:nearest-contact}\\
j({Q})&:=\min\argmin_{j\in[m]}\|{Q}-Q_j\|.
\end{align}
In particular, ${{\Pi}_{\calC}}({Q_i^{\rm s}})$ is the selected nearest contact to ${Q_i^{\rm s}}$. The smallest minimizing index resolves ties, so the selected contact is a single-valued measurable function of the posterior. For each $j\in[m]$, define $\imath_{Q_j}^{\rm c}:=\imath_j^{\rm c}$, $d_{Q_j}^{\rm c}:=d_{j}^{\rm c}$, and ${\Sigma_{Q_j}}:={\Sigma_j}$ when indexing the scores and matrices by their contact posterior. Define the scalar scores and, for $j\in[m]$, the empirical weights by
\begin{align}
\Lambda_n&:=n^{-1/2}\sum_{i\in[n]} \big(\jW{X_i}-R_{\mathrm{WZ}}(D)\big),\label{eq:source-block-score}\\
J_n&:=n^{-1/2}\sum_{i\in[n]} \imath_{{{\Pi}_{\calC}}({Q_i^{\rm s}})}^{\rm c}(\tilX_i,\tilY_i),\label{eq:information-block-score}\\
G_n&:=n^{-1/2}\sum_{i\in[n]} d_{{{\Pi}_{\calC}}({Q_i^{\rm s}})}^{\rm c}(\tilX_i,\tilY_i),\label{eq:distortion-block-score}\\
{\tilde t_n}&:=\frac{nD-\sum_{i\in[n]} d_i^*}{\sqrt n}+G_n+\frac{{\Delta_n^{\rm p}}}{{\lambda^*}\sqrt n},\label{eq:code-allocation}\\
\theta_{n,j}&:=\frac1n\sum_{i\in[n]}\bbI\{{{\Pi}_{\calC}}({Q_i^{\rm s}})=Q_j\},\label{eq:empirical-weights}\\
\btheta_n&:=(\theta_{n,1},\ldots,\theta_{n,m}).
\end{align}
The source score $\Lambda_n$ is invariant under permutation: the sum in~\eqref{eq:source-block-score} also equals $n^{-1/2}\sum_{i\in[n]}(\jW{\tilX_i}-R_{\mathrm{WZ}}(D))$. Thus its definition in the original coordinates is the same one used in Section~\ref{sec:achievability}.
For the reverse revelation order, let ${{\calF}_n^{\rm 0}}:={{\calS}}(\tilX^n,{{\Upsilon}})$ and, for each $i\in[n]$, let ${{\calF}_{n,i}^{\rm r}}:={{\calS}}(\tilX^n,{{\Upsilon}},\tilY_{n-i+1}^n)$, with ${{\calF}_{n,0}^{\rm r}}={{\calF}_n^{\rm 0}}$.
For each $i\in[n]$, define the increment corresponding to the $i$th revelation in the order $\tilY_n,\ldots,\tilY_1$ by
\begin{align}
{{\boldsymbol{\zeta}}_{n,i}}:=
\begin{pmatrix}
\imath_{{{\Pi}_{\calC}}({Q_{n-i+1}^{\rm s}})}^{\rm c}(\tilX_{n-i+1},\tilY_{n-i+1})\\
d_{{{\Pi}_{\calC}}({Q_{n-i+1}^{\rm s}})}^{\rm c}(\tilX_{n-i+1},\tilY_{n-i+1})
\end{pmatrix}.\label{eq:reverse-channel-increment}
\end{align}
The centered scores in~\eqref{eq:centered-scores}--\eqref{eq:centered-distortion-score} imply that $\bbE_{{P_n^{\rm o}}}[{{\boldsymbol{\zeta}}_{n,i}}|{{\calF}_{n,i-1}^{\rm r}}]=\bzero_2$. Define its predictable covariance by
\begin{align}
\Omega_n:=\sum_{i\in[n]}{\Sigma}_{{{\Pi}_{\calC}}({Q_i^{\rm s}})}(\tilX_i).\label{eq:predictable-covariance}
\end{align}
Lemma~\ref{lem:localization} supplies the two threshold bounds and covariance approximation needed to construct source-measurable parameters.

For any real-valued allocation $t$, possibly random, and any positive real number $\eta\in\bbR_+$, define the compression rate and distortion events by
\begin{align}
\calR_n(t;\eta)&:=\{J_n\leq L-\Lambda_n-\lambda^*t+\eta\},\label{eq:rate-event}\\
\calD_n(t;\eta)&:=\{G_n\leq t+\eta\}.\label{eq:distortion-event}
\end{align}
The rate event in~\eqref{eq:rate-event} requires the centered information score to fit within the message budget after subtracting the source fluctuation and the allocation cost. The distortion event in~\eqref{eq:distortion-event} requires the centered distortion score to fit within the distortion allowance. The exact decompositions~\eqref{eq:exactH}--\eqref{eq:exactd} transfer $H_n\leq\log M+o(\sqrt n)$ and $\sum_{i\in[n]} d_i\leq nD$ to these two normalized constraints, with the residuals controlled by Lemma~\ref{lem:localization}.

\begin{lemma}\label{lem:localization}
For the preceding code sequence and every positive real number $\eta\in\bbR_+$, the following hold under the original distribution ${P_n^{\rm o}}$.
\begin{enumerate}
\item The positive violations of the rate and distortion constraints satisfy
\begin{align}
\Pr\{{\calE_n}\cap\calR_n({\tilde t_n};\eta)^c\}&\longrightarrow0,\label{eq:initial-information-violation}\\
\Pr\left\{{\calE_n}\cap\calD_n\left({\tilde t_n}-\frac{{\Delta_n}}{\lambda^*\sqrt n};\eta\right)^c\right\}&\longrightarrow0.\label{eq:initial-distortion-violation}
\end{align}
\item The predictable covariance satisfies
\begin{align}
\Pr\Bigl\{{\calE_n}\cap\bigl\{\bigl\|\Omega_n/n-\Sigma(\btheta_n)\bigr\|_2>\eta\bigr\}\Bigr\}\longrightarrow0.\label{eq:feasible}
\end{align}
\end{enumerate}
Both claims remain valid after a binary source-only flag is appended to the message.
\end{lemma}
The proof in Appendix~\ref{app:localization} uses the posterior representation in Lemma~\ref{lem:posterior} through the supporting gap and contact set in~\eqref{eq:deficit} and~\eqref{eq:zero-gap-posteriors}. It also uses the contact properties proved in Lemma~\ref{lem:geometry}: the common tilted-density identity and the likelihood chain lead to the exact decompositions~\eqref{eq:exactH}--\eqref{eq:exactd}, and the quadratic gap bound controls their residuals and the covariance approximation. The message-size bound controls $H_n$ except on an event of vanishing probability, while $\sum_{i\in[n]} d_i\leq nD$ holds on $\calE_n$. These facts yield Claim~(i); the posterior moment estimates and quadratic gap bound yield Claim~(ii). In particular, Claim~(i) and the union bound imply that
\begin{align}
\Pcn(D)
&\leq \Pr\left\{{\calE_n}\cap\calR_n({\tilde t_n};\eta)\cap
\calD_n\left({\tilde t_n}-\frac{{\Delta_n}}{\lambda^*\sqrt n};\eta\right)\right\}+o(1),\label{eq:localization-events}
\end{align}
where the union bound discards the two vanishing violation probabilities in Claim~(i) of Lemma~\ref{lem:localization}.
Since $\Delta_n\geq0$, the gap tightens the distortion constraint in~\eqref{eq:localization-events}. Claim~(ii) supplies the corresponding covariance approximation. The allocation $\tilde t_n$ and weights $\btheta_n$ can still depend on the side information; the next subsection uses Lemma~\ref{lem:selectors} to replace them by source-measurable parameters.

\subsection{From Posterior Localization to Source-Dependent Parameters}
\label{sec:selectors}
We now replace the code-dependent allocation and empirical weights by functions of $(\tilX^n,{{\Upsilon}})$. Define the following quantities:
\begin{align}
{B_n}&:=\bbI\{\Pr\{{\calE_n}|X^n\}\geq n^{-1}\},\label{eq:source-flag}\\
&=\bbI\{\Pr\{{\calE_n}|\tilX^n,{\Upsilon}\}\geq n^{-1}\},\notag\\
{S}&:=({S_0},{B_n}),\label{eq:flagged-description}\\
{P_n^{\rm b}}&:={P_n^{\rm o}}\{\,\cdot|{B_n}=1\}.\label{eq:flagged-law}
\end{align}
The above definitions restrict attention to sources with conditional success probability at least $n^{-1}$ and are called flagged versions in subsequent analyses. Equation~\eqref{eq:source-flag} holds because $(\tilX^n,{{\Upsilon}})$ determines $X^n$ and ${{\Upsilon}}$ is independent of the original experiment. Forming the flag before permutation preserves independence of ${{\Upsilon}}$ and $(X^n,Y^n)$ under ${P_n^{\rm b}}$. On the support of $P_n^{\rm b}$, the conditional channel in~\eqref{eq:permuted-channel-product} remains unchanged because the flag is determined by $(\tilX^n,{\Upsilon})$. The permuted source $\tilX^n$ and ${\Upsilon}$ need not remain independent, as allowed in Lemma~\ref{lem:concentration}.

From this point onward, ${Q_i^{\rm s}}$, ${\Delta_n}$, ${\tilde t_n}$, and the channel martingale are formed using the flagged description ${S}=({S_0},{B_n})$, before the source prefix is disclosed. The decoder may ignore the flag, so ${\calE_n}$ and the distortion vector remain those of the initial code. Lemma~\ref{lem:localization} applies with message alphabet size at most $2M$; the additional description costs $\log2=o(\sqrt n)$ nats. No comparison between the posteriors before and after adding the flag is required.

\begin{lemma}\label{lem:selectors}
There are functions $\bar{\btheta}_n(\tilX^n,{{\Upsilon}})\in\calA$ and $\bart_n(\tilX^n,{{\Upsilon}})\in\bbR$, measurable before the channel observations are revealed, such that, for every positive real number $\eta\in\bbR_+$,
\begin{enumerate}
\item The flagged description ${S}$ satisfies
\begin{align}
\Pr\{{\calE_n}\cap\calR_n(\bart_n;\eta)^c\}&\longrightarrow0,\label{eq:selector-eventwise}\\
\Pr\{{\calE_n}\cap\calD_n(\bart_n;\eta)^c\}&\longrightarrow0.\label{eq:selector-eventwise-distortion}
\end{align}
\item The predictable covariance satisfies
\begin{align}
\Pr\Bigl\{{\calE_n}\cap\bigl\{\bigl\|\Omega_n/n-\Sigma(\bar{\btheta}_n)\bigr\|_2>\eta\bigr\}\Bigr\}\longrightarrow0.\label{eq:cov-selector}
\end{align}
\end{enumerate}
\end{lemma}
The proof of Lemma \ref{lem:selectors} is provided in Appendix~\ref{app:selectors}, which uses the posterior representation in Lemma~\ref{lem:posterior} and the localization estimates in Lemma~\ref{lem:localization}. A source-only repair first yields an $o(n)$ expected gap under $P_n^{\rm b}$. Lemmas~\ref{lem:tc-deficit} and~\ref{lem:concentration} then construct source-measurable approximations to the allocation and empirical weights. Transferring them back to the flagged code and using Claim~(i) of Lemma~\ref{lem:localization} proves the two threshold bounds; Claim~(ii) supplies the covariance approximation. Thus, Lemma~\ref{lem:selectors} completes the source-measurable reduction begun by Lemma~\ref{lem:localization}.

For positive real numbers $(\tau,\eta)\in\bbR_+^2$, let $\Gamma_n:=\Sigma(\bar{\btheta}_n)$ and define the retained threshold and covariance events by
\begin{align}
{\calE_n^{\rm t}}({\eta})&:=\calR_n(\bart_n;{\eta})\cap\calD_n(\bart_n;{\eta}),\label{eq:roadmap-threshold-event}\\
{\calE_n^{\rm v}}&:=\{\|\Omega_n/n-\Gamma_n\|_2\leq\tau/2\}.\label{eq:roadmap-covariance-event}
\end{align}
We now apply these selectors to the localized success event in~\eqref{eq:localization-events}. Both claims of Lemma~\ref{lem:selectors} imply that the part of $\calE_n$ outside $\calE_n^{\rm v}\cap\calE_n^{\rm t}(\eta)$ has vanishing probability. Thus,~\eqref{eq:localization-events} and the union bound lead to
\begin{align}
\Pcn(D)\leq \Pr\{{\calE_n}\cap{\calE_n^{\rm v}}\cap{\calE_n^{\rm t}}({\eta})\}+o(1),\label{eq:selector-stage-bound}
\end{align}
where Claim~(i) of Lemma~\ref{lem:selectors} controls the two threshold violations and Claim~(ii) controls the covariance violation. This is the output of the source-dependent parameter reduction: all thresholds and the covariance target are measurable with respect to ${\calS}(\tilX^n,{\Upsilon})$. The next subsection uses precisely these properties for a conditional Gaussian comparison.

\subsection{Gaussian Bound on the Non-Excess-Distortion Probability}
\label{sec:gaussian-bound}
We finally bound the probability term in~\eqref{eq:selector-stage-bound} by a conditional Gaussian comparison. We then optimize its source-measurable covariance and allocation, average the source fluctuation, and remove smoothing to obtain~\eqref{eq:main-converse}. Recall the definitions of $\calP^*(D)$, $\calA$, and $\Sigma(\btheta)$ in~\eqref{eq:optimal-set},~\eqref{eq:weightset}, and~\eqref{eq:mixture-covariance}, respectively. Appendix~\ref{app:geometry} shows
\begin{align}
\{\Sigma(P):P\in\calP^*(D)\}
=\{\Sigma(\btheta):\btheta\in\calA\}.\label{eq:covariance-equivalence}
\end{align}
In particular, $\Gamma_n=\Sigma(\bar{\btheta}_n)$ belongs to $\calK_D$ in~\eqref{eq:covariance-family-definition}. Recall that $\Lambda\sim\calN(0,V)$, with $V$ defined in~\eqref{eq:source-variance}. For $\btheta\in\calA$, let $(J_{\btheta},G_{\btheta})\sim {\calN}(\bzero_2,\Sigma(\btheta))$ independently of ${\Lambda}$. For any real number $\ell\in\bbR$ and positive real numbers $(\tau,{\eta})\in\bbR_+^2$, define
\begin{align}
{F^{\rm c}_{\btheta}}(\ell)&:=\sup_{t\in\bbR}\Phi_{\Sigma(\btheta)}(\ell-{\lambda^*} t,t),\label{eq:gaussian-envelopes}\\
{F^{\rm c}}(\ell)&:=\sup_{\btheta\in\calA}{F^{\rm c}_{\btheta}}(\ell),\label{eq:full-pointwise-envelope}\\
{F^{\rm c}_{\tau,\eta}}(\ell)&:=\sup_{(\btheta,t)\in\calA\times\bbR}
\Phi_{\Sigma(\btheta)+\tau {\mathrm{Id}_2}}(\ell-{\lambda^*} t+2{\eta},t+2{\eta}).\label{eq:smoothed-envelope}
\end{align}
The covariance identity~\eqref{eq:covariance-equivalence} leads to $F_*(L)=\bbE[{F^{\rm c}}(L-{\Lambda})]$. The envelope in~\eqref{eq:smoothed-envelope} bounds the Gaussian probability after optimizing its source-measurable covariance and allocation.

Let a two-dimensional martingale have $n$ uniformly bounded centered increments ${{\boldsymbol{\zeta}}_{n,i}}$ adapted to $({{\calF}_{n,i}^{\rm r}})_{i=0}^n$, with ${{\calF}_{n,0}^{\rm r}}={{\calF}_n^{\rm 0}}$. Define
\begin{align}
\bXi_n&:=n^{-1/2}\sum_{i\in[n]}{{\boldsymbol{\zeta}}_{n,i}},\label{eq:generic-martingale-quantities}\\
\Omega_n&:=\sum_{i\in[n]}\bbE[{{\boldsymbol{\zeta}}_{n,i}}{{\boldsymbol{\zeta}}_{n,i}}^\top|{{\calF}_{n,i-1}^{\rm r}}].
\end{align}
Let $\Gamma_n$ be a uniformly bounded positive semidefinite matrix measurable with respect to ${{\calF}_n^{\rm 0}}$, and let ${\ell_{n,1}},{\ell_{n,2}}$ be real ${{\calF}_n^{\rm 0}}$-measurable thresholds. The events ${\calE_n}$ are as above, and all probability orders in Lemma~\ref{lem:completion} are under ${P_n^{\rm o}}$.

\begin{lemma}\label{lem:completion}
Suppose that, for every positive real number $\eta\in\bbR_+$, the covariance error satisfies
\begin{align}
\Pr\Bigl\{{\calE_n}\cap\bigl\{\bigl\|\Omega_n/n-\Gamma_n\bigr\|_2>\eta\bigr\}\Bigr\}&\longrightarrow0,\label{eq:completion-covariance-hypothesis}
\end{align}
and, for each ${j_2}\in\{1,2\}$, the positive threshold violation satisfies
\begin{align}
\Pr\Bigl\{{\calE_n}\cap\bigl\{\Xi_{n,{j_2}}-{\ell_{n,j_2}}>\eta\bigr\}\Bigr\}&\longrightarrow0.\label{eq:completion-threshold-hypothesis}
\end{align}
Then, for every pair of positive real numbers $(\tau,{\eta})\in\bbR_+^2$,
\begin{align}
\Pr\{{\calE_n}\}\leq
\bbE[\Phi_{\Gamma_n+\tau {\mathrm{Id}_2}}({\ell_{n,1}}+2{\eta},{\ell_{n,2}}+2{\eta})]
+3\tau/{\eta}^2+o(1).\label{eq:completionbound}
\end{align}
\end{lemma}
The proof of Lemma~\ref{lem:completion} is provided in Appendix~\ref{app:completion}. For its application, take ${{\calF}_n^{\rm 0}}={{\calS}}(\tilX^n,{{\Upsilon}})$ and reveal $\tilY_n,\ldots,\tilY_1$. For every $i\in[n]$,
\begin{align}
\bbE_{{P_n^{\rm o}}}\left[\left.
\begin{pmatrix}\imath_{{{\Pi}_{\calC}}({Q_i^{\rm s}})}^{\rm c}(\tilX_i,\tilY_i)\\d_{{{\Pi}_{\calC}}({Q_i^{\rm s}})}^{\rm c}(\tilX_i,\tilY_i)\end{pmatrix}
\right|\tilX^n,{{\Upsilon}},\tilY_{i+1}^n\right]=\begin{pmatrix}0\\0\end{pmatrix},\label{eq:gaussian-martingale-check}
\end{align}
where~\eqref{eq:gaussian-martingale-check} follows from~\eqref{eq:centered-scores}--\eqref{eq:centered-distortion-score}: ${{\Pi}_{\calC}}({Q_i^{\rm s}})$ is measurable given $(\tilX^n,{{\Upsilon}},\tilY_{i+1}^n)$ and $\tilY_i$ has conditional distribution $P_{Y|X}(\cdot|\tilX_i)$ under ${P_n^{\rm o}}$. The flagged message is source-measurable. The increments are uniformly bounded because the contact set is finite, its posteriors have full support, and distortion is bounded. Their normalized terminal value is ${\bXi_n=(J_n,G_n)^\top}$, and their predictable covariance is $\Omega_n$.

For this application, the covariance target and thresholds are
\begin{align}
\Gamma_n&=\Sigma(\bar{\btheta}_n),\label{eq:gaussian-target-thresholds}\\
({\ell_{n,1}},{\ell_{n,2}})&:=(L-\Lambda_n-\lambda^*\bart_n,\bart_n).
\end{align}
Lemma~\ref{lem:selectors} supplies exactly the hypotheses of Lemma~\ref{lem:completion}: its covariance claim is~\eqref{eq:completion-covariance-hypothesis} for the target in~\eqref{eq:gaussian-target-thresholds}, and its two threshold claims are~\eqref{eq:completion-threshold-hypothesis} for $(\Xi_{n,1},\Xi_{n,2})=(J_n,G_n)$. In contrast to Lemma~\ref{lem:selectors}, Lemma~\ref{lem:completion} is a generic martingale comparison; its proof uses these hypotheses and uniform increment bounds, rather than posterior geometry.

Starting from~\eqref{eq:selector-stage-bound}, we apply Lemma~\ref{lem:completion} with~\eqref{eq:gaussian-target-thresholds}. It replaces the retained channel event by its Gaussian rectangle probability. Optimizing that rectangle over the admissible weights and allocation leads to
\begin{align}
\Pcn(D)
&\leq\bbE\left[\Phi_{\Sigma(\bar{\btheta}_n)+\tau \mathrm{Id}_2}
(L-\Lambda_n-\lambda^*\bart_n+2\eta,\bart_n+2\eta)\right]
+3\tau/\eta^2+o(1),\label{eq:roadmap-gaussian-loss}\\
&\leq\bbE[F^{\rm c}_{\tau,\eta}(L-\Lambda_n)]+3\tau/\eta^2+o(1).
\label{eq:roadmap-envelope-loss}
\end{align}
where~\eqref{eq:roadmap-gaussian-loss} follows from Lemma~\ref{lem:completion}, whose hypotheses were just verified, and~\eqref{eq:roadmap-envelope-loss} follows from~\eqref{eq:smoothed-envelope} because $\bar{\btheta}_n\in\calA$. No additional localization is performed here:~\eqref{eq:selector-stage-bound} already completed that step. It remains to average the source fluctuation and remove the Gaussian smoothing.

To average $\Lambda_n$ and remove smoothing, we use Lemma~\ref{lem:gaussian}. Recall the definition of the compact set $\calA$ in~\eqref{eq:weightset}.

\begin{lemma}\label{lem:gaussian}
Under Assumptions~(\ref{ass:finite})--(\ref{ass:interior}), for any positive real numbers $(\tau,{\eta})\in\bbR_+^2$, $\ell\in\bbR$, $\btheta\in\calA$, and $t\in\bbR$,
\begin{align}
\Phi_{\Sigma(\btheta)+\tau {\mathrm{Id}_2}}(\ell-\lambda^*t+2{\eta},t+2{\eta})
\leq {F^{\rm c}_{\tau,\eta}}(\ell).\label{eq:pointwise-envelope-bound}
\end{align}
Furthermore, the following hold.
\begin{enumerate}
\item For any positive real numbers $(\tau,{\eta})\in\bbR_+^2$, ${F^{\rm c}_{\tau,\eta}}$ is bounded and continuous. For every $\ell\in\bbR$,
\begin{align}
\limsup_{\tau\to0}{F^{\rm c}_{\tau,\tau^{1/4}}}(\ell)\leq {F^{\rm c}}(\ell).\label{eq:smoothing-removal}
\end{align}
\item For every $\veps\in(0,1)$, the quantile $L_*(\veps)$ in~\eqref{eq:quantile} satisfies
\begin{align}
L<L_*(\veps)&\ \Longrightarrow\ F_*(L)<1-\veps,\label{eq:quantile-separation}\\
L>L_*(\veps)&\ \Longrightarrow\ F_*(L)>1-\veps.
\end{align}
\end{enumerate}
\end{lemma}
The proof of Lemma~\ref{lem:gaussian} is provided in Appendix~\ref{app:gaussian}. Claim~(i) of Lemma~\ref{lem:gaussian} permits source averaging and removes smoothing from~\eqref{eq:roadmap-envelope-loss} with ${\eta}=\tau^{1/4}$. Claim~(ii) of Lemma~\ref{lem:gaussian} leads to the strict quantile separation used to conclude the converse below and to invert the achievability bound in Section~\ref{sec:ach-limit-passage}.

Taking the limit in~\eqref{eq:roadmap-envelope-loss}, we have
\begin{align}
\limsup_{n\to\infty}\Pcn(D)
&\leq\limsup_{n\to\infty}\bbE[{F^{\rm c}_{\tau,\eta}}(L-\Lambda_n)]+3\tau/{\eta}^2,\label{eq:smoothing}\\
&=\bbE[{F^{\rm c}_{\tau,\eta}}(L-{\Lambda})]+3\tau/{\eta}^2,\label{eq:v20-smoothing-2}
\end{align}
where~\eqref{eq:smoothing} follows from~\eqref{eq:roadmap-envelope-loss}, and~\eqref{eq:v20-smoothing-2} follows from the scalar central limit theorem $\Lambda_n\Rightarrow {\Lambda}$ and the bounded continuity in Claim~(i) of Lemma~\ref{lem:gaussian}. Set ${\eta}=\tau^{1/4}$, so the extra term $3\tau/{\eta}^2=3\sqrt\tau$ tends to zero. Then
\begin{align}
\limsup_{n\to\infty}\Pcn(D)
&\leq\bbE\Bigl[\limsup_{\tau\to0}{F^{\rm c}_{\tau,\tau^{1/4}}}(L-{\Lambda})\Bigr],\label{eq:converse}\\
&\leq\bbE[{F^{\rm c}}(L-{\Lambda})],\label{eq:converse-final-bound}\\
&=F_*(L),
\end{align}
where~\eqref{eq:converse} follows from~\eqref{eq:smoothing} and bounded reverse Fatou, while~\eqref{eq:converse-final-bound} follows from~\eqref{eq:smoothing-removal} and~\eqref{eq:Fstar}. Equation~\eqref{eq:converse-final-bound} is the bound required in~\eqref{eq:main-converse}. By Claim~(ii) of Lemma~\ref{lem:gaussian}, $L<L_*(\veps)$ implies $F_*(L)<1-\veps$, which excludes a sequence with normalized rate offset at most $L$ and limiting excess-distortion probability at most $\veps$. Combining~\eqref{eq:converse-final-bound} with the quantile separation~\eqref{eq:quantile-separation}, we obtain $L_{\mathrm{WZ}}^*(D,\veps)\geq L_*(\veps)$ and complete the proof of Theorem~\ref{thm:main}.

\section{Achievability Proof}
\label{sec:achievability}
Throughout this section, Assumptions~(\ref{ass:finite})--(\ref{ass:interior}) hold. Take any real number $L\in\bbR$. Here ${P_n^{\rm o}}$ denotes the probability distribution of the current construction, with $(X^n,Y^n)$ marginal $P_{XY}^n$.

\subsection{Type-Dependent Code Construction}
\label{sec:ach-code}

For the coding step, let ${\calU_{\rm a}}^n$ denote a finite auxiliary sequence alphabet. For $(x^n,y^n,u^n)\in\calX^n\times\calY^n\times{\calU_{\rm a}}^n$, let ${P_n^{\rm a}}(u^n|x^n)$ be a conditional distribution and let ${Q_n^{\rm a}}(u^n|y^n)>0$ be a decoder reference distribution. Let $z^n:{\calU_{\rm a}}^n\times\calY^n\to\calZ^n$ be a reconstruction map, with coordinates $z_i$, $i\in[n]$. Under ${P_n^{\rm o}}=P_{X^nY^n}{P_n^{\rm a}}$, define
\begin{align}
\imath_n(X^n;U^n|Y^n)&:=\log\frac{{P_n^{\rm a}}(U^n|X^n)}{{Q_n^{\rm a}}(U^n|Y^n)},\label{eq:block-aux-information}\\
D_n&:=\sum_{i\in[n]} d(X_i,z_i(U^n,Y^n)).\label{eq:block-aux-distortion}
\end{align}
\begin{lemma}\label{lem:oneshot}
For every positive integer $M\in\bbN$ and every positive real number $\gamma\in\bbR_+$, there exists a deterministic $(n,M)$-code such that
\begin{align}
\Pen(D)
&\leq1-\bbE\left[\frac{\bbI\{D_n\leq nD\}}{1+\exp\{\imath_n(X^n;U^n|Y^n)\}/M}\right],\label{eq:oneshot}\\
&\leq\Pr\{D_n>nD\ \mathrm{or}\ \imath_n(X^n;U^n|Y^n)>\log M-\gamma\}+\exp\{-\gamma\}.
\end{align}
\end{lemma}
The proof of Lemma \ref{lem:oneshot} is provided in Appendix~\ref{app:oneshot}, which uses a finite exponential race based on the matching method of Li and Anantharam~\cite{poisson}. Lemma~\ref{lem:oneshot} with $\gamma=2\log n$ leads to~\eqref{eq:ach-error-reduction} after averaging over source types. We now specify its auxiliary distribution, decoder reference, and reconstruction. The reference may differ from the actual distribution of $U^n$ lead to n $Y^n$.

Recall the definitions of $\calC$ and $g_j$ in~\eqref{eq:zero-gap-posteriors} and~\eqref{eq:contact-rule-output}, respectively. For weights $\btheta\in\calA$ from~\eqref{eq:weightset} and every $(x,y,j)\in\calX\times\calY\times[m]$, the optimal distribution and its induced conditionals are determined by
\begin{align}
P_{\btheta}(x,y,j):=\theta_jQ_j(x)P_{Y|X}(y|x),\label{eq:optimal-joint}
\end{align}
The defining constraints of $\calA$ in~\eqref{eq:weightset} lead to
\begin{align}
\sum_{j\in[m]}\sum_{y\in\calY}P_{\btheta}(x,y,j)
&=\sum_{j\in[m]}\theta_jQ_j(x)=P_X(x),\label{eq:ach-optimal-source-marginal}\\
\bbE_{P_{\btheta}}[d(X,g_U(Y))]
&=\sum_{j\in[m]}\theta_j\delta(Q_j)=D,\label{eq:ach-optimal-distortion}
\end{align}
where~\eqref{eq:ach-optimal-source-marginal} uses the unit row sums of $P_{Y|X}$, and~\eqref{eq:ach-optimal-distortion} uses the Bayes rule $g_j$ at $Q_j$ from~\eqref{eq:contact-rule-output}. Thus, $P_{\btheta}$ is a joint distribution with the prescribed source marginal and expected distortion. Its optimality follows from the zero gaps of the contacts in~\eqref{eq:zero-gap-posteriors} and the identity~\eqref{eq:mixture-gap-identity}.
The induced conditional probabilities satisfy
\begin{align}
P_{\btheta}(j|x)&=\frac{\theta_jQ_j(x)}{P_X(x)},\label{eq:optimal-conditionals}\\
P_{\btheta}(j|y)&=\frac{\theta_jQ_{j,Y}(y)}{P_Y(y)},
\end{align}
where both identities follow from the joint distribution in~\eqref{eq:optimal-joint} and Bayes' rule.
Retain the contact-label alphabet $\calU_{\btheta}:=\{j\in[m]:\theta_j>0\}$ and use reconstruction $g_j$ for label $j$. The induced conditional probabilities are positive on $\calU_{\btheta}$, and the corresponding covariance is $\Sigma(\btheta)$ in~\eqref{eq:mixture-covariance}. When evaluating $\widetilde V(P_{\btheta})$, each retained label is identified with the index of its Bayes rule in $\widetilde{\calU}$, as justified in Appendix~\ref{app:geometry}.

Recall the definitions of $\bard_P$ and ${{\widetilde V}(P)}$ in~\eqref{eq:mean-distortion-source} and~\eqref{eq:existing-direction}, respectively. For $x\in\calX$ and retained labels $j$, define
\begin{align}
\bard_{\btheta}(x)&:=\bbE_{P_{\btheta}}[d(X,{g_U(Y)})|X=x],\label{eq:weight-mean-distortion}\\
{{\widetilde V}_{\btheta}}&:={{\widetilde V}(P_{\btheta})},,\label{eq:weight-direction-variance}\\
{\dot P_{\btheta}}(j|x)&:=\frac{P_{\btheta}(j|x)(\bard(x,j)-\bard_{\btheta}(x))}{{{\widetilde V}_{\btheta}}},\label{eq:distortion-direction}
\end{align}
where $\bard(x,j)=\sum_{y\in\calY}P_{Y|X}(y|x)d(x,g_j(y))$, consistent with~\eqref{eq:mean-distortion-rule}. The inequality ${{\widetilde V}_{\btheta}}\geq{\underline{\widetilde V}}>0$ follows from Claim~(v) of Proposition~\ref{prop:optimal-properties} and~\eqref{eq:uniform-direction}, so the denominator in~\eqref{eq:distortion-direction} stays uniformly away from zero. Appendix~\ref{app:types} proves that the direction ${\dot P_{\btheta}}$ preserves each row sum and has unit distortion derivative. Let $\calP_n(\calX)$ be the set of length-$n$ source types. For $\hatP\in\calP_n(\calX)$, $x^n$ of type $\hatP$, and every $x\in\calX$, define
\begin{align}
\hatT_{x^n}(x)&:=n^{-1}\sum_{i\in[n]}\bbI\{x_i=x\},\label{eq:empirical-source-type}\\
\bxi&:=\sqrt n(\hatP-P_X),\label{eq:type-fluctuation}\\
\Lambda_n(x^n)&=\langle\bxi,\jW{\cdot}-R_{\mathrm{WZ}}(D)\rangle,\label{eq:type-source-score}
\end{align}
where~\eqref{eq:type-source-score} is the type representation of the source score in~\eqref{eq:source-block-score}. The retained alphabet $\calU_{\btheta}$ was defined above. We use the standard definitions of conditional types and their classes~\cite[Sec.~1.4]{TanBook}. lead to n $\hatP\in\calP_n(\calX)$, define
\begin{align}
\calP_n(\calU_{\btheta}|\hatP)
:=\bigl\{Q_{U|X}\in\calP(\calU_{\btheta}|\calX):\;
n\hatP(x)Q_{U|X}(j|x)\in\{0,1,\ldots,n\hatP(x)\}
\text{ for all }(x,j)\in\calX\times\calU_{\btheta}\bigr\}.
\label{eq:conditional-type-set}
\end{align}
Thus, a conditional type specifies integer joint counts whose sum in each source row is $n\hatP(x)$. For an allocation $t\in\bbR$ and $(x,j)\in\calX\times\calU_{\btheta}$, define the target conditional distribution by
\begin{align}
P_{\btheta,n}(j|x)
:=P_{\btheta}(j|x)
+\frac{-t-\langle\bxi,\bard_{\btheta}\rangle}{\sqrt n}
\dot P_{\btheta}(j|x).
\label{eq:perturb}
\end{align}
The row-sum and unit-distortion-derivative identities~\eqref{eq:v20-direction-1}--\eqref{eq:v20-direction-2} explain this choice: the first preserves normalization, while the second cancels the source-type distortion shift and imposes the allocation $t$. For each fixed $\btheta$, bounded $t$, and $\|\bxi\|\leq n^{1/8}$, the change in every conditional probability is $O(n^{-3/8})$. Since all retained probabilities $P_{\btheta}(j|x)$ are positive, $P_{\btheta,n}$ remains a positive conditional distribution for all sufficiently large $n$.

Let $n_x:=n\hatP(x)$. Round each $n_xP_{\btheta,n}(j|x)$ down to an integer and add one to the required number of entries to make the row sum $n_x$. This constructs $\hatP_n\in\calP_n(\calU_{\btheta}|\hatP)$ satisfying the following bound for every $(x,j)\in\calX\times\calU_{\btheta}$:
\begin{align}
\left|\hatP(x)\hatP_n(j|x)-\hatP(x)P_{\btheta,n}(j|x)\right|
\leq\frac1n.
\label{eq:conditional-type-rounding-error}
\end{align}
For a pair $(x^n,u^n)\in\calX^n\times\calU_{\btheta}^n$, its joint type is
\begin{align}
\hatT_{x^n,u^n}(x,j)
:=\frac1n\sum_{i\in[n]}\bbI\{(x_i,u_i)=(x,j)\}.
\label{eq:empirical-joint-type}
\end{align}
For $x^n$ of source type $\hatP$, the conditional type class is
\begin{align}
\calT_{\hatP_n}(x^n)
:=\{u^n\in\calU_{\btheta}^n:\;\hatT_{x^n,u^n}=\hatP\hatP_n\}.
\label{eq:conditional-type-class}
\end{align}
Here $\hatP\hatP_n$ denotes the joint distribution with entries $\hatP(x)\hatP_n(j|x)$. The integer row counts ensure that this class is nonempty: assign the prescribed number of labels to the $n_x$ positions containing each source symbol $x$. The rounding error in~\eqref{eq:conditional-type-rounding-error} changes the total mean distortion by $O(1)$; Appendix~\ref{app:types} shows that its effect on the information mean is also $o(\sqrt n)$. Thus, rounding does not change the second-order means in Lemma~\ref{lem:types}.
Let ${P_n^{\rm a}}(\cdot|x^n)$ be uniform on this set, choose ${Q_n^{\rm a}}(u^n|y^n)=\prod_{i=1}^nP_{\btheta}(u_i|y_i)$, and reconstruct with $z_i=g_{u_i}(y_i)$. These are the three inputs required by Lemma~\ref{lem:oneshot}; Lemma~\ref{lem:types} states their conditional means and Gaussian fluctuations.

\begin{lemma}\label{lem:types}
Fix $\btheta\in\calA$ and a compact interval $\calI\subset\bbR$, retaining only positive-weight labels. For all sufficiently large $n$, the construction is well defined when
\begin{align}
t&\in\calI,\label{eq:type-admissibility}\\
\hatT_{x^n}&=\hatP,\\
\|\sqrt n(\hatP-P_X)\|&\leq n^{1/8},\label{eq:type-cutoff}\\
u^n&\in\calT_{\hatP_n}(x^n).\label{eq:type-admissibility-end}
\end{align}
Recall the definitions of $\imath_n$ and $D_n$ in~\eqref{eq:block-aux-information} and~\eqref{eq:block-aux-distortion}, respectively; their arguments $(X^n,U^n,Y^n)$ are suppressed below. For each pair $(x^n,u^n)\in\calX^n\times{\calU_{\btheta}}^n$, all expectations and convergence statements are with respect to $Y^n\sim P_{Y|X}^n(\cdot|x^n)$. The following hold uniformly over~\eqref{eq:type-admissibility}--\eqref{eq:type-admissibility-end}.
\begin{enumerate}
\item Its conditional means satisfy
\begin{align}
\frac{{\bbE_{P_n^{\rm o}}}[\imath_n|X^n,U^n]-nR_{\mathrm{WZ}}(D)}{\sqrt n}
&=\Lambda_n+{\lambda^*} t+o(1),\label{eq:type-means}\\
\frac{{\bbE_{P_n^{\rm o}}}[D_n|X^n,U^n]-nD}{\sqrt n}
&=-t+o(1).\label{eq:type-distortion-mean}
\end{align}
\item Conditionally on $(X^n,U^n)$, the centered vector satisfies
\begin{align}
\frac1{\sqrt n}
\begin{pmatrix}
\imath_n-{\bbE_{P_n^{\rm o}}}[\imath_n|X^n,U^n]\\
D_n-{\bbE_{P_n^{\rm o}}}[D_n|X^n,U^n]
\end{pmatrix}\Rightarrow {\calN}(\bzero_2,\Sigma(\btheta)).\label{eq:typeclt}
\end{align}
Convergence holds along every sequence of admissible $(x^n,u^n,t)$, including when $\Sigma(\btheta)$ is singular.
\end{enumerate}
\end{lemma}
The proof of Lemma~\ref{lem:types} is provided in Appendix~\ref{app:types}. Its conditional means and Gaussian limit evaluate the auxiliary event from Lemma~\ref{lem:oneshot}, which is used for analyses in the next subsections.

\subsection{Non-Excess-Distortion Probability and the Optimal Rate}
\label{sec:ach-analysis}
We first bound the operational error using the constructed auxiliary distribution.

\subsubsection{Operational Error and Conditional Gaussian Constraints}
\label{sec:ach-error-gaussian}

Fix a nonempty finite set $\calA_{\rm f}\subset\calA$ of optimal weights, a real number $L\in\bbR$, and a positive real number $\eta\in\bbR_+$. Recall $\Lambda_n$ from~\eqref{eq:source-block-score}. Define the message size used by each source-type code and the information slack by
\begin{align}
\widetilde M
&:=\left\lceil\exp\{nR_{\mathrm{WZ}}(D)
+\sqrt n[L+(1+\lambda^*)\eta]\}\right\rceil,
\label{eq:payload-choice}\\
\gamma_n&:=2\log n.\label{eq:payload-slack}
\end{align}
There are $|\calP_n(\calX)|\leq(n+1)^{|\calX|}$ source types. For each source type $\hatP$ satisfying the cutoff, use Lemma~\ref{lem:oneshot} to choose a code with $\widetilde M$ messages under the source distribution conditioned on that type; its selected weights and allocation will be specified in Section~\ref{sec:ach-source-selection}. Denote its encoder and decoder by $f_{n,\hatP}$ and $\phi_{n,\hatP}$. For an excluded type, define these maps to use a fixed message and reconstruction. The full encoder sends the source type together with this code's message, and the decoder applies the corresponding source-type decoder for $(\hatP,s,y^n)\in\calP_n(\calX)\times[\widetilde M]\times\calY^n$:
\begin{align}
f_n(x^n)&:=\bigl(\hatT_{x^n},f_{n,\hatT_{x^n}}(x^n)\bigr),
\label{eq:ach-full-encoder}\\
\phi_n((\hatP,s),y^n)&:=\phi_{n,\hatP}(s,y^n).
\label{eq:ach-full-decoder}
\end{align}
The message pairs can be indexed by $[M]$ with $M:=|\calP_n(\calX)|\widetilde M$, so this is an $(n,M)$-code and
\begin{align}
\log M
&=nR_{\mathrm{WZ}}(D)+\sqrt n[L+(1+\lambda^*)\eta]+O(\log n).
\label{eq:ach-total-message-size}
\end{align}
The source type determines the weights and allocation, and its description costs at most $|\calX|\log(n+1)$ nats. This specifies an encoder and decoder on every source sequence. Let $Z^n=\phi_n(f_n(X^n),Y^n)$ be the actual code reproduction; $D_n$ remains the auxiliary distortion sum.
On each source type satisfying the cutoff in~\eqref{eq:type-cutoff}, define
\begin{align}
\calE_n^{\rm u}:=\{D_n>nD\}\cup\{\imath_n>\log\widetilde M-\gamma_n\}.\label{eq:ach-auxiliary-failure}
\end{align}
On an excluded type, define $\calE_n^{\rm u}$ to be the whole conditional probability space, so its conditional probability is one. The assembled code satisfies
\begin{align}
\Pen(D)
&=\sum_{\hatP\in\calP_n(\calX)}\Pr\{\hatT_{X^n}=\hatP\}
\Pr\{{d^{(n)}}(X^n,Z^n)>D|\hatT_{X^n}=\hatP\},\label{eq:ach-error-reduction}\\
&\leq\sum_{\hatP\in\calP_n(\calX)}\Pr\{\hatT_{X^n}=\hatP\}
\big[\Pr\{\calE_n^{\rm u}|\hatT_{X^n}=\hatP\}+n^{-2}\big],\label{eq:v20-type-error-2}\\
&=\Pr\{\calE_n^{\rm u}\}+n^{-2},\label{eq:v20-type-error-3}
\end{align}
where~\eqref{eq:ach-error-reduction} follows by conditioning on the source type, ~\eqref{eq:v20-type-error-2} follows from Lemma~\ref{lem:oneshot} with $\gamma_n=2\log n$ on each retained type, whose conditional channel remains $P_{Y|X}^n$, and from the unit error bound on excluded types, and~\eqref{eq:v20-type-error-3} follows by averaging the type-dependent auxiliary distributions.
Each subcode is fixed with error no larger than its random-code mean. No shared randomness is retained in the resulting code. The probability in the last line averages the type-dependent auxiliary distributions on retained types and assigns auxiliary failure probability one to excluded types; it bounds the operational excess-distortion probability.

lead to n a retained type with source score $\Lambda_n$, choose a weight vector $\btheta\in{\calA_{\rm f}}$ and use allocation $t+\eta$ in~\eqref{eq:perturb}. Define the centered scores
\begin{align}
{\Xi_{n,1}^{\rm a}}&:=\frac{\imath_n-\bbE[\imath_n|X^n,U^n]}{\sqrt n},\label{eq:ach-centered-information}\\
{\Xi_{n,2}^{\rm a}}&:=\frac{D_n-\bbE[D_n|X^n,U^n]}{\sqrt n}.\label{eq:ach-centered-distortion}
\end{align}
On retained types, the auxiliary event is
\begin{align}
(\calE_n^{\rm u})^c
=\big\{{\Xi_{n,1}^{\rm a}}\leq L-\Lambda_n-{\lambda^*} t+\eta+o(1),\quad
{\Xi_{n,2}^{\rm a}}\leq t+\eta+o(1)\big\},\label{eq:ach-rectangle}
\end{align}
where~\eqref{eq:ach-rectangle} follows from~\eqref{eq:ach-auxiliary-failure},~\eqref{eq:payload-choice}, and Claim~(i) of Lemma~\ref{lem:types}. The remainders, including $\gamma_n/\sqrt n$ and type rounding, vanish uniformly over permitted types, bounded allocations, and the finite set ${\calA_{\rm f}}$. We next lower bound this event using source-dependent choices of $(\btheta,t)$.

\subsubsection{Source-Dependent Selection and Gaussian Averaging}
\label{sec:ach-source-selection}
From the auxiliary event in~\eqref{eq:ach-rectangle}, the operational non-excess-distortion probability satisfies
\begin{align}
\Pcn(D)\geq \Pr\{(\calE_n^{\rm u})^c\}-n^{-2},\label{eq:ach-success-reduction}
\end{align}
where~\eqref{eq:ach-success-reduction} follows from~\eqref{eq:ach-error-reduction}. Recall $\Lambda_n$ and ${F^{\rm c}_{\btheta}}$ from~\eqref{eq:source-block-score} and~\eqref{eq:gaussian-envelopes}, respectively, and ${\Lambda}\sim {\calN}(0,V)$ from Section~\ref{sec:main-result}. The finite set ${\calA_{\rm f}}\subset\calA$ was fixed in Section~\ref{sec:ach-error-gaussian}. Take any positive real number ${\veps_{\rm d}}\in\bbR_+$ and choose a positive real number ${{A}}\in\bbR_+$ such that $\Pr\{|{\Lambda}|>{{A}}\}\leq{\veps_{\rm d}}$. Partition $[-{{A}},{{A}}]$ into disjoint intervals $\calI_{{j_{\calI}}}$, ${j_{\calI}}\in[{m_{\calI}}]$, of length at most ${{\eta_0}}\in(0,\eta/[8(1+\lambda^*)])$, and let ${\ell_{j_{\calI}}^{\rm s,-}}$ denote the left endpoint of $\calI_{{j_{\calI}}}$. Choose the endpoints to have zero Gaussian probability; when $V=0$, place $0$ inside one interval. Choose a finite positive real number ${t_{\max}}\in\bbR_+$ and, for each ${j_{\calI}}\in[{m_{\calI}}]$, a pair $(\btheta_{{j_{\calI}}},\hatt_{{j_{\calI}}})\in{\calA_{\rm f}}\times[-{t_{\max}},{t_{\max}}]$ satisfying
\begin{align}
\Phi_{\Sigma(\btheta_{{j_{\calI}}})}(L-{\ell_{j_{\calI}}^{\rm s,-}}-\lambda^*\hatt_{{j_{\calI}}},\hatt_{{j_{\calI}}})
\geq\max_{\btheta\in{\calA_{\rm f}}}{F^{\rm c}_{\btheta}}(L-{\ell_{j_{\calI}}^{\rm s,-}})-2{\veps_{\rm d}}.\label{eq:ach-cell-choice}
\end{align}
Such choices exist by~\eqref{eq:gaussian-envelopes}: the set and partition are finite, and each supremum admits a finite ${\veps_{\rm d}}$-maximizing allocation. Take ${t_{\max}}$ larger than all chosen allocations. Round each $\hatt_{{j_{\calI}}}$ to a finite grid to obtain $|t_{{j_{\calI}}}-\hatt_{{j_{\calI}}}|\leq {{\eta_0}}$, and use $(\btheta_{{j_{\calI}}},t_{{j_{\calI}}}+\eta)$ for types with $\Lambda_n\in\calI_{{j_{\calI}}}$ satisfying the cutoff. For retained types with $\Lambda_n\notin[-{A},{A}]$, choose any fixed weight vector in $\calA_{\rm f}$ and allocation zero. These types are discarded in the following lower bound. Then
\begin{align}
\Pcn(D)
&\geq\sum_{{j_{\calI}}\in[{m_{\calI}}]}\Pr\{\Lambda_n\in\calI_{{j_{\calI}}},\ \|\sqrt n(\hatT_{X^n}-P_X)\|\leq n^{1/8}\}\notag\\
&\qquad\times\inf_{\substack{x^n\in\calX^n:\ \Lambda_n(x^n)\in\calI_{{j_{\calI}}}\\
\|\sqrt n(\hatT_{x^n}-P_X)\|\leq n^{1/8}}}
\Pr\{(\calE_n^{\rm u})^c|X^n=x^n\}-n^{-2},\label{eq:ach-partition-bound}
\end{align}
where~\eqref{eq:ach-partition-bound} follows from~\eqref{eq:ach-success-reduction} by conditioning on $X^n$, retaining the displayed source events, and bounding each conditional probability by its infimum. An interval with no retained sequence contributes zero.

For a retained $x^n$ in $\calI_{{j_{\calI}}}$, put ${\ell_{\rm I}}:=L-{\ell_{j_{\calI}}^{\rm s,-}}-\lambda^*\hatt_{{j_{\calI}}}$. The conditional success probability satisfies
\begin{align}
\liminf_{n\to\infty}\inf_{\substack{x^n\in\calX^n:\ \Lambda_n(x^n)\in\calI_{{j_{\calI}}}\\
\|\sqrt n(\hatT_{x^n}-P_X)\|\leq n^{1/8}}}
\Pr\{(\calE_n^{\rm u})^c|X^n=x^n\}
&\geq \Pr\{J_{\btheta_{{j_{\calI}}}}<{\ell_{\rm I}}+\eta/4,\ G_{\btheta_{{j_{\calI}}}}<\hatt_{{j_{\calI}}}+\eta/4\},\label{eq:ach-cell-open}\\
&\geq\Phi_{\Sigma(\btheta_{{j_{\calI}}})}({\ell_{\rm I}},\hatt_{{j_{\calI}}}),\label{eq:ach-cell-closed}\\
&\geq\max_{\btheta\in{\calA_{\rm f}}}{F^{\rm c}_{\btheta}}(L-{\ell_{j_{\calI}}^{\rm s,-}})-2{\veps_{\rm d}},\label{eq:cell-success}
\end{align}
where~\eqref{eq:ach-cell-open} follows from~\eqref{eq:ach-rectangle}, Claim~(ii) of Lemma~\ref{lem:types}, and the open-set Portmanteau bound~\cite[Sec.~2]{billingsley}. Indeed, $|\Lambda_n(x^n)-{\ell_{j_{\calI}}^{\rm s,-}}|\leq {{\eta_0}}$ and $|t_{{j_{\calI}}}-\hatt_{{j_{\calI}}}|\leq {{\eta_0}}$ leave at least $\eta/2$ in both finite-$n$ thresholds for all large $n$. The open rectangle therefore lies inside the auxiliary event, including for singular limiting covariances. Uniformity over retained types and auxiliary sequences follows by applying the sequencewise convergence in Lemma~\ref{lem:types} to any putative violating sequence and then averaging over $U^n$. Equation~\eqref{eq:ach-cell-closed} follows by rectangle inclusion, and~\eqref{eq:cell-success} follows from~\eqref{eq:ach-cell-choice}. Infima are taken only over nonempty retained sets.

Averaging these conditional bounds leads to
\begin{align}
\liminf_{n\to\infty}(\Pcn(D))
&\geq\sum_{{j_{\calI}}\in[{m_{\calI}}]}\Pr\{{\Lambda}\in\calI_{{j_{\calI}}}\}
\big[\max_{\btheta\in{\calA_{\rm f}}}{F^{\rm c}_{\btheta}}(L-{\ell_{j_{\calI}}^{\rm s,-}})-2{\veps_{\rm d}}\big],\label{eq:ach-cell-average}\\
&\geq\bbE\Bigl[\bbI\{|{\Lambda}|\leq {{A}}\}\max_{\btheta\in{\calA_{\rm f}}}{F^{\rm c}_{\btheta}}(L-{\Lambda})\Bigr]-2{\veps_{\rm d}},\label{eq:ach-truncated-average}\\
&\geq\bbE\Bigl[\max_{\btheta\in{\calA_{\rm f}}}{F^{\rm c}_{\btheta}}(L-{\Lambda})\Bigr]-3{\veps_{\rm d}},\label{eq:ach-finite-design}
\end{align}
where~\eqref{eq:ach-cell-average} follows from~\eqref{eq:ach-partition-bound} and~\eqref{eq:cell-success}. The cutoff fails with probability $O(n^{-1/4})$ by Markov's inequality and the bounded second moment of $\sqrt n(\hatT_{X^n}-P_X)$; $\Pr\{\Lambda_n\in\calI_{{j_{\calI}}}\}\to \Pr\{{\Lambda}\in\calI_{{j_{\calI}}}\}$ follows from the scalar central limit theorem and the choice of endpoints. Equation~\eqref{eq:ach-truncated-average} follows from monotonicity of ${F^{\rm c}_{\btheta}}$ and ${\ell_{j_{\calI}}^{\rm s,-}}\leq {\Lambda}$ on $\calI_{{j_{\calI}}}$. Equation~\eqref{eq:ach-finite-design} follows from $0\leq {F^{\rm c}_{\btheta}}\leq1$ and $\Pr\{|{\Lambda}|>{{A}}\}\leq{\veps_{\rm d}}$.

Letting ${\veps_{\rm d}}\downarrow0$ through successive finite designs leads to
\begin{align}
\liminf_{n\to\infty}(\Pcn(D))
\geq\bbE\Bigl[\max_{\btheta\in{\calA_{\rm f}}}{F^{\rm c}_{\btheta}}(L-{\Lambda})\Bigr],\label{eq:ach-finite-family}
\end{align}
where~\eqref{eq:ach-finite-family} follows from~\eqref{eq:ach-finite-design} by a diagonal choice of blocklength thresholds. The next part makes this choice together with removal of the finite-set restriction and the rate enlargement $L+(1+\lambda^*)\eta$ in~\eqref{eq:payload-choice}.

\subsubsection{Limit Passage and Exact Second-Order Rate}
\label{sec:ach-limit-passage}
Recall the definitions of $\calA$, ${F^{\rm c}_{\btheta}}$, and ${F^{\rm c}}$ in~\eqref{eq:weightset},~\eqref{eq:gaussian-envelopes}, and~\eqref{eq:full-pointwise-envelope}, respectively. The following approximation removes the finite-set restriction.

\begin{lemma}\label{lem:finite-approximation}
Under Assumptions~(\ref{ass:finite})--(\ref{ass:interior}), there are increasing nonempty finite subsets ${\calA_k}\subset\calA$, $k\in\bbN$, such that, for every $\ell\in\bbR$,
\begin{align}
\max_{\btheta\in{\calA_k}}{F^{\rm c}_{\btheta}}(\ell)\longrightarrow {F^{\rm c}(\ell)},\label{eq:finiteapprox}
\end{align}
monotonically.
\end{lemma}
The proof is of Lemma \ref{lem:finite-approximation} is provided in Appendix~\ref{app:gaussian}. We apply this approximation to~\eqref{eq:ach-finite-design}: for each $k\in\bbN$ choose a finite set whose expectation differs from $F_*(L)$ by at most $k^{-1}$, choose a real number $\eta\in(0,k^{-1}]$, and choose the finite truncation and selection parameters with total probability loss at most $k^{-1}$. Keep this design fixed as $n\to\infty$. For suitable increasing integers ${n_k}\to\infty$, its codes satisfy, for every $n\geq {n_k}$,
\begin{align}
\Pcn(D)&\geq F_*(L)-4/k,\label{eq:ach-diagonal-probability}\\
\frac{\log M-nR_{\mathrm{WZ}}(D)}{\sqrt n}&\leq L+(2+\lambda^*)/k,\label{eq:ach-diagonal-rate}
\end{align}
where~\eqref{eq:ach-diagonal-probability} follows from~\eqref{eq:ach-finite-design} and the finite-set approximation in Lemma~\ref{lem:finite-approximation}, using monotone convergence. Equation~\eqref{eq:ach-diagonal-rate} follows from~\eqref{eq:payload-choice}, increasing ${n_k}$ to absorb the $O(\log n)$ source-type description. Each subcode has already been fixed deterministically with error no larger than its random-code mean.

For ${n_k}\leq n<{n_{k+1}}$, use the $k$th design. It follows that
\begin{align}
\liminf_{n\to\infty}(\Pcn(D))&\geq F_*(L),\label{eq:ach-diagonal-limit}\\
\limsup_{n\to\infty}\frac{\log M-nR_{\mathrm{WZ}}(D)}{\sqrt n}&\leq L,\label{eq:ach-diagonal-expansion}
\end{align}
where~\eqref{eq:ach-diagonal-limit} follows from~\eqref{eq:ach-diagonal-probability} as $k\to\infty$, and~\eqref{eq:ach-diagonal-expansion} follows from the message-size formula with $\eta\to0$. This establishes the required code construction for every real $L$, including singular covariances. Since the first-order rate is positive, each source-type code has an exponentially growing message size even when $L\leq0$.

To prove Theorem~\ref{thm:achievability}, fix $L>L_*(\veps)$. Claim~(ii) of Lemma~\ref{lem:gaussian} leads to $F_*(L)>1-\veps$. Thus,~\eqref{eq:ach-diagonal-limit}--\eqref{eq:ach-diagonal-expansion} provide a code sequence meeting the excess-distortion criterion at rate offset at most $L$. Consequently,
\begin{align}
L_{\mathrm{WZ}}^*(D,\veps)\leq L_*(\veps),\label{eq:ach-exact-quantile}
\end{align}
where the inequality follows by letting $L\downarrow L_*(\veps)$ in the definition of the second-order coefficient. Combining the preceding rate bound with the success-probability bound~\eqref{eq:ach-success-reduction} and the quantile properties in Lemma~\ref{lem:gaussian}, we complete the proof of Theorem~\ref{thm:achievability}. Corollary~\ref{cor:secondorder} follows by combining Theorems~\ref{thm:main} and~\ref{thm:achievability}.

\section{Conclusion}
\label{sec:conclusion}
We derived the exact second-order asymptotics for the discrete Wyner--Ziv problem. Our main contribution is a new second-order converse bound and a slightly generalized achievability bound. In particular, we showed that the previous best second-order achievability bound in~\cite[Theorem~15]{lili} is not optimal in general and provided numerical examples where the result of~\cite[Theorem~15]{lili} is strictly suboptimal.

We next discuss future research directions. Firstly, one can generalize our results to the Gaussian case under quadratic distortion measure~\cite{wynergeneral}. In this case, the type-based method could not be used. Novel ideas for continuous sequences are required. Secondly, one can consider generalize our results to multi-terminal lossy source coding problems with more than distortion constraints, e.g., the Heegard--Berger problem~\cite{heegard1985}. In this case, one needs to additionally consider the tradeoff of different excess-distortion events and adapt the proof techniques accordingly. Finally, one could explore further refined third-order asymptotics~\cite[Theorem~18]{kostina2012} for the discrete Wyner-Ziv problem. In this case, the third dominant term of the non-asymptotic theoretical benchmark needs to be characterized explicitly.

\section*{Acknowledgment}
Parts of the proofs, algebra, and writing were carried out with assistance from OpenAI's ChatGPT (GPT-5.6 and GPT-6). The author reviewed the AI-assisted derivations and assumes full responsibility for the statements, proofs, and correctness of the mathematical content.

\appendices
\hypersetup{linkcolor=black,citecolor=black,urlcolor=black,anchorcolor=black}

\section*{Additional Notation for the Appendices}
We use $[0]:=\varnothing$, $0\log0=0$, $h_2(0)=h_2(1)=0$, and $\|\cdot\|_1$ for the $\ell_1$ norm. For $n\in\bbN$, ${\boldsymbol v}\in\bbR^n$, nonempty $\calB\subseteq\bbR^n$, ${v}\in\bbR$, and $Q\in\calP(\calX)$, write $\dist({\boldsymbol v},\calB):=\inf_{{\boldsymbol v}'\in\calB}\|{\boldsymbol v}-{\boldsymbol v}'\|$, $({v})_+:=\max\{{v},0\}$, and $\supp(Q):=\{x\in\calX:Q(x)>0\}$. For $\ell\in\bbR$, define the standard Gaussian density and cumulative distribution function by
\begin{align}
{{\Phi'}(\ell)}:=(2\pi)^{-1/2}\exp(-{\ell}^2/2),\qquad \Phi({\ell}):=\int_{-\infty}^{{\ell}}{{\Phi'}(\ell')}\,{d\ell'}.\label{eq:standard-gaussian-functions}
\end{align}

\section{Proof of Proposition~\ref{prop:optimal-properties}}\label{app:optimal-properties}
\label{sec:geometry}
Under Assumptions~(\ref{ass:finite})--(\ref{ass:interior}), recall $\rho$, $\delta$, and $\Delta$ from~\eqref{eq:posterior-entropy},~\eqref{eq:posterior-risk}, and~\eqref{eq:deficit}, and the contact and weight sets $\calC$ and $\calA$ from~\eqref{eq:zero-gap-posteriors}--\eqref{eq:weightset}. We first justify the auxiliary-cardinality bound and derive the dual, then prove positivity, contact regularity, dual uniqueness, and stationarity. These properties establish Proposition~\ref{prop:optimal-properties} and the following lemma used by subsequent proofs.

\Needspace{12\baselineskip}
\begin{lemma}\label{lem:geometry}
Under Assumptions~(\ref{ass:finite})--(\ref{ass:interior}), the following hold.
\begin{enumerate}
\item The normalized dual parameters ${({\lambda^*},\bbeta,\beta_0)}$ in~\eqref{eq:deficit}--\eqref{eq:dualvalue} exist and are unique.
\item The contact set $\calC$ is finite and contained in the interior of $\calP(\calX)$.
\item The Bayes rules $(g_j)_{j\in[m]}$ are locally unique up to actions having identical distortion on every source symbol that can produce the observed side information symbol.
\item There exists ${{C}>0}$ such that, for every $Q\in\calP(\calX)$,
\begin{align}
\dist(Q,\calC)^2\leq {C}\Delta(Q).\label{eq:quadratic}
\end{align}
\item For every $(j,x)\in[m]\times\calX$,
\begin{align}
\sum_{{y\in\calY:\,P_{Y|X}(y|x)>0}}P_{Y|X}(y|x)[\imath_j(x,y)+{\lambda^*} d(x,g_j(y))]
=\jW{x}+{\lambda^*} D.\label{eq:stationarity}
\end{align}
\end{enumerate}
\end{lemma}
\subsection{Auxiliary Cardinality in the Posterior Representation}
\label{app:posterior-cardinality}

We justify the cardinality reduction used in Lemma~\ref{lem:posterior}. Fix any feasible finite mixture in~\eqref{eq:firstorder}, discard its zero-weight components, and relabel the remaining components by $[m_0]$. If $m_0>|\calX|+1$, the vectors $(Q_u,\delta(Q_u))\in\bbR^{|\calX|+1}$, $u\in[m_0]$, are linearly dependent. Thus there is a nonzero weight direction $\dot{\btheta}=(\dot\theta_u)_{u\in[m_0]}$ such that
\begin{align}
\sum_{u\in[m_0]}\dot\theta_uQ_u=\bzero_{|\calX|},\qquad \sum_{u\in[m_0]}\dot\theta_u\delta(Q_u)=0.
\end{align}
Summing the coordinates of the first equality gives $\sum_{u\in[m_0]}\dot\theta_u=0$, so this direction has both positive and negative coordinates. Change its sign, if necessary, so that $\sum_{u\in[m_0]}\dot\theta_u\rho(Q_u)\leq0$.

Starting from $\btheta$, increase the coefficient of this direction until the first weight becomes zero. This occurs at a positive finite coefficient because every initial weight is positive and the direction has a negative coordinate. At that point all weights are nonnegative, the two equalities preserve the source mean and distortion, and the chosen sign ensures that the objective does not increase. Remove the zero weights and repeat until at most $|\calX|+1$ components remain. The test-channel and Bayes-reconstruction construction in the proof of Lemma~\ref{lem:posterior} then gives a feasible auxiliary in~\eqref{eq:wzclassic} with no greater objective. Conversely, every auxiliary in~\eqref{eq:wzclassic} yields a feasible mixture in~\eqref{eq:firstorder}. This proves that the two minima agree.

\subsection{Supporting Inequality and Optimal Posterior Distributions}
Initially, $\lambda^*$ denotes an optimal multiplier in~\eqref{eq:posterior-dual}; uniqueness below will establish~\eqref{eq:multiplier-definition}. The parameters $(\bbeta,\beta_0)$ are from~\eqref{eq:deficit}, and $\nabla$ and $\nabla^2$ denote the gradient and Hessian. For a positive vector ${Q^+}=({Q^+}(x))_{x\in\calX}\in\bbR_+^{|\calX|}$, define the homogeneous extension of $\rho$ by
\begin{align}
\rho({Q^+}):=\Bigl(\sum_{x\in\calX}{Q^+}(x)\Bigr)
\rho\left(\frac{{Q^+}}{\sum_{x\in\calX}{Q^+}(x)}\right).\label{eq:homogeneous-extension}
\end{align}
Use the induced distributions $Q_Y$ and $Q_{X|Y}$ in~\eqref{eq:posterior-output}--\eqref{eq:posterior-conditional}.
For every interior $Q\in\calP(\calX)$ and $x\in\calX$,
\begin{align}
[\nabla\rho(Q)]_x&=\sum_{{y\in\calY:\,P_{Y|X}(y|x)>0}}P_{Y|X}(y|x)\log Q_{X|Y}(x|y),\label{eq:entropy-gradient}\\
\langle Q,\nabla\rho(Q)\rangle&=\rho(Q),\label{eq:v20-gradient-2}
\end{align}
where~\eqref{eq:entropy-gradient} follows by differentiating~\eqref{eq:homogeneous-extension}, and~\eqref{eq:v20-gradient-2} is Euler's identity for a homogeneous function of degree one. Terms with $P_{Y|X}(y|x)=0$ are omitted.

We derive~\eqref{eq:posterior-dual} from the moment form of~\eqref{eq:firstorder}. Let $\mathfrak{m}$ be a finite nonnegative measure on $\calP(\calX)$. The primal constraints are unit mass, mean $P_X$, and distortion at most $D$. For candidate multipliers $(\lambda,\bbeta,\beta_0)\in\bbR_+\times\bbR^{|\calX|}\times\bbR$, its Lagrangian is
\begin{align}
&\int_{\calP(\calX)}\rho(Q)\,\mathfrak{m}(dQ)-\rho(P_X)
+\lambda\left[\int_{\calP(\calX)}\delta(Q)\,\mathfrak{m}(dQ)-D\right]\notag\\
&\quad+\left\langle\bbeta,P_X-\int_{\calP(\calX)}Q\,\mathfrak{m}(dQ)\right\rangle
+\beta_0\left[1-\mathfrak{m}(\calP(\calX))\right]\notag\\
&=\beta_0+\langle\bbeta,P_X\rangle-\rho(P_X)-\lambda D\notag\\
&\quad+\int_{\calP(\calX)}[\rho(Q)+\lambda\delta(Q)-\beta_0-\langle\bbeta,Q\rangle]\,\mathfrak{m}(dQ).
\end{align}
Its infimum over all such measures equals the first four terms on the last line if the integrand is nonnegative for every $Q$, and is $-\infty$ otherwise, by increasing the mass at a violating posterior. Maximizing this infimum gives exactly~\eqref{eq:posterior-dual}.

For unit-mass measures, the attainable moment triples are the convex hull of the compact set $\{(Q,\delta(Q),\rho(Q)):Q\in\calP(\calX)\}$, hence form a compact set and can be represented by finite mixtures. Thus, the primal minimum is attained. For $x\in\calX$, let $\be_x(x')=\bbI\{x'=x\}$ for $x'\in\calX$, and let $P_{\be_x}$ be unit mass at $\be_x$. The mixture $\sum_{x\in\calX}P_X(x)P_{\be_x}$ has mean $P_X$ and distortion $D_{\min}<D$. Full support and this strict inequality persist for all prescribed means and distortion bounds sufficiently near $(P_X,D)$ in the affine hull of the simplex. The minimum of the $\rho$ moment is therefore a finite convex function on this neighborhood and admits a supporting affine function at $(P_X,D)$. Write its mean coefficient as $\bbeta$ and its distortion coefficient as $-\lambda$, where $\lambda\geq0$ because the minimum is nonincreasing in the distortion bound. Set $\beta_0:=R_{\mathrm{WZ}}(D)+\rho(P_X)+\lambda D-\langle\bbeta,P_X\rangle$. Applying the supporting inequality to the one-point posterior measure at $Q$ gives $\beta_0+\langle\bbeta,Q\rangle\leq\rho(Q)+\lambda\delta(Q)$; the dual objective then equals $R_{\mathrm{WZ}}(D)$. This proves dual attainment and strong duality.

Since every posterior has unit mass, the shift $(\bbeta,\beta_0)\mapsto(\bbeta+\beta_g\bone_{|\calX|},\beta_0-\beta_g)$ leaves the dual constraint, objective, and gap unchanged. Here $\bone_{|\calX|}$ is the all-ones vector and $\beta_g\in\bbR$. Choosing $\beta_g=-\langle\bbeta,P_X\rangle$ gives the normalization $\langle\bbeta,P_X\rangle=0$.

To prove positivity, observe that
\begin{align}
I(X;U|Y)=0
&\Longrightarrow P_{U|X}(\cdot|x)=P_{U|Y}(\cdot|y)
\quad\text{whenever }P_{XY}(x,y)>0,\\
&\Longrightarrow U\perp X
\Longrightarrow \bbE[d(X,g(U,Y))]\geq D_0.
\end{align}
The first implication uses equality in conditional mutual information and $U-X-Y$; pairwise overlap~\eqref{eq:overlap} leads to the second by connecting the conditional distributions for every pair of source symbols. Independence then limits distortion to at least $D_0$. Thus, $R_{\mathrm{WZ}}(D)>R_{\mathrm{WZ}}(D_0)=0$; convexity implies $\lambda^*>0$, and complementary slackness makes distortion active.

We next obtain uniform curvature from the overlap constant $\kappa_W$ in~\eqref{eq:overlap}.
For $Q\in\calP(\calX)$ and $(x,x')\in\supp(Q)^2$, define the posterior transition kernel by the equality below. The kernel satisfies
\begin{align}
{K}_Q(x,x'):=\sum_{y\in\calY:\,Q_Y(y)>0}
P_{Y|X}(y|x)\frac{Q(x')P_{Y|X}(y|x')}{Q_Y(y)}
\geq\kappa_W Q(x'),\label{eq:posterior-kernel}
\end{align}
where~\eqref{eq:posterior-kernel} follows from $Q_Y(y)\leq1$ and~\eqref{eq:overlap}; its kernel is reversible with stationary distribution $Q$. For an interior $Q$ and a tangent vector ${{\dot Q}=({\dot Q}(x))_{x\in\calX}}$ with ${\sum_{x\in\calX}{\dot Q}(x)=0}$, the Hessian satisfies
\begin{align}
{\dot Q}^\top\nabla^2\rho(Q){\dot Q}
&=\Biggl[\sum_{x\in\calX}\frac{{\dot Q}(x)^2}{Q(x)}
-\sum_{y\in\calY:\,Q_Y(y)>0}\frac{(\sum_{x\in\calX}{\dot Q}(x)P_{Y|X}(y|x))^2}{Q_Y(y)}\Biggr],\label{eq:hessian}\\
&=\frac12\sum_{(x,x')\in\calX^2}Q(x){K}_Q(x,x')\left[\frac{{\dot Q}(x)}{Q(x)}-\frac{{\dot Q}(x')}{Q(x')}\right]^2,\label{eq:v20-curvature-2}\\
&\geq\kappa_W\sum_{x\in\calX}\frac{{\dot Q}(x)^2}{Q(x)}
\geq\kappa_W\|{\dot Q}\|^2,\label{eq:strongcurv}
\end{align}
where~\eqref{eq:hessian} follows by differentiating~\eqref{eq:entropy-gradient};~\eqref{eq:v20-curvature-2} follows from~\eqref{eq:posterior-kernel} and reversibility; and~\eqref{eq:strongcurv} follows from the minorization, ${\bbE_Q[{\dot Q}(X)/Q(X)]=0}$, and $Q(x)\leq1$. Continuity extends strong convexity to the closed simplex.

To exclude boundary contacts, suppose $Q\in\calC$ has $Q(x_*)=0$ for some $x_*\in\calX$. Choose $x'\in\supp(Q)$ and, for $\vartheta_{\rm b}\in(0,1)$, set $Q_{\vartheta_{\rm b}}=(1-\vartheta_{\rm b})Q+\vartheta_{\rm b}\be_{x_*}$. Outputs with $Q_Y(y)=0$ satisfy $Q_{\vartheta_{\rm b},Y}(y)=\vartheta_{\rm b} P_{Y|X}(y|x_*)$. Define the retained-output mass $\kappa(Q,x_*)$ by the first line below. For sufficiently small $\vartheta_{\rm b}$,
\begin{align}
\kappa(Q,x_*)&:=\sum_{y\in\calY:\,Q_Y(y)>0}P_{Y|X}(y|x_*)\\
&\geq\sum_{y\in\calY}P_{Y|X}(y|x_*)P_{Y|X}(y|x')\geq\kappa_W,\\
\rho(Q_{\vartheta_{\rm b}})-\rho(Q)
&=\vartheta_{\rm b}\log\vartheta_{\rm b}-\sum_{y\in\calY:\,Q_Y(y)=0}\vartheta_{\rm b} P_{Y|X}(y|x_*)\log\vartheta_{\rm b}+O(\vartheta_{\rm b})\\
&=\kappa(Q,x_*)\vartheta_{\rm b}\log\vartheta_{\rm b}+O(\vartheta_{\rm b}),\\
\Delta(Q_{\vartheta_{\rm b}})&=\kappa(Q,x_*)\vartheta_{\rm b}\log\vartheta_{\rm b}+O(\vartheta_{\rm b})<0.\label{eq:boundary-contact-gap}
\end{align}
The overlap bound uses $Q(x')>0$, $P_{Y|X}(y|x')\leq1$, and~\eqref{eq:overlap}. The entropy expansion uses smoothness on positive supports and affineness of the channel entropy term. Since $\delta$ is the minimum of finitely many affine risks with bounded coefficients, its change and the supporting affine terms are $O(\vartheta_{\rm b})$. The negative term dominates because $\kappa(Q,x_*)\geq\kappa_W>0$, contradicting~\eqref{eq:deficit}. Thus, contacts have full support.

For deterministic $g:\calY\to\calZ$, $x\in\calX$, and $Q\in\calP(\calX)$, define
\begin{align}
\bar d_g(x)&:=\sum_{y\in\calY}P_{Y|X}(y|x)d(x,g(y)),\qquad
\overline{\bd}_g:=(\bar d_g(x))_{x\in\calX},\label{eq:rule-distortion-vector}\\
\Delta_g(Q)&:=\rho(Q)+\lambda^*\langle Q,\overline{\bd}_g\rangle-\beta_0-\langle\bbeta,Q\rangle.\label{eq:rule-gap}
\end{align}
For $j\in[m]$, use the Bayes rule $g_j$ from~\eqref{eq:contact-rule-output} in~\eqref{eq:rule-gap} to obtain ${\Delta_{g_j}}$. Since ${\Delta_g\geq\Delta\geq0}$ is strongly convex, each rule has at most one zero. Every contact is such a zero for a Bayes rule; hence $\calC$ is finite.

For any distortion-feasible mixture in Lemma~\ref{lem:posterior},~\eqref{eq:mixture-gap-identity} expresses its excess rate as the average gap plus $\lambda^*$ times its distortion slack. Both terms are nonnegative. Since $\lambda^*>0$, the mixture is optimal if and only if every positive-weight posterior belongs to $\calC$ and its distortion equals $D$. Primal attainment therefore makes $\calC$ nonempty and, after merging repeated contacts, gives a weight vector in $\calA$. Conversely, every $\btheta\in\calA$ is an optimal mixture by the same identity. Summing its mean constraint gives $\sum_{j\in[m]}\theta_j=1$, so $\calA$ is a nonempty closed subset of the probability simplex. Its remaining constraints are affine; hence it is compact and convex.

For $(y,z,z')\in\calY\times\calZ^2$, define
\begin{align}
z\sim_y z'\quad\Longleftrightarrow\quad
\sum_{x\in\calX}P_{Y|X}(y|x)\big(d(x,z)-d(x,z')\big)^2=0.\label{eq:action-equivalence}
\end{align}
Suppose two Bayes rules tie at a contact and differ only by actions $z,z'$ at $y$. For every $x\in\calX$,
\begin{align}
{\lambda^*} P_{Y|X}(y|x)[d(x,z)-d(x,z')]=0,\label{eq:tied-action-gradient}
\end{align}
where~\eqref{eq:tied-action-gradient} follows from stationarity of the two nonnegative smooth functions ${\Delta_g}$ at their common zero: their gradient difference is constant across coordinates, and equality of their values makes this constant zero. Since $\lambda^*>0$,~\eqref{eq:action-equivalence} identifies the tied actions. For $Q$ near $Q_j$,
\begin{align}
\Delta(Q)\geq\frac{\kappa_W}{2}\|Q-Q_j\|^2,\label{eq:local-contact-gap}
\end{align}
where~\eqref{eq:local-contact-gap} follows from~\eqref{eq:strongcurv}: the finitely many nonequivalent actions have a positive Bayes margin at $Q_j$, so ${\Delta=\Delta_{g_j}}$ locally. On the compact complement of these neighborhoods, $\Delta$ has a positive minimum. Together these bounds prove~\eqref{eq:quadratic}.

This proves Claims~(ii)--(iv) of Lemma~\ref{lem:geometry}. We next prove dual uniqueness and stationarity.
\subsection{Uniqueness and Stationarity}
For each contact $Q_j$ and Bayes rule $g_j$ in~\eqref{eq:zero-gap-posteriors} and~\eqref{eq:contact-rule-output}, full support and local uniqueness lead to the following tangent stationarity identity, where $\bone_{|\calX|}$ is the all-ones vector:
\begin{align}
\nabla \rho(Q_j)+{\lambda^*} \overline{\bd}_{g_j}=\bbeta+\beta_0\bone_{|\calX|},\label{eq:contactgradient}
\end{align}
where~\eqref{eq:contactgradient} follows from tangent stationarity at $Q_j$; taking its inner product with $Q_j$ and using~\eqref{eq:entropy-gradient} and $\Delta(Q_j)=0$ fixes the constant ${\beta_0}$.
Fix $\btheta\in\calA$ from~\eqref{eq:weightset} and set $\calJ_+:=\{j\in[m]:\theta_j>0\}$. Since ${\sum_{j\in[m]}}\theta_jQ_j=P_X$ and $\delta(P_X)=D_0>D$, $|\calJ_+|\geq2$.
For two optimal dual triples ${(\lambda^*,\bbeta,\beta_0)}$ and ${(\widetilde{\lambda^*},\tilde{\bbeta},\tilde{\beta}_0)}$, every active contact of a fixed optimum is common to both. Subtracting~\eqref{eq:contactgradient} leads to
\begin{align}
({\lambda^*}-\widetilde{\lambda^*})\overline{\bd}_{g_j}
=\bbeta-\tilde{\bbeta}+(\beta_0-\tilde{\beta}_0)\bone_{|\calX|},
\end{align}
where both triples satisfy~\eqref{eq:contactgradient} at every active contact.
If the multipliers differ, all $j\in\calJ_+$ have the same $\overline{\bd}_{g_j}$ and hence the same $\Delta_{g_j}$, whose strong convexity permits only one contact, contradicting $|\calJ_+|\geq2$. Thus, $\lambda^*=\widetilde{\lambda^*}$. Normalization $\langle\bbeta,P_X\rangle=\langle\tilde{\bbeta},P_X\rangle=0$ in the same identity implies $\beta_0=\tilde{\beta}_0$ and $\bbeta=\tilde{\bbeta}$, proving Lemma~\ref{lem:geometry}(i). Every supporting slope of the finite convex function $R_{\mathrm{WZ}}$ at the interior level $D$ supplies an optimal multiplier; uniqueness implies differentiability and $-R_{\mathrm{WZ}}'(D)=\lambda^*>0$, proving Proposition~\ref{prop:optimal-properties}(i).

For Claim~(ii), the fixed-rule feasible set in~\eqref{eq:fixed-alphabet} is nonempty and compact, and information and distortion are continuous. Hence $\calP^*(D)$ is nonempty and compact; complementary slackness and $\lambda^*>0$ lead to~\eqref{eq:optimal-active-constraint}.

For Claim~(iii), subtract $\nabla\rho(P_X)$ from~\eqref{eq:contactgradient} and use $\imath_j$ in~\eqref{eq:contact-information}:
\begin{align}
\sum_{{y\in\calY:\,P_{Y|X}(y|x)>0}}P_{Y|X}(y|x)[\imath_j(x,y)+{\lambda^*}d(x,g_j(y))]
&=\beta(x)+\beta_0-[\nabla\rho(P_X)]_x,\label{eq:contact-conditional-mean}
\end{align}
where~\eqref{eq:contact-conditional-mean} uses~\eqref{eq:entropy-gradient} and~\eqref{eq:contact-information}. Complementary slackness places every active optimal posterior at a contact with a Bayes rule. It follows from averaging given $X=x$ in~\eqref{eq:existing-tilted} that for every $P\in\calP^*(D)$,
\begin{align}
\jmath_P(x|D)=\beta(x)+\beta_0-[\nabla\rho(P_X)]_x-{\lambda^*}D,\label{eq:common-tilted-from-dual}
\end{align}
The right side of~\eqref{eq:common-tilted-from-dual} is independent of $P$, proving Claim~(iii). Denote it by $\jW{x}$; substitution in~\eqref{eq:contact-conditional-mean} leads to~\eqref{eq:stationarity}, proving Lemma~\ref{lem:geometry}(v).

For Claim~(iv), averaging~\eqref{eq:common-tilted-from-dual} under $P_X$ leads to
\begin{align}
\bbE_{P_X}[\jW{X}]&=\langle P_X,\bbeta\rangle+\beta_0-\rho(P_X)-{\lambda^*}D
=R_{\mathrm{WZ}}(D),\label{eq:v20-source-mean-1}\\
\jW{x}-R_{\mathrm{WZ}}(D)&=\beta(x)-[\nabla\rho(P_X)]_x
-\langle P_X,\bbeta-\nabla\rho(P_X)\rangle,\label{eq:v20-source-mean-2}
\end{align}
Equation~\eqref{eq:v20-source-mean-1} uses Euler's identity~\eqref{eq:v20-gradient-2} and~\eqref{eq:dualvalue}; subtracting it from~\eqref{eq:common-tilted-from-dual} leads to~\eqref{eq:v20-source-mean-2}. The common density and source marginal $P_X$ prove Claim~(iv).

For Claim~(v), use $\bard$ and ${\widetilde V}$ in~\eqref{eq:mean-distortion-rule} and~\eqref{eq:existing-direction}. Fix $P\in\calP^*(D)$. Every active posterior $Q_{\widetilde u}=P_{X|{\widetilde U}}(\cdot|{\widetilde u})$ has full support, so $P_{{\widetilde U}|X}({\widetilde u}|x)>0$ for all $x\in\calX$. If ${\widetilde V}(P)=0$, its sum of squares forces $\bard(x,{\widetilde u})=\bar d_{{\widetilde g}({\widetilde u},\cdot)}(x)$ to be independent of active ${\widetilde u}$. The functions $\Delta_{{\widetilde g}({\widetilde u},\cdot)}$ then coincide and, by strong convexity, have at most one zero. Thus, every active posterior equals $P_X$, giving distortion $D_0>D$, contrary to~\eqref{eq:optimal-active-constraint}. Hence ${\widetilde V}(P)>0$, and continuity on compact $\calP^*(D)$ implies~\eqref{eq:uniform-direction}. This proves Claim~(v) and completes Proposition~\ref{prop:optimal-properties} and Lemma~\ref{lem:geometry}.

\section{Reconstruction Rules and Proof of Proposition~\ref{prop:covariance-family}}
\label{app:geometry}
We establish the fixed-rule representation, characterize all optimal covariances under Assumptions~(\ref{ass:finite})--(\ref{ass:interior}) using Lemma~\ref{lem:geometry}, and prove existence of the maximizer used in Lemma~\ref{lem:comparison}.
\subsection{Fixed Reconstruction Rule Representation}
\label{app:rule-alphabet}
In~\eqref{eq:fixed-alphabet}, ${\widetilde{\calU}}=[|\calZ|^{|\calY|}]$ indexes all maps ${\widetilde g}({\widetilde u},\cdot):\calY\to\calZ$. For feasible $({U},g)$ in~\eqref{eq:wzclassic}, let ${\widetilde U}$ index $g({U},\cdot)$. Then ${\widetilde U}-X-Y$, reconstruction is unchanged, and data processing implies $I(X;{\widetilde U}|Y)\leq I(X;{U}|Y)$. Conversely, a fixed-rule distribution is a finite-auxiliary Wyner--Ziv distribution; the reduction to $|\calX|+1$ labels in Appendix~\ref{app:posterior-cardinality} leaves the unrestricted finite-auxiliary minimum unchanged. These directions prove~\eqref{eq:fixed-alphabet}.
\subsection{Preservation of Optimal Fluctuations}

Recall the optimal set $\calP^*(D)$ and weight/covariance representation in~\eqref{eq:optimal-set},~\eqref{eq:weightset}, and~\eqref{eq:mixture-covariance}. Lemma~\ref{lem:geometry} and complementary slackness place every active posterior at a contact with a Bayes rule. For active ${\widetilde u}$, set $\theta_{\widetilde u}=P_{\widetilde U}({\widetilde u})>0$ and $Q_{\widetilde u}=P_{X|{\widetilde U}}(\cdot|{\widetilde u})$, with output marginal $Q_{{\widetilde u},Y}$. For $(x,y)\in\calX\times\calY$ with $P_{Y|X}(y|x)>0$,
\begin{align}
\frac{P_{{\widetilde U}|X}({\widetilde u}|x)}{P_{{\widetilde U}|Y}({\widetilde u}|y)}
=\frac{Q_{\widetilde u}(x)P_Y(y)}{P_X(x)Q_{{\widetilde u},Y}(y)},\label{eq:merged-density-ratio}
\end{align}
Equation~\eqref{eq:merged-density-ratio} follows from Bayes' rule and ${\widetilde U}-X-Y$. Equal contacts have equal information densities and, by~\eqref{eq:action-equivalence}, equal distortions. Merging them preserves~\eqref{eq:existing-tilted} and~\eqref{eq:cov}, yielding weights in $\calA$. Conversely,~\eqref{eq:optimal-joint} is optimal for each $\btheta\in\calA$ and has covariance $\Sigma(\btheta)$. It belongs to the fixed-rule set because distinct contacts have distinct Bayes rules: $\Delta_g$ in~\eqref{eq:rule-gap} is strongly convex and has at most one zero.

To include every optimal finite auxiliary, fix $({U},g^{*})$ with joint distribution $P^*=P_XP^*_{{U}|X}P_{Y|X}$ and let ${\widetilde U}$ index $g^{*}({U},\cdot)$ as above. Reconstruction is unchanged, and
\begin{align}
R_{\mathrm{WZ}}(D)\leq I(X;{\widetilde U}|Y)\leq I(X;{U}|Y)=R_{\mathrm{WZ}}(D),\label{eq:v20-merging-1}
\end{align}
The inequalities in~\eqref{eq:v20-merging-1} use feasibility and data processing; its equality uses optimality of ${U}$. For $x\in\calX$, merged labels share the conditional mean distortion $\bard(x,{u}):=\sum_{y\in\calY}P_{Y|X}(y|x)d(x,g^{*}({u},y))$. Merging preserves its first two moments under $P^*_{{U}|X}(\cdot|x)$ and hence ${\widetilde V}$ in~\eqref{eq:existing-direction}.
For active ${u}$, set $\theta_{u}=P^*_{U}({u})>0$ and $Q_{u}=P^*_{X|{U}}(\cdot|{u})$. For ${\widetilde u}\in{\widetilde{\calU}}$, let $\calU_{\widetilde u}$ collect the active classical labels mapped to ${\widetilde u}$. For each nonempty merged group $\calU_{\widetilde u}$, set $\theta_{\widetilde u}=\sum_{{u}\in\calU_{\widetilde u}}\theta_{u}>0$ and $Q_{\widetilde u}=\sum_{{u}\in\calU_{\widetilde u}}(\theta_{u}/\theta_{\widetilde u})Q_{u}$. Empty groups are omitted below. Then
\begin{align}
0&=\sum_{{\widetilde u}\in{\widetilde{\calU}}:\,\calU_{\widetilde u}\ne\varnothing}\Bigl[\sum_{{u}\in\calU_{\widetilde u}}\theta_{u}\rho(Q_{u})-\theta_{\widetilde u}\rho(Q_{\widetilde u})\Bigr],\label{eq:v20-merge-gap-1}\\
&\geq\frac{\kappa_W}{2}\sum_{{\widetilde u}\in{\widetilde{\calU}}:\,\calU_{\widetilde u}\ne\varnothing}\sum_{{u}\in\calU_{\widetilde u}}\theta_{u}\|Q_{u}-Q_{\widetilde u}\|^2,\label{eq:v20-merge-gap-2}
\end{align}
where~\eqref{eq:v20-merge-gap-1} is the zero information loss in~\eqref{eq:v20-merging-1}, expressed through conditional entropies, and~\eqref{eq:v20-merge-gap-2} follows from~\eqref{eq:strongcurv} and the mixture formula for $Q_{\widetilde u}$.
Thus, $Q_{u}=Q_{\widetilde u}$ in each group, so~\eqref{eq:merged-density-ratio} and unchanged reconstruction preserve covariance. Every fixed-rule distribution is also an admissible finite auxiliary; hence the attainable set is exactly $\calK_D$ in~\eqref{eq:covariance-family-definition}. Proposition~\ref{prop:optimal-properties}(iii) makes the source Gaussian $\Lambda$ common to all optimizers. Substituting this covariance image in~\eqref{eq:fixed-versus-adaptive} leads to~\eqref{eq:fixed-covariance-representation} and completes Proposition~\ref{prop:covariance-family}.

\subsection{\texorpdfstring{Existence of a Maximizing Covariance}{Existence of a Maximizing Covariance}}
The linear covariance map~\eqref{eq:mixture-covariance} sends nonempty compact $\calA$ onto $\calK_D$, so $\calK_D$ is nonempty and compact. For $\ell\in\bbR$ and $\Sigma\in\calK_D$, define
\begin{align}
F^{\rm c}_{\Sigma}(\ell):=\sup_{t\in\bbR}\Phi_\Sigma(\ell-\lambda^*t,t).\label{eq:covariance-envelope-usc}
\end{align}
Let $(\Sigma_k)_{k\in\bbN}$ converge to $\Sigma$ in $\calK_D$ and, for each $k\in\bbN$, choose $t_k\in\bbR$ such that $F^{\rm c}_{\Sigma_k}(\ell)\leq\Phi_{\Sigma_k}(\ell-\lambda^*t_k,t_k)+1/k$. If ${\limsup_{k\to\infty}F^{\rm c}_{\Sigma_k}(\ell)>0}$, uniform Gaussian marginal tails and $\lambda^*>0$ bound $t_k$ along a maximizing subsequence. Passing to $t_k\to t$ and applying the closed-set inequality for weak convergence leads to
\begin{align}
\limsup_{k\to\infty} F^{\rm c}_{\Sigma_k}(\ell)
&\leq\Phi_\Sigma(\ell-\lambda^*t,t)\leq F^{\rm c}_{\Sigma}(\ell).\label{eq:covariance-usc-bound}
\end{align}
The bound is immediate if the limiting upper value is zero and holds for singular $\Sigma$. Fix $L\in\bbR$ and the common $\Lambda\sim \calN(0,V)$ from~\eqref{eq:source-variance}. Since $0\leq F^{\rm c}_\Sigma(\ell)\leq1$, bounded reverse Fatou leads to
\begin{align}
\limsup_{k\to\infty}\bbE[F^{\rm c}_{\Sigma_k}(L-\Lambda)]\leq\bbE[F^{\rm c}_{\Sigma}(L-\Lambda)].
\end{align}
Thus, this upper semicontinuous objective attains its maximum on compact $\calK_D$, proving the existence used in Lemma~\ref{lem:comparison}. The argument includes singular covariances and a zero-variance source Gaussian.

\section{Proof of Lemma~\ref{lem:localization}}
\label{app:localization}
We use the sequential posterior $Q_i^{\rm s}$, gap $\Delta_n$, scores $(\Lambda_n,J_n,G_n)$, and allocation $\tilde t_n$ from~\eqref{eq:sequential-posterior},~\eqref{eq:accumulated-gap}, and~\eqref{eq:source-block-score}--\eqref{eq:code-allocation}, and the success event $\calE_n$ of~\eqref{eq:success-event}; $\calC$ and $\calA$ are defined in~\eqref{eq:zero-gap-posteriors}--\eqref{eq:weightset}. Lemma~\ref{lem:posterior} supplies the supporting gap~\eqref{eq:deficit}, and Lemma~\ref{lem:geometry} its contact properties. We first derive exact rate and distortion decompositions, then control their residuals by the quadratic gap bound, and finally prove the threshold violations and posterior moment and covariance estimates.
Unless specified otherwise, $P_n^{\rm o}$ governs all distributions and probability statements, with the source-only description $S$ determined by $(\tilX^n,{\Upsilon})$ and independent uniform permutation ${\Upsilon}$ of~\eqref{eq:permutation-convention}. Constants ${C}$ may change and depend only on the fixed model and parameters. For the nonnegative regret ${\bar\mu_i^{\mathrm d}}$ in~\eqref{eq:reconstruction-regret} and a positive-probability binary source-only flag $B=B(\tilX^n,{\Upsilon})$ included in $S$, define
\begin{align}
\Delta_{n,k}&:=\sum_{i\in[k]}[\Delta(Q_i^{\rm s})+{\bar\mu_i^{\mathrm d}}],\qquad k\in[n],\label{eq:prefix-gap}\\
P_n^{\rm b}&:=P_n^{\rm o}\{\,\cdot|B=1\}.\label{eq:source-conditioned-distribution}
\end{align}
Likelihoods retain their $P_n^{\rm o}$ normalization. The sequential form of~\eqref{eq:mixture-gap-identity} below retains the nonnegative gap $\Delta_n$.

\subsection{Two-Filtration Identity and Exact Decomposition}
In the next identities, $Q\in\calP(\calX)$ denotes a one-letter posterior; $Q_Y$ and $Q_{X|Y}$ are induced by $Q_{XY}$ as in~\eqref{eq:posterior-output}--\eqref{eq:posterior-conditional}. For ${\widetilde a^{\circ}}:\supp(Q)\to\bbR$ and $x\in\supp(Q)$, define the action of the kernel in~\eqref{eq:posterior-kernel} by
\begin{align}
({{K}_Q} {\widetilde a^{\circ}})(x):=\sum_{\substack{y\in\calY:\ Q_Y(y)>0}}P_{Y|X}(y|x)
\sum_{x'\in\supp(Q)}Q_{X|Y}(x'|y){\widetilde a^{\circ}}(x'),
\end{align}
With $\mathrm{Id}$ the identity operator, for a $Q$-centered $a^{\circ}:\supp(Q)\to\bbR$, take the inverse on the $Q$-centered subspace and write $\operatorname{osc}(a^{\circ}):={\max_{x\in\supp(Q)}a^{\circ}(x)-\min_{x\in\supp(Q)}a^{\circ}(x)}$. The minorization~\eqref{eq:posterior-kernel} leads to
\begin{align}
(\mathrm{Id}-{K}_Q)^{-1}a^{\circ}=\sum_{r=0}^{\infty}{K}_Q^ra^{\circ},\qquad
\operatorname{osc}\big((\mathrm{Id}-{K}_Q)^{-1}a^{\circ}\big)\leq\frac{\operatorname{osc}(a^{\circ})}{\kappa_W},\label{eq:v20-resolvent-1}
\end{align}
by summing the geometric oscillation contraction; the series is $Q$-centered uniformly over support faces. With $\widetilde a_Q^{\circ}:=(\mathrm{Id}-{K}_Q)^{-1}a^{\circ}$ and $(X,Y)\sim Q_{XY}$,
\begin{align}
a^{\circ}(X)
&=\big(\widetilde a_Q^{\circ}(X)-\bbE_Q[\widetilde a_Q^{\circ}(X)|Y]\big)
+\big(a^{\circ}(X)-\widetilde a_Q^{\circ}(X)+\bbE_Q[\widetilde a_Q^{\circ}(X)|Y]\big),\label{eq:twofiltration}
\end{align}
the two terms are centered given $Y$ and $X$, respectively, since $(\mathrm{Id}-{K}_Q)\widetilde a_Q^{\circ}=a^{\circ}$. All functions are restricted to $\supp(Q)$.

Recall $\rho$ from~\eqref{eq:posterior-entropy}, the supporting parameters from~\eqref{eq:deficit}, and $g_Q$ from~\eqref{eq:bayes-rule-local}. All posterior gradients below are restricted to the support face of $Q$. For $x\in\calX$, set ${\bar d}_{g_Q}(x):=\sum_{y\in\calY}P_{Y|X}(y|x)d(x,g_Q(y))$. Define the $Q$-centered function explicitly by
\begin{align}
{{a}}_Q(x):=[\nabla\rho(Q)]_x+{\lambda^*}{\bar d}_{g_Q}(x)-{\beta(x)}-{\beta_0}-\Delta(Q),
\qquad x\in\supp(Q),\label{eq:posterior-centered-gap}
\end{align}
Define the resolvent used below by
\begin{align}
{{\widetilde a}}_Q:=({\mathrm{Id}}-{{K}_Q})^{-1}{{a}}_Q.\label{eq:posterior-resolvent}
\end{align}
Euler's identity and~\eqref{eq:deficit} lead to $\bbE_Q[{{a}}_Q(X)]=0$. For $x\in\supp(Q)$ and $y\in\calY$ with $P_{Y|X}(y|x)>0$, define the generic posterior information density and the two scores centered over $Y\sim P_{Y|X}(\cdot|x)$ by
\begin{align}
\imath_Q(x,y)
&:=\log\frac{Q(x)}{P_X(x)}+\log\frac{P_Y(y)}{Q_Y(y)},\\
\imath_Q^{\rm c}(x,y)
&:=\imath_Q(x,y)-\sum_{\substack{y'\in\calY:\ P_{Y|X}(y'|x)>0}}
P_{Y|X}(y'|x)\imath_Q(x,y'),\\
d_Q^{\rm c}(x,y)
&:=d(x,g_Q(y))-\sum_{y'\in\calY}P_{Y|X}(y'|x)d(x,g_Q(y')).
\end{align}
For contact posteriors these agree with~\eqref{eq:contact-information}--\eqref{eq:centered-distortion-score}.
Recall the common tilted-density expression~\eqref{eq:common-tilted-from-dual}.
Substituting these definitions and~\eqref{eq:posterior-centered-gap} leads to
\begin{align}
\imath_Q(x,y)+{\lambda^*} d(x,g_Q(y))
&=\jW{x}+{\lambda^*} D+\Delta(Q)+{{a}}_Q(x)
+\imath_{Q}^{\rm c}(x,y)+{\lambda^*} d_{Q}^{\rm c}(x,y),\label{eq:pointwise}
\end{align}
At each contact $Q_j\in\calC$, ${a}_{Q_j}=0$ by $\Delta(Q_j)=0$ and~\eqref{eq:contactgradient}.

Recall the definitions of ${Q_i^{\rm s}}$, $d_i$, $d_i^*$, ${{\bar\mu_i^{\mathrm d}}}$, and ${{\Pi}_{\calC}}({Q})$ in~\eqref{eq:sequential-posterior},~\eqref{eq:coordinate-distortion},~\eqref{eq:bayes-coordinate-distortion},~\eqref{eq:reconstruction-regret}, and~\eqref{eq:nearest-contact}, respectively. For each $i\in[n]$, the decomposition~\eqref{eq:twofiltration} and the differences between the centered scores at ${Q_i^{\rm s}}$ and ${{\Pi}_{\calC}}({Q_i^{\rm s}})$ define
\begin{align}
{{\mu}_i}^{\rm f}&:={{\widetilde a}}_{{Q_i^{\rm s}}}(\tilX_i)-\bbE_{{Q_i^{\rm s}}}[{{\widetilde a}}_{{Q_i^{\rm s}}}(\tilX_i)|\tilY_i],\label{eq:forward-residual}\\
{{\mu}_i}^{\rm r}
&:={{a}}_{{Q_i^{\rm s}}}(\tilX_i)-{{\widetilde a}}_{{Q_i^{\rm s}}}(\tilX_i)+\bbE_{{Q_i^{\rm s}}}[{{\widetilde a}}_{{Q_i^{\rm s}}}(\tilX_i)|\tilY_i]\notag\\
&\quad+\imath_{{Q_i^{\rm s}}}^{\rm c}(\tilX_i,\tilY_i)-\imath_{{{\Pi}_{\calC}}({Q_i^{\rm s}})}^{\rm c}(\tilX_i,\tilY_i)
+{\lambda^*}\big(d_{{Q_i^{\rm s}}}^{\rm c}(\tilX_i,\tilY_i)-d_{{{\Pi}_{\calC}}({Q_i^{\rm s}})}^{\rm c}(\tilX_i,\tilY_i)\big),\label{eq:reverse-residual}\\
{{\mu_i^{\mathrm d}}}&{:={\lambda^*}(d_i-d_i^*)-{\bar\mu_i^{\mathrm d}},}\\
{{\mu}_n}&{:=\sum_{i\in[n]}({\mu}_i^{\rm f}+{\mu}_i^{\rm r}),\qquad {\mu_n^{\mathrm d}}:=\sum_{i\in[n]}{\mu_i^{\mathrm d}}.}\label{eq:residual-sums}
\end{align}
For $(i,k)\in\{0,\ldots,n\}^2$, define the filtrations by
\begin{align}
{\calF}_i^{\rm f}:={\calS}(S,{\Upsilon},\tilY^n,\tilX^i),\qquad
{\calF}_k^{\rm r}:={\calS}(\tilX^n,{\Upsilon},\tilY_{n-k+1}^n).
\end{align}
For $i\in[n]$ and $y\in\calY$, define
\begin{align}
Q_{i,Y}^{\rm s}(y)&:=\sum_{x\in\calX}Q_i^{\rm s}(x)P_{Y|X}(y|x),\\
Q_{i,X|Y}^{\rm s}(x|y)&:=\frac{Q_i^{\rm s}(x)P_{Y|X}(y|x)}{Q_{i,Y}^{\rm s}(y)},
\qquad x\in\calX,\ Q_{i,Y}^{\rm s}(y)>0.
\label{eq:sequential-conditioned-posterior}
\end{align}
Under $P_n^{\rm o}$, this is the conditional distribution of $\tilX_i$ given ${\calF}_{i-1}^{\rm f}$ when $y=\tilY_i$:
the message is determined by the source and permutation, and past observations are irrelevant once $\tilX^{i-1}$ is known.
The first term of~\eqref{eq:twofiltration} and the definition of ${\bar\mu_i^{\mathrm d}}$ therefore lead to
\begin{align}
\bbE_{P_n^{\rm o}}[{\mu}_i^{\rm f}|{\calF}_{i-1}^{\rm f}]&=0,\qquad
\bbE_{P_n^{\rm o}}[{\mu_i^{\mathrm d}}|{\calF}_{i-1}^{\rm f}]=0,\\
\bbE_{P_n^{\rm o}}[{\mu}_i^{\rm r}|{\calF}_{n-i}^{\rm r}]&=0.
\label{eq:localization-residual-centering}
\end{align}
The last equality uses measurability of $Q_i^{\rm s}$ given $(\tilX^n,{\Upsilon},\tilY_{i+1}^n)$,
$\tilY_i\sim P_{Y|X}(\cdot|\tilX_i)$, and the second centered term of~\eqref{eq:twofiltration}.

For $H_n$ and $\Delta_n^{\rm p}$ in~\eqref{eq:message-likelihood} and~\eqref{eq:posterior-gap-sum}, Bayes' rule, the posterior chain rule, and~\eqref{eq:sequential-conditioned-posterior} lead to
\begin{align}
H_n
&=\log\frac{P_n^{\rm o}(\tilX^n|S,\tilY^n,{\Upsilon})}{P_n^{\rm o}(\tilX^n|\tilY^n)}\\
&=\sum_{i\in[n]}\log\frac{Q_{i,X|Y}^{\rm s}(\tilX_i|\tilY_i)}{P_{X|Y}(\tilX_i|\tilY_i)}\\
&=\sum_{i\in[n]}\imath_{Q_i^{\rm s}}(\tilX_i,\tilY_i).
\label{eq:v20-likelihood-chain-1}
\end{align}
The first equality uses $S$ being determined by $(\tilX^n,{\Upsilon})$, the second the product source distribution conditional on ${\Upsilon}$, and the third the definition of $\imath_{Q_i^{\rm s}}$.
Summing~\eqref{eq:pointwise} and substituting
$\lambda^*\sqrt n\tilde t_n=\lambda^*(nD-{\sum_{i\in[n]}} d_i^*)+\lambda^*\sqrt nG_n+\Delta_n^{\rm p}$
from~\eqref{eq:code-allocation} yields
\begin{align}
H_n&=nR_{\mathrm{WZ}}(D)+\sqrt n(\Lambda_n+J_n+{\lambda^*} \tilde t_n)+{\mu}_n,\label{eq:exactH}\\
\sum_{i\in[n]} d_i&=nD+\sqrt n(G_n-\tilde t_n)+(\Delta_n+{\mu_n^{\mathrm d}})/{\lambda^*},\label{eq:exactd}
\end{align}
\begin{align}
H_n+{\lambda^*}\Bigl(\sum_{i\in[n]} d_i-nD\Bigr)
&=nR_{\mathrm{WZ}}(D)+\sqrt n(\Lambda_n+J_n+{\lambda^*} G_n)+\Delta_n+{\mu}_n+{\mu_n^{\mathrm d}}.\label{eq:combined}
\end{align}
Here~\eqref{eq:exactd} also uses $\lambda^*{\sum_{i\in[n]}}(d_i-d_i^*)={\sum_{i\in[n]}} {\bar\mu_i^{\mathrm d}}+{\mu_n^{\mathrm d}}$, and~\eqref{eq:combined} adds the first two identities.
\subsection{Residual Estimates}
We bound residual moments and tails by the nonnegative gap, then localize the gap on $\calE_n$. Set
\begin{align}
{\underline p_W}:=\min\{P_{Y|X}(y|x):(x,y)\in\calX\times\calY,\ P_{Y|X}(y|x)>0\}>0.\label{eq:min-positive-entry}
\end{align}
Fix $Q\in\calP(\calX)$, put $Q_y:=Q_{X|Y}(\cdot|y)$ when $Q_Y(y)>0$, and define ${b}_Q(x):=\log Q(x)-\bbE_Q[\log Q(X)]$ on $\supp(Q)$. Omitting zero-transition summands, the posterior identity leads to
\begin{align}
[\nabla \rho(Q)]_x-({\mathrm{Id}}-{{K}_Q}){{b}_Q}(x)
={\sum_{\substack{y\in\calY:\ P_{Y|X}(y|x)>0}}}P_{Y|X}(y|x)
\bigl[\log P_{Y|X}(y|x)-\bbE_{Q_y}[\log P_{Y|X}(y|X)]-H(Q_y)\bigr],
\end{align}
where substitution of~\eqref{eq:posterior-conditional} into~\eqref{eq:entropy-gradient} separates the source and channel logarithms.
Set ${a_Q^{\mathrm f}}:={a}_Q-(\mathrm{Id}-{K}_Q){b}_Q$. Both terms are $Q$-centered, so the resolvent~\eqref{eq:v20-resolvent-1} applies. Define
\begin{align}
{\widetilde a}_Q^{\rm b}&:=(\mathrm{Id}-{K}_Q)^{-1}{a_Q^{\mathrm f}},\label{eq:v20-residual-split-2}\\
{\widetilde a}_Q&={b}_Q+{\widetilde a}_Q^{\rm b},\label{eq:v20-residual-split-1}\\
\sup_{Q\in\calP(\calX)}\max_{x\in\supp(Q)}|{\widetilde a}_Q^{\rm b}(x)|&<\infty,
\label{eq:v20-residual-split-3}
\end{align}
The split follows from $(\mathrm{Id}-{K}_Q){\widetilde a}_Q={a}_Q$. The gradient identity bounds ${a_Q^{\mathrm f}}$ uniformly, since positive transitions are at least $\underline p_W$ and $H(Q_y)\leq\log|\calX|$;~\eqref{eq:v20-resolvent-1} and centering then bound ${\widetilde a}_Q^{\rm b}$.
For $(x,y)\in\calX\times\calY$ with $Q(x)P_{Y|X}(y|x)>0$, the forward residual is
\begin{align}
{\widetilde a}_Q(x)-\bbE_{Q_y}[{\widetilde a}_Q(X)]
&=\log Q_y(x)-\bbE_{Q_y}[\log Q_y(X)]
-\log P_{Y|X}(y|x)+\bbE_{Q_y}[\log P_{Y|X}(y|X)]\notag\\
&\quad+{\widetilde a}_Q^{\rm b}(x)-\bbE_{Q_y}[{\widetilde a}_Q^{\rm b}(X)],
\label{eq:forward-log-cancel}
\end{align}
where~\eqref{eq:v20-residual-split-1} and Bayes' rule cancel the normalizing term $\log Q_Y(y)$ under conditional centering.
For any real number ${\bar a_0}\in(0,1)$, uniformly over probability vectors ${Q}$ on $\calX$,
\begin{align}
\sum_{x\in\calX:\,{Q}(x)>0}{Q}(x)^{1-{\bar a_0}}(1+|\log {Q}(x)|^2)\leq {{C}},\label{eq:posterior-log-moment}
\end{align}
where~\eqref{eq:posterior-log-moment} follows from boundedness of ${{p_0}}^{1-{\bar a_0}}(1+|\log {{p_0}}|^2)$ on $(0,1]$. In~\eqref{eq:forward-log-cancel}, all other terms are uniformly bounded by~\eqref{eq:min-positive-entry} and~\eqref{eq:v20-residual-split-3}. A uniform conditional exponential second moment follows from these bounds and $|\bbE_{{Q}}[\log {Q}(X)]|\leq\log|\calX|$.

For $y$ with $Q_Y(y)>0$, define
\begin{align}
{c}_Q(y):=-H(Q_y)-\bbE_{Q_y}[\log P_{Y|X}(y|X)]
+\bbE_{Q_y}[{\widetilde a}_Q^{\rm b}(X)]+\log P_Y(y).
\end{align}
For $Q(x)P_{Y|X}(y|x)>0$, the split~\eqref{eq:v20-residual-split-1} and Bayes' rule lead to
$\bbE_{Q_y}[{\widetilde a}_Q(X)]+\imath_Q(x,y)={b}_Q(x)-\log P_X(x)+{c}_Q(y)$.
Centering under $P_{Y|X}(\cdot|x)$ yields
\begin{align}
&{a}_Q(x)-{\widetilde a}_Q(x)+\bbE_{Q_y}[{\widetilde a}_Q(X)]+\imath_Q^{\rm c}(x,y)+\lambda^*d_Q^{\rm c}(x,y)\notag\\
&\quad={c}_Q(y)+\lambda^*d(x,g_Q(y))
-\sum_{\substack{y'\in\calY:\ P_{Y|X}(y'|x)>0}}
P_{Y|X}(y'|x)[{c}_Q(y')+\lambda^*d(x,g_Q(y'))],
\end{align}
where the last equality follows from ${{a}}_Q-{{\widetilde a}}_Q=-{{K}_Q}{{\widetilde a}}_Q$; the three terms depending only on $x$ cancel. The remaining ${{c}_Q}$ is bounded by~\eqref{eq:min-positive-entry},~\eqref{eq:v20-residual-split-3}, and positivity of $P_Y$. A uniform bound on ${{\mu}_i}^{\rm r}$ follows from bounded distortion and contact scores.

For ${\mu_i^{\mathrm 0}}={\mu}_i^{\rm f},\allowbreak {\mu}_i^{\rm r},\allowbreak {\mu_i^{\mathrm d}}$ from~\eqref{eq:residual-sums},\allowbreak  let ${\calF}_{i-}$ be ${\calF}_{i-1}^{\rm f},\allowbreak {\calF}_{n-i}^{\rm r},\allowbreak {\calF}_{i-1}^{\rm f}$,\allowbreak  respectively and take ${\Delta_i^{\mathrm c}}=\Delta(Q_i^{\rm s}),\allowbreak \Delta(Q_i^{\rm s}),\allowbreak \Delta(Q_i^{\rm s})+{\bar\mu_i^{\mathrm d}}$,\allowbreak  respectively. For a sufficiently small positive real number $\bar a_0\in\bbR_+$,\allowbreak 
\begin{align}
\bbE[{{\mu_i^{\mathrm 0}}}^2e^{{\bar a_0}|{{\mu_i^{\mathrm 0}}}|}|{{\calF}}_{i-}]\leq {{C}}{{\Delta_i^{\mathrm c}}}.\label{eq:residual-exponential-moment}
\end{align}
Choose disjoint neighborhoods of the finitely many interior contacts so that each has a fixed nearest contact. By Lemma~\ref{lem:geometry}, the Bayes rule is locally fixed up to~\eqref{eq:action-equivalence}. Smoothness of the scores and resolvent, with ${a}_{Q_j}={\widetilde a}_{Q_j}=0$, leads to $|{\mu}_i^{\rm f}|+|{\mu}_i^{\rm r}|\leq {C}\|Q_i^{\rm s}-{\Pi}_{\calC}(Q_i^{\rm s})\|$ there. The quadratic bound~\eqref{eq:quadratic} therefore leads to
\begin{align}
\bbE_{P_n^{\rm o}}[({\mu}_i^{\rm f})^2e^{\bar a_0|{\mu}_i^{\rm f}|}|{\calF}_{i-1}^{\rm f}]
+\bbE_{P_n^{\rm o}}[({\mu}_i^{\rm r})^2e^{\bar a_0|{\mu}_i^{\rm r}|}|{\calF}_{n-i}^{\rm r}]
&\leq {C}\Delta(Q_i^{\rm s}).
\end{align}
Equivalent actions have zero distortion difference on the relevant channel support; nonequivalent actions have uniformly positive conditional Bayes regret in these neighborhoods. Bounded distortion thus bounds the regret residual's conditional exponential second moment by ${C}{\bar\mu_i^{\mathrm d}}$. On their compact complement, $\Delta$ has a positive minimum, so~\eqref{eq:posterior-log-moment} and the bounded reverse and regret residuals absorb all moments into ${C}\Delta(Q_i^{\rm s})$. This proves~\eqref{eq:residual-exponential-moment} for sufficiently small $\bar a_0\in(0,1)$.
For $|{{a_0}}|\leq {\bar a_0}$,
\begin{align}
\bbE[e^{{{a_0}}{{\mu_i^{\mathrm 0}}}}|{{\calF}}_{i-}]
&\leq1+{{a_0}}\bbE[{{\mu_i^{\mathrm 0}}}|{{\calF}}_{i-}]
+\tfrac12{{a_0}}^2\bbE[{{\mu_i^{\mathrm 0}}}^2e^{|{{a_0}}{{\mu_i^{\mathrm 0}}}|}|{{\calF}}_{i-}],\label{eq:v20-residual-mgf-1}\\
&\leq1+{{C}{a_0}^2}{{\Delta_i^{\mathrm c}}}\leq e^{{{C}{a_0}^2}{{\Delta_i^{\mathrm c}}}},\label{eq:v20-residual-mgf-2}
\end{align}
where~\eqref{eq:v20-residual-mgf-1} is Taylor's inequality;~\eqref{eq:v20-residual-mgf-2} follows from conditional centering in~\eqref{eq:localization-residual-centering}, the preceding exponential second-moment bound, and $1+x\leq e^x$.
Thus, for $|{{a_0}}|\leq {\bar a_0}$,
\begin{align}
{\bbE[e^{{a_0}{\mu}_i^{\rm f}}|{\calF}_{i-1}^{\rm f}]}&{\leq e^{{C}{a_0}^2\Delta(Q_i^{\rm s})},\qquad
\bbE[e^{{a_0}{\mu}_i^{\rm r}}|{\calF}_{n-i}^{\rm r}]\leq e^{{C}{a_0}^2\Delta(Q_i^{\rm s})},}\label{eq:mgf}\\
{\bbE[e^{{a_0}{\mu_i^{\mathrm d}}}|{\calF}_{i-1}^{\rm f}]}&{\leq e^{{C}{a_0}^2(\Delta(Q_i^{\rm s})+{\bar\mu_i^{\mathrm d}})},}
\end{align}
where~\eqref{eq:v20-residual-mgf-2} is applied to the forward, reverse, and regret residuals, respectively.
Each compensator is predictable. For ${\overline{{\mu}}_k^{\rm f}}:=\sum_{i\in[k]}{{\mu}_i}^{\rm f}$, ${\Delta_k^{\rm p}}:=\sum_{i\in[k]}\Delta({Q_i^{\rm s}})$, ${{a_0}}\in(0,{\bar a_0}]$, any positive real number ${{b_0}}\in\bbR_+$, and ${{s_{\pm}}}\in\{-1,1\}$, we have
\begin{align}
\bbE[\exp\{{{s_{\pm}}{a_0}}{\overline{{\mu}}_k^{\rm f}}-{{C}{a_0}^2}{\Delta_k^{\rm p}}\}]&\leq1,\label{eq:v20-residual-tail-1}\\
\Pr\{{{s_{\pm}}}{\overline{{\mu}}_k^{\rm f}}>{{C}{a_0}}{\Delta_k^{\rm p}}+{{b_0}}/{{a_0}}\}&\leq e^{-{{b_0}}},\label{eq:v20-residual-tail-2}
\end{align}
where~\eqref{eq:v20-residual-tail-1} follows by iterating the first bound in~\eqref{eq:mgf}, and~\eqref{eq:v20-residual-tail-2} follows by Markov's inequality.
A union bound over the two signs leads to $\Pr\{|\overline{{\mu}}_k^{\rm f}|>{C}{a_0}\Delta_k^{\rm p}+{b_0}/{a_0}\}\leq2e^{-{b_0}}$.
Recall the retained block distribution $P_n^{\rm b}$ in~\eqref{eq:source-conditioned-distribution}. Define
\begin{align}
{\mu}_{n,k}:=\sum_{i\in[k]}({\mu}_i^{\rm f}+{\mu}_i^{\rm r}),\qquad {\mu_{n,k}^{\mathrm d}}:=\sum_{i\in[k]}{\mu_i^{\mathrm d}}.\label{eq:prefix-residuals}
\end{align}
For $(k,{a_0},{b_0})\in[n]\times(0,\bar a_0]\times(0,\infty)$ and $\nu\in\{P_n^{\rm o},P_n^{\rm b}\}$,
\begin{align}
\Pr_\nu\{|{\mu}_{n,k}|+|{\mu_{n,k}^{\mathrm d}}|>{C}{a_0}\Delta_{n,k}+{C}{b_0}/{a_0}\}&\leq6e^{-{b_0}},\label{eq:residualexp}\\
\bbE_\nu[{\mu}_{n,k}^2]+\bbE_\nu[{(\mu_{n,k}^{\mathrm d})^2}]&\leq {C}\bbE_\nu[\Delta_{n,k}].\label{eq:residual-second-moment}
\end{align}
The tail bound combines the six signed forward, reverse, and regret tails, using $\Delta_k^{\rm p}\leq \Delta_{n,k}$. In reverse order, zeroing increments outside a deterministic prefix preserves the martingale property. Orthogonality and Cauchy--Schwarz lead to the second-moment bound. Here $\Pr_\nu$ and $\bbE_\nu$ use distribution $\nu$.
Since $S$ already contains $B$, its posteriors under $P_n^{\rm o}$ and $P_n^{\rm b}$ agree on $B=1$, and the channel remains $P_{Y|X}^n$; both bounds therefore hold under either distribution, with likelihoods normalized under $P_n^{\rm o}$. This compares the same flagged description: appending $B$ can change the original description's posteriors.

\subsection{Localization and Empirical Moments}
All limits below hold for every fixed $\eta>0$. For the current message alphabet $\calM_n$, $H_n=-\log P_n^{\rm o}(S|\tilY^n,{\Upsilon})$ leads to, for ${b_0}>0$,
\begin{align}
\Pr\{H_n>\log|\calM_n|+{b_0}\}
&=\bbE\Bigl[\sum_{s\in\calM_n}P_n^{\rm o}(s|\tilY^n,{\Upsilon})
\bbI\{P_n^{\rm o}(s|\tilY^n,{\Upsilon})<e^{-{b_0}}/|\calM_n|\}\Bigr],\label{eq:message-spectrum}\\
&\leq e^{-{b_0}},\label{eq:v20-message-spectrum-2}\\
\Pr\{H_n>\log(2M)+{b_0}\}&\leq e^{-{b_0}}.\label{eq:spectrum}
\end{align}
The first equality conditions on $(\tilY^n,{\Upsilon})$; each of the $|\calM_n|$ summands is at most $e^{-{b_0}}/|\calM_n|$, and $|\calM_n|\leq2M$ leads to the last inequality. A flag of probability bounded below preserves the vanishing tail.

The three normalized scores satisfy
\begin{align}
\lim_{{{A}}\to\infty}\limsup_{n\to\infty}\Pr\{\max(|{\Lambda_n}|,|J_n|,|G_n|)>{{A}}\}&=0,\label{eq:score-tight-tails}
\end{align}
where~\eqref{eq:score-tight-tails} follows from uniformly bounded second moments of the iid source score and uniformly bounded increments of the reverse channel martingale. On ${\calE_n}$, outside an event of vanishing probability,
\begin{align}
(1-{{C}}n^{-1/4}){\Delta_n}
\leq\sqrt n\big(L-{\Lambda_n}-J_n-{\lambda^*}G_n+o(1)\big)+{{C}}n^{1/4}\log n,\label{eq:v20-gap-path-1}
\end{align}
where~\eqref{eq:v20-gap-path-1} follows from~\eqref{eq:rateconverse},~\eqref{eq:combined}, and~\eqref{eq:spectrum}: the distortion term is nonpositive on ${\calE_n}$, and $|{{\mu}_n}|+|{{\mu_n^{\mathrm d}}}|\leq {{C}}n^{-1/4}{\Delta_n}+{{C}}n^{1/4}\log n$ holds outside an event of vanishing probability by~\eqref{eq:residualexp} with ${{a_0}}=n^{-1/4}$ and ${{b_0}}$ proportional to $\log n$.
Consequently,
\begin{align}
\lim_{{{A}}\to\infty}\limsup_{n\to\infty} \Pr\Bigl\{{\calE_n}\cap\bigl\{{\Delta_n}>{{A}}\sqrt n\bigr\}\Bigr\}&=0,\label{eq:v20-gap-residual-limits-1}\\
\Pr\Bigl\{{\calE_n}\cap\bigl\{\bigl|{{\mu}_n}\bigr|>\eta\sqrt n\bigr\}\Bigr\}&\longrightarrow0,\label{eq:v20-gap-residual-limits-2}\\
\Pr\Bigl\{{\calE_n}\cap\bigl\{\bigl|{{\mu_n^{\mathrm d}}}\bigr|>\eta\sqrt n\bigr\}\Bigr\}&\longrightarrow0,\label{eq:v20-gap-residual-limits-3}
\end{align}
where~\eqref{eq:v20-gap-residual-limits-1} follows from~\eqref{eq:v20-gap-path-1} and~\eqref{eq:score-tight-tails};~\eqref{eq:v20-gap-residual-limits-2} and~\eqref{eq:v20-gap-residual-limits-3} substitute this tight gap bound into~\eqref{eq:residualexp} with ${{a_0}}=n^{-1/4}$ and ${{b_0}}$ proportional to $\log n$.
On the message-likelihood event with tail parameter $2\log n$,~\eqref{eq:combined} leads to
\begin{align}
\lambda^*\Bigl(\sum_{i\in[n]}d_i-nD\Bigr)
&=nR_{\mathrm{WZ}}(D)-H_n
+\sqrt n(\Lambda_n+J_n+\lambda^*G_n)+\Delta_n+{\mu}_n+{\mu_n^{\mathrm d}}\\
&\geq\sqrt n(\Lambda_n+J_n+\lambda^*G_n-L+o(1))+{\mu}_n+{\mu_n^{\mathrm d}}.
\end{align}
The lower bound uses~\eqref{eq:rateconverse},~\eqref{eq:spectrum}, and $\Delta_n\geq0$.
On $\calE_n$, the left side is also at most zero; score tightness and the residual limits therefore bound its magnitude at the $\sqrt n$ scale.
The distortion deviation satisfies
\begin{align}
\lim_{{{A}}\to\infty}\limsup_{n\to\infty} \Pr\Bigl\{{\calE_n}\cap\bigl\{\bigl|\sum\nolimits_{i\in[n]} d_i-nD\bigr|>{{A}}\sqrt n\bigr\}\Bigr\}&=0,\label{eq:v20-distortion-tight-1}
\end{align}
For every deterministic $k=k(n)\in[n]$,
\begin{align}
\Pr\Bigl\{{\calE_n}\cap\bigl\{\bigl|{{\mu}_{n,k}}\bigr|>\eta\sqrt n\bigr\}\Bigr\}&\longrightarrow0,\label{eq:prefix-residual-rate}\\
\Pr\Bigl\{{\calE_n}\cap\bigl\{\bigl|{{\mu_{n,k}^{\mathrm d}}}\bigr|>\eta\sqrt n\bigr\}\Bigr\}&\longrightarrow0,\label{eq:v20-prefix-residual-1}
\end{align}
where~\eqref{eq:prefix-residual-rate} and~\eqref{eq:v20-prefix-residual-1} follow from~\eqref{eq:residualexp} and ${\Delta_{n,k}}\leq {\Delta_n}$, with the same choices of ${{a_0}},{{b_0}}$ as above.

Summing~\eqref{eq:quadratic}, using $\Delta_n^{\rm p}\leq \Delta_n$, and then~\eqref{eq:v20-gap-residual-limits-1} yields
\begin{align}
\frac1n\sum_{i\in[n]}\dist(Q_i^{\rm s},\calC)^2
&\leq\frac{{C}}{n}\sum_{i\in[n]}\Delta(Q_i^{\rm s})\leq\frac{{C}\Delta_n}n,\label{eq:gap-to-distance}\\
\lim_{{A}\to\infty}\limsup_{n\to\infty} \Pr\Bigl\{\calE_n\cap\bigl\{\sum\nolimits_{i\in[n]}\dist(Q_i^{\rm s},\calC)^2>{A}\sqrt n\bigr\}\Bigr\}&=0.
\end{align}
Let $\be_x$ denote the unit vector indexed by $x\in\calX$. Recall the definitions of $\calC$, ${{\Pi}_{\calC}}({Q})$, and $\theta_{n,j}$ in~\eqref{eq:zero-gap-posteriors},~\eqref{eq:nearest-contact}, and~\eqref{eq:empirical-weights}, respectively. Fix any bounded measurable ${\boldsymbol\psi}:\calP(\calX)\to\bbR^{|\calX|}$. Then
\begin{align}
\lim_{{{A}}\to\infty}\limsup_{n\to\infty} \Pr\Bigl\{\bigl|\sum_{i\in[n]}\langle {{\boldsymbol\psi}}({Q_i^{\rm s}}),\be_{\tilX_i}-{Q_i^{\rm s}}\rangle\bigr|>{{A}}\sqrt n\Bigr\}&=0,\label{eq:moment}
\end{align}
where~\eqref{eq:moment} follows by applying~\eqref{eq:twofiltration} to $\langle{\boldsymbol\psi}(Q),\be_X-Q\rangle$: its two martingale sums have uniformly bounded increments and second moments $O(n)$.
Cauchy--Schwarz,~\eqref{eq:quadratic}, and~\eqref{eq:v20-gap-residual-limits-1} similarly lead to
\begin{align}
\frac1n\sum_{i\in[n]}\|Q_i^{\rm s}-{\Pi}_{\calC}(Q_i^{\rm s})\|
&\leq\Bigl(\tfrac{{C}}{n}\sum_{i\in[n]}\Delta(Q_i^{\rm s})\Bigr)^{1/2},\label{eq:nearestmean}\\
\lim_{{A}\to\infty}\limsup_{n\to\infty} \Pr\Bigl\{\calE_n\cap\bigl\{\frac1n\sum\nolimits_{i\in[n]}\|Q_i^{\rm s}-{\Pi}_{\calC}(Q_i^{\rm s})\|>{A}n^{-1/4}\bigr\}\Bigr\}&=0.\label{eq:nearestmean-tail}
\end{align}
The empirical contact mixture satisfies
\begin{align}
\Pr\Bigl\{{\calE_n}\cap\bigl\{\bigl\|\sum\nolimits_{j\in[m]} \theta_{n,j}Q_j-P_X\bigr\|>\eta\bigr\}\Bigr\}&\longrightarrow0,\label{eq:v20-weight-mean-1}
\end{align}
where~\eqref{eq:v20-weight-mean-1} follows from~\eqref{eq:moment} applied coordinatewise, convergence in probability of the empirical source frequencies to $P_X$, and~\eqref{eq:nearestmean-tail}.
The Bayes distortions satisfy
\begin{align}
\lim_{{{A}}\to\infty}\limsup_{n\to\infty} \Pr\Bigl\{{\calE_n}\cap\bigl\{\bigl|\sum\nolimits_{i\in[n]} d_i^*-nD\bigr|>{{A}}\sqrt n\bigr\}\Bigr\}&=0,\label{eq:v20-bayes-distortion-1}
\end{align}
where~\eqref{eq:v20-bayes-distortion-1} follows from the distortion tail bound above and $\lambda^*\sum_{i\in[n]}(d_i-d_i^*)=\allowbreak\sum_{i\in[n]}{\bar\mu_i^{\mathrm d}}+{\mu_n^{\mathrm d}}$, as $\sum_{i\in[n]}{\bar\mu_i^{\mathrm d}}\leq \Delta_n$.
Bayes optimality leads to $\delta(Q)=\langle Q,\bar{\boldsymbol d}_{g_Q}\rangle$ for the distortion vector~\eqref{eq:rule-distortion-vector}. Thus,
\begin{align}
\sum_{i\in[n]}d_i^*-\sum_{i\in[n]}\delta(Q_i^{\rm s})
&=\sum_{i\in[n]}[d_i^*-\bar d_{g_{Q_i^{\rm s}}}(\tilX_i)]
+\sum_{i\in[n]}\langle\bar{\boldsymbol d}_{g_{Q_i^{\rm s}}},\be_{\tilX_i}-Q_i^{\rm s}\rangle.
\end{align}
The first sum is a reverse channel martingale with bounded increments,
because $g_{Q_i^{\rm s}}$ is measurable before $\tilY_i$ is revealed and $\tilY_i$ has conditional distribution $P_{Y|X}(\cdot|\tilX_i)$.
The second sum satisfies~\eqref{eq:moment} with ${\boldsymbol\psi}(Q)=\bar{\boldsymbol d}_{g_Q}$.
Combining these bounds with~\eqref{eq:v20-bayes-distortion-1} leads to
\begin{align}
\lim_{{{A}}\to\infty}\limsup_{n\to\infty} \Pr\Bigl\{{\calE_n}\cap\bigl\{\bigl|\sum\nolimits_{i\in[n]}\delta({Q_i^{\rm s}})-nD\bigr|>{{A}}\sqrt n\bigr\}\Bigr\}&=0,\label{eq:v20-posterior-distortion-1}
\end{align}
Hence,
\begin{align}
\Pr\Bigl\{{\calE_n}\cap\bigl\{\bigl|\sum\nolimits_{j\in[m]} \theta_{n,j}\delta(Q_j)-D\bigr|>\eta\bigr\}\Bigr\}&\longrightarrow0,\label{eq:v20-weight-distortion-1}
\end{align}
where~\eqref{eq:v20-weight-distortion-1} follows from the preceding posterior distortion bound, Lipschitz continuity of $\delta$, and~\eqref{eq:nearestmean-tail}.
Hence
\begin{align}
\Pr\Bigl\{{\calE_n}\cap\bigl\{\dist({\btheta_n},\calA)>\eta\bigr\}\Bigr\}\longrightarrow0,\label{eq:weights-feasible}
\end{align}
where~\eqref{eq:weights-feasible} follows from the preceding approximate constraints and compactness of the weight simplex.

Recall the definitions of ${\Sigma_j}(x)$, $\Sigma(\btheta)$, and ${\Omega_n}$ in~\eqref{eq:contact-covariance},~\eqref{eq:mixture-covariance}, and~\eqref{eq:predictable-covariance}, respectively. The predictable covariance satisfies
\begin{align}
\Pr\Bigl\{\bigl\|\frac1n{\Omega_n}-\frac1n\sum_{i\in[n]}\sum_{x\in\calX}{Q_i^{\rm s}}(x){\Sigma}_{{{\Pi}_{\calC}}({Q_i^{\rm s}})}(x)\bigr\|_2>\eta\Bigr\}&\longrightarrow0,\label{eq:v20-localized-covariance-1}\\
\Pr\Bigl\{{\calE_n}\cap\bigl\{\bigl\|{\Omega_n}/n-\Sigma({\btheta_n})\bigr\|_2>\eta\bigr\}\Bigr\}&\longrightarrow0,\label{eq:v20-localized-covariance-2}
\end{align}
where~\eqref{eq:v20-localized-covariance-1} follows by applying~\eqref{eq:moment} entrywise to the bounded matrix scores;~\eqref{eq:v20-localized-covariance-2} then follows from~\eqref{eq:nearestmean-tail} and the definition of $\Sigma({\btheta_n})$.
Combining~\eqref{eq:weights-feasible} and~\eqref{eq:v20-localized-covariance-2} establishes the localization of the weights and Claim~(ii) of Lemma~\ref{lem:localization}.

\subsection{The Non-Excess-Distortion Probability Bound}
Fix $\eta>0$ and $\gamma=2\log n$. Since $|\calM_n|\leq2M$,~\eqref{eq:rateconverse} implies, for all sufficiently large $n$,
$\log|\calM_n|+\gamma\leq nR_{\mathrm{WZ}}(D)+\sqrt n(L+\eta/2)$.
By~\eqref{eq:exactH}, a rate violation on this likelihood event forces ${\mu}_n<-\eta\sqrt n/2$. Thus,
\begin{align}
\Pr\{\calE_n\cap\calR_n(\tilde t_n;\eta)^c\}
&\leq n^{-2}+\Pr\{\calE_n\cap\{|{\mu}_n|>\eta\sqrt n/2\}\}
\longrightarrow0,
\end{align}
using~\eqref{eq:v20-message-spectrum-2} and~\eqref{eq:v20-gap-residual-limits-2}.
Similarly,~\eqref{eq:exactd} and ${\sum_{i\in[n]}} d_i\leq nD$ on $\calE_n$ show that a distortion violation forces ${\mu_n^{\mathrm d}}<-\lambda^*\eta\sqrt n$, so
\begin{align}
\Pr\left\{\calE_n\cap\calD_n\left(\tilde t_n-\frac{\Delta_n}{\lambda^*\sqrt n};\eta\right)^c\right\}
&\leq\Pr\{\calE_n\cap\{|{\mu_n^{\mathrm d}}|>\lambda^*\eta\sqrt n\}\}
\longrightarrow0
\end{align}
by~\eqref{eq:v20-gap-residual-limits-3}.
These prove~\eqref{eq:initial-information-violation}--\eqref{eq:initial-distortion-violation}.
The success probability is therefore bounded by
\begin{align}
\Pcn(D)
\leq \Pr\Bigl\{J_n\leq L-{\Lambda_n}-{\lambda^*} {\tilde t_n}+\eta,\quad
G_n\leq {\tilde t_n}-\tfrac{{\Delta_n}}{{\lambda^*}\sqrt n}+\eta\Bigr\}+o(1),
\end{align}
by a union bound over the two violations. Together with~\eqref{eq:weights-feasible} and~\eqref{eq:v20-localized-covariance-2}, this proves Lemma~\ref{lem:localization}.

\Needspace{10\baselineskip}
\section{Proof of Lemma~\ref{lem:selectors}}
\label{app:selectors}
Recall the original success event $\calE_n$, flagged description $S$, retained distribution $P_n^{\rm b}$, and weight set $\calA$ from~\eqref{eq:success-event},~\eqref{eq:flagged-description},~\eqref{eq:flagged-law}, and~\eqref{eq:weightset}. We bound coordinate dependence by the expected posterior gap. For the source-only repair, define
\begin{align}
k_n:=\lceil n^{3/4}\rceil,\qquad S':=(S,\tilX^{k_n}),\label{eq:repaired-description}
\end{align}
We then transfer source-measurable weights and allocation back to $S$. 

\subsection{Allocation and Coordinate Dependence Bounds}
\label{app:projection}
For $i\in[n]$, let $\tilY_{-i}$ contain all observations except $\tilY_i$. For every $y\in\calY$, define
\begin{align}
{\tilde Q_i^{\rm s}}:={P^{\rm o}_{n,\tilX_i|S,\tilY_{-i},{\Upsilon}}},\qquad
{\tilde Q_{i,Y}^{\rm s}}(y):=\sum_{x\in\calX}{\tilde Q_i^{\rm s}}(x)P_{Y|X}(y|x).\label{eq:leave-one-out-posterior}
\end{align}
For a finite random variable ${{\boldsymbol W}}$ and any joint distribution ${\nu}$ of $(X^n,{{\boldsymbol W}})$, define the averaged conditional total correlation by
\begin{align}
\TC_{{\nu}}(X^n|{{\boldsymbol W}})
&:=\sum_{i\in[n]}H_{{\nu}}(X_i|{{\boldsymbol W}})-H_{{\nu}}(X^n|{{\boldsymbol W}})\label{eq:TC-definition}\\
&=\bbE_{{\nu}}\Bigl[{\operatorname{KL}}\bigl({\nu}_{X^n|{{\boldsymbol W}}}(\cdot|{{\boldsymbol W}})
\,\bigm\|\,\prod_{i=1}^n{\nu}_{X_i|{{\boldsymbol W}}}(\cdot|{{\boldsymbol W}})\bigr)\Bigr]\geq0,
\end{align}
where conditional relative entropy leads to the equality and nonnegativity.
Recall $H_n,\Delta_n,\tilde t_n$ from~\eqref{eq:message-likelihood},~\eqref{eq:accumulated-gap}, and~\eqref{eq:code-allocation}. The sequential $Q_i^{\rm s}$ conditions on $(S,\tilX^{i-1},\tilY_{i+1}^n,{\Upsilon})$. Here $P_n^{\rm b}=P_n^{\rm o}$ or $P_n^{\rm o}\{\,\cdot|B=1\}$ for a source-only flag $B$ included in $S$. Define
\begin{align}
{{I_n^{\mathrm t}}}&:=\TC_{{P_n^{\rm b}}}(\tilX^n|{S},\tilY^n,{{\Upsilon}})\label{eq:code-total-correlation}\\
t_n^{{S}}(\tilX^n,{{\Upsilon}})&:=\frac{{\bbE_{P_n^{\rm o}}}[H_n|\tilX^n,{{\Upsilon}}]-nR_{\mathrm{WZ}}(D)-\sum_{i\in[n]} \big(\jW{\tilX_i}-R_{\mathrm{WZ}}(D)\big)}{{\lambda^*}\sqrt n},\label{eq:scalar-selector-definition}
\end{align}
Here ${I_n^{\mathrm t}}$ is deterministic and $t_n^{S}$ is source-measurable. On the flag, source-only conditioning preserves the conditional expectation of $H_n$ given $(\tilX^n,{\Upsilon})$, so it agrees under $P_n^{\rm o}$ and $P_n^{\rm b}$.

\begin{lemma}\label{lem:tc-deficit}
For the preceding distributions, constants ${C}_0,{C}>0$ depending only on the fixed model and distortion level satisfy the following.
\begin{enumerate}
\item The remaining coordinate dependence satisfies
\begin{align}
\bbE_{{P_n^{\rm b}}}[{\Delta_n}]
&={{I_n^{\mathrm t}}}+\sum_{i\in[n]}\bbE_{{P_n^{\rm b}}}[\Delta({\tilde Q_i^{\rm s}})]
+{\lambda^*}\Bigl(\bbE_{{P_n^{\rm b}}}[\sum_{i\in[n]} d_i]-\sum_{i\in[n]}\bbE_{{P_n^{\rm b}}}[\delta({\tilde Q_i^{\rm s}})]\Bigr)\geq {{I_n^{\mathrm t}}}.\label{eq:tc-deficit}
\end{align}
The identity also holds for any source prior with the same product channel, using its posteriors and entropies.
\item The posterior replacement error satisfies
\begin{align}
\bbE_{{P_n^{\rm b}}}[{\Delta_n}]&\geq {{C}_0}\sum_{i\in[n]}\bbE_{{P_n^{\rm b}}}[\|{Q_i^{\rm s}}-{\tilde Q_i^{\rm s}}\|^2],\label{eq:post-stability}
\end{align}
\item The allocation projection satisfies
\begin{align}
\bbE_{{P_n^{\rm b}}}[({\tilde t_n}-t_n^{{S}})^2]&\leq {{C}}\bbE_{{P_n^{\rm b}}}[{\Delta_n}]/n.\label{eq:scalarprojection}
\end{align}
\end{enumerate}
\end{lemma}

\subsubsection{Proof of Lemma~\ref{lem:tc-deficit}}
Condition on ${\Upsilon}$ and average afterward. Use $P_n^{\rm b}=P_n^{\rm o}\{\,\cdot|B=1\}$; $B\equiv1$ includes $P_n^{\rm b}=P_n^{\rm o}$. Then
\begin{align}
\bbE_{{P_n^{\rm b}}}[\rho({Q_i^{\rm s}})]&=-H_{{P_n^{\rm b}}}(\tilX_i|{S},\tilX^{i-1},\tilY^n,{{\Upsilon}})\label{eq:v20-posterior-tower-1}\\
\bbE_{{P_n^{\rm b}}}[\rho({\tilde Q_i^{\rm s}})]&=-H_{{P_n^{\rm b}}}(\tilX_i|{S},\tilY^n,{{\Upsilon}})\label{eq:v20-posterior-tower-2}\\
\bbE_{{P_n^{\rm b}}}[{Q_i^{\rm s}}|{S},\tilY_{-i},{{\Upsilon}}]&={\tilde Q_i^{\rm s}},\label{eq:v20-posterior-tower-3}
\end{align}
The entropy identities use $\rho(Q)=-H_{QP_{Y|X}}(X|Y)$. For the tower identity, fix $x\in\calX$. The product channel makes past observations irrelevant given $\tilX^{i-1}$, so
\begin{align}
{Q_i^{\rm s}}(x)&={P_n^{\rm b}}\{\tilX_i=x|{S},\tilX^{i-1},\tilY_{-i},{{\Upsilon}}\}\label{eq:posterior-enlarged-field}\\
\bbE_{{P_n^{\rm b}}}[{Q_i^{\rm s}}(x)|{S},\tilY_{-i},{{\Upsilon}}]
&{=\bbE_{P_n^{\rm b}}\big[P_n^{\rm b}\{\tilX_i=x|S,\tilX^{i-1},\tilY_{-i},{\Upsilon}\}\bigm|S,\tilY_{-i},{\Upsilon}\big]}\\
&={P_n^{\rm b}}\{\tilX_i=x|{S},\tilY_{-i},{{\Upsilon}}\}={\tilde Q_i^{\rm s}}(x),\label{eq:posterior-nested-tower}
\end{align}
The product channel leads to~\eqref{eq:posterior-enlarged-field}; the nested tower property leads to~\eqref{eq:posterior-nested-tower}. These posteriors agree under $P_n^{\rm o},P_n^{\rm b}$ on $B=1$ because $B$ is included in $S$; see Appendix~\ref{app:localization}.
The accumulated gap satisfies
\begin{align}
\bbE_{{P_n^{\rm b}}}[{\Delta_n}]
&=\sum_{i\in[n]}\bbE_{{P_n^{\rm b}}}[\rho({Q_i^{\rm s}})-{\beta_0}-\langle {\bbeta},{Q_i^{\rm s}}\rangle]+{\lambda^*}\bbE_{{P_n^{\rm b}}}\bigl[\sum_{i\in[n]} d_i\bigr]\label{eq:v20-mean-gap-1}\\
&=\sum_{i\in[n]}\bbE_{{P_n^{\rm b}}}[\rho({Q_i^{\rm s}})-\rho({\tilde Q_i^{\rm s}})]+\sum_{i\in[n]}\bbE_{{P_n^{\rm b}}}[\Delta({\tilde Q_i^{\rm s}})]
+{\lambda^*}\Bigl(\bbE_{{P_n^{\rm b}}}\bigl[\sum_{i\in[n]} d_i\bigr]-\sum_{i\in[n]}\bbE_{{P_n^{\rm b}}}[\delta({\tilde Q_i^{\rm s}})]\Bigr)\label{eq:v20-mean-gap-2}\\
&={{I_n^{\mathrm t}}}+\sum_{i\in[n]}\bbE_{{P_n^{\rm b}}}[\Delta({\tilde Q_i^{\rm s}})]
+{\lambda^*}\Bigl(\bbE_{{P_n^{\rm b}}}\bigl[\sum_{i\in[n]} d_i\bigr]-\sum_{i\in[n]}\bbE_{{P_n^{\rm b}}}[\delta({\tilde Q_i^{\rm s}})]\Bigr)\geq {{I_n^{\mathrm t}}},\label{eq:v20-mean-gap-3}
\end{align}
The first equality uses $\bbE_{P_n^{\rm b}}[{\bar\mu_i^{\mathrm d}}]=\lambda^*(\bbE_{P_n^{\rm b}}[d_i]-\bbE_{P_n^{\rm b}}[\delta(Q_i^{\rm s})])$ and~\eqref{eq:deficit}; the second uses $\bbE_{P_n^{\rm b}}[Q_i^{\rm s}]=\bbE_{P_n^{\rm b}}[\tilde Q_i^{\rm s}]$; the third uses the entropy chain rule. The remaining terms are nonnegative by $\Delta\geq0$ and Bayes optimality. This proves Claim~(i) for every source prior. Here $P_n^{\rm b}$ denotes the block distribution.
The posterior replacement error satisfies
\begin{align}
\bbE_{{P_n^{\rm b}}}[{\Delta_n}]
&\geq\sum_{i\in[n]}\bbE_{{P_n^{\rm b}}}[\rho({Q_i^{\rm s}})-\rho({\tilde Q_i^{\rm s}})]\label{eq:v20-posterior-stability-1}\\
&\geq {{C}_0}\sum_{i\in[n]}\bbE_{{P_n^{\rm b}}}[\|{Q_i^{\rm s}}-{\tilde Q_i^{\rm s}}\|^2],\label{eq:v20-posterior-stability-2}
\end{align}
Dropping the nonnegative terms leads to the first inequality. Strong curvature~\eqref{eq:strongcurv} and~\eqref{eq:v20-posterior-tower-3} cancel the linear Taylor term and lead to the second; continuity covers boundary faces. This proves Claim~(ii). For Claim~(iii), compare reverse Doob increments with centered contact scores.

Fix $i\in[n]$, a possible message ${s}$, $y_{-i}\in\calY^{n-1}$, and a permutation ${{\upsilon}}$ of $[n]$. For $y\in\calY$ on the positive-probability support, the message likelihood satisfies
\begin{align}
H_n({s};y_{-i},y,{{\upsilon}})
=-\log \Pr\{{S}={s}|\tilY_{-i}=y_{-i},{{\Upsilon}}={{\upsilon}}\}
+\log P_Y(y)-\log{\tilde Q_{i,Y}^{\rm s}}(y).\label{eq:likelihoodcoordinate}
\end{align}
Bayes' rule under $P_n^{\rm o}$ leads to~\eqref{eq:likelihoodcoordinate}. For $Q\in\calP(\calX)$, write $Q_Y$ for its output distribution~\eqref{eq:posterior-output} and $\imath_Q^{\rm y}(y):=\log P_Y(y)-\log Q_Y(y)$ for $Q_Y(y)>0$. For $x\in\supp(Q)$ and $Q_Y(y)>0$, define
\begin{align}
\imath_Q^{\rm y,c}(x,y)&:=\imath_Q^{\rm y}(y)-\sum_{y'\in\calY:\,P_{Y|X}(y'|x)>0}P_{Y|X}(y'|x)\imath_Q^{\rm y}(y').
\label{eq:centered-output-score}
\end{align}
The sum omits zero-transition terms; on supported pairs, $\imath_{Q_j}^{\rm y,c}(x,y)=\imath_j^{\rm c}(x,y)$.
Recall the contact set and nearest-contact rule~\eqref{eq:zero-gap-posteriors} and~\eqref{eq:nearest-contact}, and $\underline p_W$ from~\eqref{eq:min-positive-entry}. Set ${\underline q}:=\min_{j\in[m],x\in\calX}Q_j(x)>0$ and ${\calF}_i^{\rm *}:={\calS}(S,\tilY_{-i},{\Upsilon})$. Conditionally on ${\calF}_i^{\rm *}$, $\bbE_{P_n^{\rm b}}[Q_i^{\rm s}|{\calF}_i^{\rm *}]=\tilde Q_i^{\rm s}$, $P^{\rm b}_{n,\tilX_i|Q_i^{\rm s},{\calF}_i^{\rm *}}=Q_i^{\rm s}$, and $P^{\rm b}_{n,\tilY_i|\tilX_i,Q_i^{\rm s},{\calF}_i^{\rm *}}=P_{Y|X}$.

If $\min_{x\in\calX}{\tilde Q_i^{\rm s}}(x)\geq {{\underline q}}/2$, then
\begin{align}
|{\imath_{\tilde Q_i^{\rm s}}^{\rm y,c}}(x,y)-\imath_{{{\Pi}_{\calC}}({Q_i^{\rm s}})}^{\rm c}(x,y)|^2
&\leq {{C}}\|{\tilde Q_i^{\rm s}}-{{\Pi}_{\calC}}({Q_i^{\rm s}})\|^2\label{eq:v20-score-lipschitz-1}\\
&\leq {{C}}\big[\|{\tilde Q_i^{\rm s}}-{Q_i^{\rm s}}\|^2+\Delta({Q_i^{\rm s}})\big],\label{eq:v20-score-lipschitz-2}
\end{align}
Supported output probabilities are at least ${\underline q}\underline p_W/2$, giving the first inequality; the triangle inequality and~\eqref{eq:quadratic} lead to the second.
Otherwise, choose $x_*\in\calX$ with ${\tilde Q_i^{\rm s}}(x_*)<{{\underline q}}/2$. Then
\begin{align}
\bbE_{{P_n^{\rm b}}}[\dist({Q_i^{\rm s}},\calC)|{{\calF}_i^{\rm *}}]
&\geq\bbE_{{P_n^{\rm b}}}[{{\Pi}_{\calC}}({Q_i^{\rm s}})(x_*)-{Q_i^{\rm s}}(x_*)|{{\calF}_i^{\rm *}}]\geq {{\underline q}}/2\label{eq:v20-away-contact-1}\\
\bbE_{{P_n^{\rm b}}}[\Delta({Q_i^{\rm s}})|{{\calF}_i^{\rm *}}]&\geq {{C}_{\rm -}} {{\underline q}}^2,\label{eq:v20-away-contact-2}
\end{align}
The coordinate bound and $\bbE_{P_n^{\rm b}}[Q_i^{\rm s}|{\calF}_i^{\rm *}]=\tilde Q_i^{\rm s}$ lead to the first inequality;~\eqref{eq:quadratic} and Jensen lead to the second.
Under this conditioning, $(\tilX_i,\tilY_i)\sim{\tilde Q_i^{\rm s}} P_{Y|X}$; hence
\begin{align}
\bbE_{{\tilde Q_i^{\rm s}} P_{Y|X}}[|{\imath_{\tilde Q_i^{\rm s}}^{\rm y,c}}(\tilX_i,\tilY_i)|^2]
&\leq\bbE_{{\tilde Q_{i,Y}^{\rm s}}}[|\log P_Y(\tilY_i)-\log {\tilde Q_{i,Y}^{\rm s}}(\tilY_i)|^2]\label{eq:v20-score-moment-1}\\
&\leq 2\max_{y\in\calY}|\log P_Y(y)|^2+
2\sum_{y\in\calY:\,{\tilde Q_{i,Y}^{\rm s}}(y)>0}{\tilde Q_{i,Y}^{\rm s}}(y)|\log {\tilde Q_{i,Y}^{\rm s}}(y)|^2
\leq {{C}},\label{eq:v20-score-moment-2}
\end{align}
Conditional variance is bounded by the second moment. Then $({v}_1-{v}_2)^2\leq2{v}_1^2+2{v}_2^2$ for real ${v}_1,{v}_2$, $P_Y>0$, and bounded ${p_0}|\log {p_0}|^2$ lead to the uniform bound.
In both cases,
\begin{align}
\bbE_{{P_n^{\rm b}}}[|{\imath_{\tilde Q_i^{\rm s}}^{\rm y,c}}(\tilX_i,\tilY_i)-\imath_{{{\Pi}_{\calC}}({Q_i^{\rm s}})}^{\rm c}(\tilX_i,\tilY_i)|^2]
\leq {{C}}\bbE_{{P_n^{\rm b}}}[\|{\tilde Q_i^{\rm s}}-{Q_i^{\rm s}}\|^2+\Delta({Q_i^{\rm s}})],\label{eq:averaged-score}
\end{align}
The first case uses~\eqref{eq:v20-score-lipschitz-2}; in the second,~\eqref{eq:v20-away-contact-2} absorbs the uniform moment and bounded contact scores. Recall $J_n,{\mu}_n$ from~\eqref{eq:information-block-score},~\eqref{eq:residual-sums}.
For $i\in[n]$, let ${{\mu_i^{\mathrm H}}}$ be the reverse Doob increment of $H_n$ after $(\tilX^n,{{\Upsilon}})$ is fixed. Its formula is
\begin{align}
{{\mu_i^{\mathrm H}}}
&=\bbE_{{P_n^{\rm b}}}[H_n|\tilX^n,{{\Upsilon}},\tilY_i^n]
-\bbE_{{P_n^{\rm b}}}[H_n|\tilX^n,{{\Upsilon}},\tilY_{i+1}^n]\label{eq:doobincrement}\\
&=\bbE_{{P_n^{\rm b}}}[{\imath_{\tilde Q_i^{\rm s}}^{\rm y,c}}(\tilX_i,\tilY_i)|\tilX^n,{{\Upsilon}},\tilY_i^n],\label{eq:v20-doob-increment-2}
\end{align}
Averaging~\eqref{eq:likelihoodcoordinate} over $\tilY^{i-1}$ cancels the $\tilY_i$-independent term; centering under $P_{Y|X}(\cdot|\tilX_i)$ leads to the second line. The contact score is measurable given $(\tilX^n,{{\Upsilon}},\tilY_i^n)$, and
\begin{align}
\bbE_{{P_n^{\rm b}}}\Bigl[\bigl(H_n-\bbE_{{P_n^{\rm b}}}[H_n|\tilX^n,{{\Upsilon}}]-\sqrt nJ_n\bigr)^2\Bigr]
&=\sum_{i\in[n]}\bbE_{{P_n^{\rm b}}}\Bigl[\bigl({{\mu_i^{\mathrm H}}}-\imath_{{{\Pi}_{\calC}}({Q_i^{\rm s}})}^{\rm c}(\tilX_i,\tilY_i)\bigr)^2\Bigr]\label{eq:doobcompare}\\
&\leq {{C}}\sum_{i\in[n]}\bbE_{{P_n^{\rm b}}}[\|{\tilde Q_i^{\rm s}}-{Q_i^{\rm s}}\|^2+\Delta({Q_i^{\rm s}})]
\leq {{C}}\bbE_{{P_n^{\rm b}}}[{\Delta_n}],\label{eq:v20-doob-2}
\end{align}
Martingale orthogonality leads to the equality. Conditional Jensen,~\eqref{eq:averaged-score},~\eqref{eq:post-stability}, and $\sum_{i\in[n]}\Delta(Q_i^{\rm s})\leq \Delta_n$ lead to the bound.
The source fixes ${B}$, so subtracting~\eqref{eq:scalar-selector-definition} from~\eqref{eq:exactH} leads to
\begin{align}
{\lambda^*}\sqrt n({\tilde t_n}-t_n^{{S}})
=H_n-\bbE_{{P_n^{\rm b}}}[H_n|\tilX^n,{{\Upsilon}}]-\sqrt nJ_n-{{\mu}_n},
\end{align}
The preceding bound and~\eqref{eq:residual-second-moment} control the right-hand side by ${C}\bbE_{P_n^{\rm b}}[\Delta_n]$ in mean square. Division by $(\lambda^*)^2n$ proves Claim~(iii) and completes Lemma~\ref{lem:tc-deficit}.

\subsection{Concentration of Covariance Weights}
\label{app:concentration}
We now control posterior averages by total correlation. Under a distribution $P_n^{\rm o}$ with arbitrary source prior, deterministic source-only $S$, and product channel $P_{Y|X}^n$ satisfying~\eqref{eq:overlap}, set $Q_i^{\rm s}=P^{\rm o}_{n,X_i|S,X^{i-1},Y_{i+1}^n}$. Let $\psi_i:\calX^n\times\calP(\calX)\to\bbR$ be ${C_\psi}$-Lipschitz in the posterior under $\|\cdot\|_1$, uniformly in $(i,x^n)$, for a fixed ${C_\psi}>0$.

\Needspace{6\baselineskip}
\begin{lemma}\label{lem:concentration}
There is a finite constant ${C}_{\rm ch}$, depending only on the
fixed side-information channel $P_{Y|X}$, with the following properties.
\begin{enumerate}
\item For a source-only description ${S}$,
\begin{align}
{\bbE_{P_n^{\rm o}}}\Bigl[{\var_{P_n^{\rm o}}}\bigl(\frac1n\sum_{i\in[n]}\psi_i(X^n,{Q_i^{\rm s}})\bigm|X^n\bigr)\Bigr]
\leq\frac{{{C}_{\rm ch}}{{C_\psi}}^2}{n}{\TC_{P_n^{\rm o}}}(X^n|{S},Y^n).
\end{align}
\item For a finite random variable ${\Upsilon}$ with $S$ determined by $(X^n,{\Upsilon})$ and ${P^{\rm o}_{n,Y^n|X^n,{\Upsilon}}}=P_{Y|X}^n$, set ${Q_i^{\rm s}}={P^{\rm o}_{n,X_i|S,X^{i-1},Y_{i+1}^n,{\Upsilon}}}$. Then
\begin{align}
{\bbE_{P_n^{\rm o}}}\Bigl[{\var_{P_n^{\rm o}}}\bigl(\frac1n\sum_{i\in[n]}\psi_i(X^n,{Q_i^{\rm s}})\bigm|X^n,{{\Upsilon}}\bigr)\Bigr]
\leq\frac{{{C}_{\rm ch}}{{C_\psi}}^2}{n}{\TC_{P_n^{\rm o}}}(X^n|{S},Y^n,{{\Upsilon}}).
\end{align}
Both sides average over ${{\Upsilon}}$ under ${P_n^{\rm o}}$; independence between $X^n$ and ${{\Upsilon}}$ is not required. The result holds coordinatewise for fixed-dimensional vector statistics.
\end{enumerate}
\end{lemma}

\subsubsection{Uniform Information Contraction and Single-Observation Influence}
For the overlap constant $\kappa_W>0$ in~\eqref{eq:overlap}, take interior one-letter distributions $Q^{(1)},Q^{(0)}$, with output distributions induced by $P_{Y|X}$. Put ${\dot Q}:=Q^{(1)}-Q^{(0)}$ and $Q^{(\rm m)}:=Q^{(0)}+{\vartheta_0}{\dot Q}$, ${\vartheta_0}\in[0,1]$. Negative entropy's integral remainder leads to
\begin{align}
{\operatorname{KL}}({Q^{(1)}}\|{Q^{(0)}})
&=\int_0^1(1-{{\vartheta_0}})\sum_{x\in\calX}\frac{{{\dot Q}}(x)^2}{{Q^{(\rm m)}}(x)}\,{d{\vartheta_0}}\label{eq:v20-kl-integral-1}\\
{\operatorname{KL}}({Q_Y^{(1)}}\|{Q_Y^{(0)}})
&=\int_0^1(1-{{\vartheta_0}})\sum_{y\in\calY}
\frac{(\sum_{x\in\calX}{{\dot Q}}(x)P_{Y|X}(y|x))^2}{{Q_Y^{(\rm m)}}(y)}\,{d{\vartheta_0}},\label{eq:v20-kl-integral-2}
\end{align}
Integrating~\eqref{eq:hessian}--\eqref{eq:strongcurv} in these formulas yields the strong data-processing bound~\cite{polyanskiywu}
\begin{align}
{\operatorname{KL}}({Q_Y^{(1)}}\|{Q_Y^{(0)}})\leq(1-\kappa_W){\operatorname{KL}}({Q^{(1)}}\|{Q^{(0)}}),\label{eq:uniform-kl}
\end{align}
Positive mixtures extend this bound to finite boundary divergences. Infinite input divergence is immediate if $\kappa_W<1$; if $\kappa_W=1$, all rows are the same point mass. Averaging over conditional input distributions for finite $U-X-Y$ leads to
\begin{align}
I({U};Y)\leq(1-\kappa_W)I({U};X),\label{eq:uniform-mi}
\end{align}
This also holds conditionally whenever the conditional channel remains $P_{Y|X}$.

Let $(X_1,X_2)$ have any joint distribution ${Q^{\rm 2}}\in\calP(\calX^2)$, and generate $Y_1,Y_2$ through conditionally independent copies of $P_{Y|X}$. The information retained after both observations satisfies
\begin{align}
I(X_1;X_2|Y_1,Y_2)
&=I(X_1;X_2|Y_1)-I(X_1;Y_2|Y_1)\label{eq:retained-information}\\
&\geq\kappa_W I(X_1;X_2|Y_1)\label{eq:v20-retained-info-2}\\
&=\kappa_W\big[I_{{Q^{\rm 2}}}(X_1;X_2)-I(Y_1;X_2)\big]\label{eq:v20-retained-info-3}\\
&\geq\kappa_W^2I_{{Q^{\rm 2}}}(X_1;X_2),\label{eq:v20-retained-info-4}
\end{align}
The equalities use the chain rule and $Y_1-X_1-X_2-Y_2$; the inequalities apply~\eqref{eq:uniform-mi} to $X_2\to Y_2$ conditionally and then to $X_1\to Y_1$.

Given $X^n$, resample $Y_j$ independently from $P_{Y|X}(\cdot|X_j)$ and evaluate the posterior as $Q_i^{\rm s,(j)}$. Then $Q_i^{\rm s,(j)}=Q_i^{\rm s}$ for $j\leq i$; for $i<j$,
\begin{align}
{\bbE_{P_n^{\rm o}}}[\|{Q_i^{\rm s}}-{Q_i^{\rm s,(j)}}\|_1^2]
\leq {{C}_{\rm ch}} I(X_i;X_j|{S},X^{i-1},Y^n).\label{eq:influence}
\end{align}
Set ${\boldsymbol W}_{ij}:=(S,X^{i-1},Y_{i+1}^{j-1},Y_{j+1}^n)$. For a supported ${\boldsymbol w}_{ij}$, fix $Q^{\rm 2}:=P^{\rm o}_{n,X_i,X_j|{\boldsymbol W}_{ij}={\boldsymbol w}_{ij}}$ with first marginal $Q^{\rm 2}_{X_i}$, and $Q^{\rm 2,y}:=Q^{\rm 2}_{X_i|Y_j}(\cdot|y)$ under the induced channel. Original and resampled observations have equal marginals given ${\boldsymbol W}_{ij}={\boldsymbol w}_{ij}$; expectations include the resampling. Thus,
\begin{align}
{\bbE_{P_n^{\rm o}}[\|Q_i^{\rm s}-Q_i^{\rm s,(j)}\|_1^2|{\boldsymbol W}_{ij}={\boldsymbol w}_{ij}]}
&\leq4{\bbE_{P_n^{\rm o}}[\|Q^{\rm 2,Y_j}-Q^{\rm 2}_{X_i}\|_1^2|{\boldsymbol W}_{ij}={\boldsymbol w}_{ij}]}\label{eq:prior-influence}\\
&\leq8{I_{Q^{\rm 2}}}(X_i;Y_j)
\leq8{I_{Q^{\rm 2}}}(X_i;X_j),\label{eq:v20-prior-influence-2}
\end{align}
The squared triangle inequality and equal marginals lead to the first bound; Pinsker and data processing lead to the second. Under this fixed distribution,~\eqref{eq:retained-information}--\eqref{eq:v20-retained-info-4} lead to $I_{Q^{\rm 2}}(X_i;X_j|Y_i,Y_j)\geq\kappa_W^2I_{Q^{\rm 2}}(X_i;X_j)$. Averaging over ${\boldsymbol W}_{ij}$ and discarding irrelevant $Y^{i-1}$ given $X^{i-1}$ prove~\eqref{eq:influence} with ${C}_{\rm ch}=8/\kappa_W^2$.

\subsubsection{Entropy Chain Rule and Conditional Efron--Stein}
To sum~\eqref{eq:influence}, expand conditional total correlation by the entropy chain rule:
\begin{align}
{\TC_{P_n^{\rm o}}}(X^n|{S},Y^n)
&=\sum_{j\in[n]}I(X_j;X^{j-1}|{S},Y^n)\label{eq:pair-sum}\\
&=\sum_{\substack{(i,j)\in[n]^2:\,i<j}}I(X_i;X_j|{S},X^{i-1},Y^n),\label{eq:v20-total-correlation-2}
\end{align}
The second line expands each mutual information along $X^{j-1}$. Set $\Psi_n:=n^{-1}\sum_{i\in[n]}\psi_i(X^n,Q_i^{\rm s})$ and define $\Psi_n^{(j)}$, $j\in[n]$, by replacing $Q_i^{\rm s}$ with $Q_i^{\rm s,(j)}$. Then
\begin{align}
{\bbE_{P_n^{\rm o}}}[{\var_{P_n^{\rm o}}}(\Psi_n|X^n)]
&\leq\frac12\sum_{j\in[n]}{\bbE_{P_n^{\rm o}}}[(\Psi_n-\Psi_n^{(j)})^2]\label{eq:v20-efron-stein-1}\\
&\leq\frac{{{C_\psi}}^2}{2n^2}\sum_{j\in[n]}(j-1)\sum_{i=1}^{j-1}{\bbE_{P_n^{\rm o}}}[\|{Q_i^{\rm s}}-{Q_i^{\rm s,(j)}}\|_1^2]\label{eq:v20-efron-stein-2}\\
&\leq\frac{{{C}_{\rm ch}}{{C_\psi}}^2}{2n}\sum_{\substack{(i,j)\in[n]^2:\,i<j}}I(X_i;X_j|{S},X^{i-1},Y^n)\label{eq:v20-efron-stein-3}\\
&=\frac{{{C}_{\rm ch}}{{C_\psi}}^2}{2n}{\TC_{P_n^{\rm o}}}(X^n|{S},Y^n),\label{eq:v20-efron-stein-4}
\end{align}
Given $X^n$, observations are independent and $S$ is fixed, so conditional Efron--Stein~\cite{steele} leads to the first line. Lipschitz continuity, Cauchy--Schwarz, and $Q_i^{\rm s,(j)}=Q_i^{\rm s}$ for $i\geq j$ lead to the second;~\eqref{eq:influence}, $j-1\leq n$, and~\eqref{eq:pair-sum} lead to the last two. Dropping $1/2$ proves Claim~(i). Given ${\Upsilon}={\upsilon}$, the source prior is arbitrary and the channel remains $P_{Y|X}^n$; averaging the same bound over ${\Upsilon}$ proves Claim~(ii). Coordinatewise application leads to the vector version and completes Lemma~\ref{lem:concentration}.

For the converse-code application under $P_n^{\rm b}$, the generic $(X^n,Y^n)$ in Lemma~\ref{lem:concentration} are the permuted $(\tilX^n,\tilY^n)$; Claim~(ii) conditions on ${\Upsilon}$ and then averages.

\subsection{Source-Only Repair and Transfer to the Original Code}
To obtain the $o(n)$ expected gap needed above, work under $P_n^{\rm b}$ in original coordinates~\eqref{eq:permutation-convention}. Every retained source has conditional success probability at least $n^{-1}$. Set
\begin{align}
{{\bar D_n^{\mathrm H}}}:=\lceil n^{9/16}\rceil,\qquad {{\eta_n^{\mathrm d}}}:=n^{-1/16}.\label{eq:repair-scales}
\end{align}
For $(s,y^n)\in[M]\times\calY^n$, define the nonempty candidate list
\begin{align}
\calL_n(s,y^n)
&:=\left\{\phi_n(s,y^{\prime n}):y^{\prime n}\in\calY^n,\;
\sum_{i\in[n]}\bbI\{y_i^{\prime}\ne y_i\}\leq {\bar D_n^{\mathrm H}}\right\}.
\label{eq:repair-candidate-list}
\end{align}
The list precedes the independent sampling permutation and satisfies $|\calL_n(s,y^n)|\leq\sum_{\ell=0}^{{\bar D_n^{\mathrm H}}}\binom n\ell(|\calY|-1)^\ell\leq(n|\calY|)^{{\bar D_n^{\mathrm H}}+1}$, hence $\log|\calL_n|\leq {C}{\bar D_n^{\mathrm H}}\log n$ for large $n$.

For any retained source realization $x^n\in\calX^n$, define its non-excess-distortion observation set and the corresponding random Hamming distance by
\begin{align}
\calY_D(x^n)&:=\{y^n\in\calY^n:{d^{(n)}}(x^n,\phi_n(f_n(x^n),y^n))\leq D\},\\
{{D_n^{\mathrm H}}}&:={\min_{y^n\in\calY_D(X^n)}}\sum_{i\in[n]}\bbI\{Y_i\ne y_i\}.
\end{align}
The set $\calY_D(x^n)$ is nonempty for every retained source because its conditional non-excess-distortion probability is at least $n^{-1}$. The conditional distance satisfies
\begin{align}
n^{-1}\leq \Pr\{{{D_n^{\mathrm H}}}=0|X^n\}
\leq\exp\left(-\frac{2(\bbE[{{D_n^{\mathrm H}}}|X^n])^2}{n}\right),\label{eq:v20-hamming-lower-1}
\end{align}
The flag leads to the first inequality. Conditionally on $X^n$, ${D_n^{\mathrm H}}$ is $1$-Lipschitz in independent observations, so bounded differences~\cite{mcdiarmid} lead to the second and thus $\bbE[{D_n^{\mathrm H}}|X^n]\leq\sqrt{n\log n/2}$.
For some positive constant ${{C}_{\rm -}}$,
\begin{align}
\Pr\{{{D_n^{\mathrm H}}}>{{\bar D_n^{\mathrm H}}}|X^n\}\leq\exp(-{{C}_{\rm -}} n^{1/8}),\label{eq:v20-hamming-upper-1}
\end{align}
The upper-tail bounded-differences bound leads to~\eqref{eq:v20-hamming-upper-1} uniformly over flagged sources for large $n$. On $\{{D_n^{\mathrm H}}\leq {\bar D_n^{\mathrm H}}\}$, the list contains a reconstruction of distortion at most $D$. The independent permutation ${\Upsilon}$ samples the $k_n$ disclosed coordinates.
For $z^n\in\calZ^n$, define sample distortion, the repaired reconstruction, and its sample-error event by
\begin{align}
\widehat d_n(z^n)&:=\frac1{k_n}\sum_{i\in[k_n]}d(X_{{\Upsilon}(i)},z_{{\Upsilon}(i)}),\label{eq:repair-sample-distortion}\\
Z^{\prime n}&\in\argmin_{z^n\in\calL_n(S_0,Y^n)}\widehat d_n(z^n),\label{eq:repair-decoder}\\
\calE_n^{\rm s}&:=\left\{\max_{z^n\in\calL_n(S_0,Y^n)}|\widehat d_n(z^n)-d^{(n)}(X^n,z^n)|>{\eta_n^{\mathrm d}}\right\}.\label{eq:repair-sample-error-event}
\end{align}
Resolve ties by a fixed ordering. The decoder uses only $S_0,Y^n,{\Upsilon}$ and disclosed source coordinates. Given $(X^n,Y^n)$, the event's only randomness is the independent uniform ${\Upsilon}$.
\begin{align}
\Pr\{\calE_n^{\rm s}|X^n,Y^n\}
&\leq2|\calL_n(S_0,Y^n)|
\exp\{-2k_n{(\eta_n^{\mathrm d})^2}/d_{\max}^2\}\\
&\leq2\exp\{{C}{\bar D_n^{\mathrm H}}\log n-2k_n{(\eta_n^{\mathrm d})^2}/d_{\max}^2\}
\label{eq:v20-sampling-1}\\
&\leq\exp(-{C}^{\prime}n^{5/8}),\label{eq:v20-sampling-2}
\end{align}
Hoeffding's without-replacement bound~\cite[Sec.~6]{hoeffding} and a union bound lead to~\eqref{eq:v20-sampling-1}; $k_n{(\eta_n^{\mathrm d})^2}\asymp n^{5/8}$ and ${\bar D_n^{\mathrm H}}\log n=o(n^{5/8})$ lead to~\eqref{eq:v20-sampling-2}, uniformly in the retained pair. Set $d_i':=d(\tilX_i,Z'_{{\Upsilon}(i)})$. On $\{{D_n^{\mathrm H}}\leq {\bar D_n^{\mathrm H}}\}\cap(\calE_n^{\rm s})^c$, choose $z_0^n$ in the list with distortion at most $D$. Sample accuracy and decoder minimization lead to
\begin{align}
d^{(n)}(X^n,Z^{\prime n})
&\leq\widehat d_n(Z^{\prime n})+{\eta_n^{\mathrm d}}\leq\widehat d_n(z_0^n)+{\eta_n^{\mathrm d}}\\
&\leq d^{(n)}(X^n,z_0^n)+2{\eta_n^{\mathrm d}}\leq D+2{\eta_n^{\mathrm d}}.
\end{align}
Using $d\leq d_{\max}$ on the two error events yields
\begin{align}
\bbE_{{P_n^{\rm b}}}\Bigl[\sum_{i\in[n]} d_i'\Bigr]
&\leq nD+2n{{\eta_n^{\mathrm d}}}+nd_{\max}\big(e^{-{{C}_{\rm -}}n^{1/8}}+e^{-{{C}^{\prime}}n^{5/8}}\big)\label{eq:repairdist}\\
&=nD+O(n^{15/16}),\label{eq:v20-repair-distortion-2}
\end{align}
where ${\eta_n^{\mathrm d}}=n^{-1/16}$ leads to the final order.

For ${S}'=({S},\tilX^{{k_n}})$, set $H_n':=-\log {P_n^{\rm o}}({S}'|\tilY^n,{{\Upsilon}})$. On $B_n=1$, Bayes' rule leads to
\begin{align}
P_n^{\rm o}(S'|\tilY^n,{\Upsilon})
&=\Pr\{B_n=1\}\,
\frac{P_n^{\rm b}(\tilY^n,{\Upsilon})}{P_n^{\rm o}(\tilY^n,{\Upsilon})}
P_n^{\rm b}(S_0,\tilX^{k_n}|\tilY^n,{\Upsilon}).
\label{eq:flagged-likelihood-factorization}
\end{align}
At the realized $(S_0,\tilX^{k_n})$, taking the expected negative logarithm leads to
\begin{align}
\bbE_{{P_n^{\rm b}}}[H_n']
&=H_{{P_n^{\rm b}}}({S_0},\tilX^{{k_n}}|\tilY^n,{{\Upsilon}})
+\log\frac1{\Pr\{{B_n}=1\}}-{\operatorname{KL}}({P^{\rm b}_{n,\tilY^n,{\Upsilon}}}\|{P^{\rm o}_{n,\tilY^n,{\Upsilon}}})\label{eq:repairentropy}\\
&\leq\log M+{k_n}\log|\calX|+O(1),
\end{align}
The inequality uses alphabet sizes, nonnegative relative entropy, and a flag probability bounded away from zero. Because the source-only flag is included in $S'$, the repaired reverse scores and forward and reverse residuals are centered under $P_n^{\rm b}$. Taking expectations in~\eqref{eq:combined} leads to
\begin{align}
\bbE_{P_n^{\rm b}}[\Delta_n']
&=\bbE_{P_n^{\rm b}}[H_n']
+\lambda^*\left(\bbE_{P_n^{\rm b}}
\left[\sum_{i\in[n]}d_i'\right]-nD\right)
-nR_{\rm WZ}(D)-\sqrt n\,
\bbE_{P_n^{\rm b}}[\Lambda_n].
\end{align}
The repaired gap then satisfies
\begin{align}
0\leq\bbE_{{P_n^{\rm b}}}[{\Delta_n}']
&\leq O(\sqrt n)+{k_n}\log|\calX|+O(n^{15/16})
=O(n^{15/16}),\label{eq:v20-repair-gap-1}
\end{align}
This follows from~\eqref{eq:repairdist},~\eqref{eq:repairentropy}, and Cauchy--Schwarz: the flag probability bounded away from zero leads to $\sqrt n|\bbE_{P_n^{\rm b}}[\Lambda_n]|=O(\sqrt n)$. The bound is $o(n)$ and controls both projections.

We first construct the covariance selector.
Choose a positive real number ${{\ell_{\calC}}}\in\bbR_+$ smaller than half the minimum contact separation and set
$\psi_j({Q}):=(1-\|{Q}-Q_j\|/{{\ell_{\calC}}})_+$ for $(j,{Q})\in[m]\times\calP(\calX)$.
For each $j\in[m]$, apply Lemma~\ref{lem:concentration} with this same function $\psi_j$ at every coordinate; its Lipschitz constant is $1/{\ell_{\calC}}$.
For $({Q},j)\in\calP(\calX)\times[m]$,
\begin{align}
|\bbI\{{{\Pi}_{\calC}}({Q})=Q_j\}-\psi_j({Q})|^2
\leq {{C}}\dist({Q},\calC)^2\leq {{C}}\Delta({Q}),\label{eq:contact-indicator-approximation}
\end{align}
Lipschitz continuity and $\psi_j(Q_k)=\bbI\{j=k\}$ lead to the first inequality;~\eqref{eq:quadratic} leads to the second.
For an arbitrary code description under one of the source-only distributions $P_n^{\rm b}$, let $\Psi_{n,j}=n^{-1}\sum_{i\in[n]}\psi_j({Q_i^{\rm s}})$. Then
\begin{align}
\bbE_{{P_n^{\rm b}}}[|\theta_{n,j}-\Psi_{n,j}|^2]&\leq {{C}}\bbE_{{P_n^{\rm b}}}[{\Delta_n}]/n\label{eq:v20-weight-moments-1}\\
\bbE_{{P_n^{\rm b}}}[{\var_{P_n^{\rm b}}}(\Psi_{n,j}|\tilX^n,{{\Upsilon}})]&\leq {{C}}{{I_n^{\mathrm t}}}/n\leq {{C}}\bbE_{{P_n^{\rm b}}}[{\Delta_n}]/n,\label{eq:v20-weight-moments-2}
\end{align}
where~\eqref{eq:v20-weight-moments-1} follows from Jensen's inequality and the preceding pointwise approximation;~\eqref{eq:v20-weight-moments-2} follows from Lemma~\ref{lem:concentration} and~\eqref{eq:tc-deficit}.
Consequently,
\begin{align}
\bbE_{{P_n^{\rm b}}}[\|{\btheta_n}-\bbE_{{P_n^{\rm b}}}[{\btheta_n}|\tilX^n,{{\Upsilon}}]\|^2]\leq {{C}}\bbE_{{P_n^{\rm b}}}[{\Delta_n}]/n,\label{eq:weights-proj}
\end{align}
where~\eqref{eq:weights-proj} follows from optimality of conditional expectation in mean square and the preceding two bounds, summed over $j\in[m]$.

Apply~\eqref{eq:weights-proj} to the repaired code under $P_n^{\rm b}$, and set $\hat{\btheta}_n':=\bbE_{P_n^{\rm b}}[\btheta_n'|\tilX^n,{\Upsilon}]$. Then
\begin{align}
\bbE_{{P_n^{\rm b}}}[\|{\btheta_n}'-\hat{\btheta}_n'\|^2]=O(n^{-1/16}),\label{eq:repaired-proj}
\end{align}
where~\eqref{eq:repaired-proj} follows from~\eqref{eq:weights-proj} and~\eqref{eq:v20-repair-gap-1}.
For $i\in\{{k_n}+1,\ldots,n\}$,
\begin{align}
{Q_i^{\rm s\prime}}&={Q_i^{\rm s}}\label{eq:suffix}\\
{{\Pi}_{\calC}}({Q_i^{\rm s\prime}})&={{\Pi}_{\calC}}({Q_i^{\rm s}})\label{eq:v20-suffix-2}\\
\|{\btheta_n}-{\btheta_n}'\|_1&\leq2{k_n}/n=o(1),\label{eq:suffix-weights}
\end{align}
The prefix is already in $\tilX^{i-1}$ on the suffix, so its posteriors and nearest contacts agree; only $k_n$ weight contributions can differ.
On ${B_n}=1$, project $\hat{\btheta}_n'$ onto $\calA$ to obtain $\bar{\btheta}_n$; extend both by a fixed point of $\calA$ on ${B_n}=0$. Then
\begin{align}
\|{\btheta_n}-\bar{\btheta}_n\|
&\leq2\|{\btheta_n}-\hat{\btheta}_n'\|+\dist({\btheta_n},\calA),\label{eq:weight-projection-distance}
\end{align}
The triangle inequality and nearest-point optimality lead to this bound. Projection onto the nonempty compact convex $\calA$ is continuous, so $\bar{\btheta}_n$ is source-measurable. For every $\eta>0$,
\begin{align}
\Pr\{{\calE_n}\cap\{\|{\btheta_n}-\bar{\btheta}_n\|>\eta\}\}
&\leq \Pr\{{\calE_n}\cap\{\|{\btheta_n}-\hat{\btheta}_n'\|>\eta/4\}\}\notag\\
&\quad+\Pr\{{\calE_n}\cap\{\dist({\btheta_n},\calA)>\eta/2\}\}\longrightarrow0,\label{eq:weight-projection-tail}
\end{align}
where in~\eqref{eq:weight-projection-tail}, the inequality follows from~\eqref{eq:weight-projection-distance}; convergence follows from~\eqref{eq:repaired-proj},~\eqref{eq:suffix-weights},~\eqref{eq:weights-feasible}, and $\Pr\{{\calE_n},{B_n}=0\}\leq n^{-1}$. Recall $\Omega_n$ and $\Sigma$ from~\eqref{eq:predictable-covariance} and~\eqref{eq:mixture-covariance}, respectively. The weight-tail bound,~\eqref{eq:feasible}, and linearity of $\Sigma$ lead to~\eqref{eq:cov-selector}, which proves Claim~(ii) of Lemma~\ref{lem:selectors}.

For Claim~(i), retain the exact prefix information to cancel its possibly larger-than-$\sqrt n$ cost. Claim~(iii) of Lemma~\ref{lem:tc-deficit},~\eqref{eq:residual-second-moment}, and~\eqref{eq:v20-repair-gap-1} lead to
\begin{align}
\bbE_{{P_n^{\rm b}}}[({\tilde t_n}'-t_n^{{S}'})^2]+n^{-1}\bbE_{{P_n^{\rm b}}}[({{\mu}_n}')^2]=O(n^{-1/16}),\label{eq:repairscalar}
\end{align}
For a deterministic prefix length $k\in[n]$, define its conditional information and the residual first-order entropy by
\begin{align}
{H_{n,k}^{\rm p}}&:=-\log {P_n^{\rm o}}(\tilX^k|{S},\tilY^n,{{\Upsilon}})\label{eq:prefix-information}\\
{H_*^{\rm r}}&:=H_{P_{XY}}(X|Y)-R_{\mathrm{WZ}}(D).\label{eq:residual-entropy}
\end{align}
For $k\in[n]$, recall the definitions of ${\Delta_{n,k}}$ and $({{\mu}_{n,k}},{{\mu_{n,k}^{\mathrm d}}})$ in~\eqref{eq:prefix-gap} and~\eqref{eq:prefix-residuals}, respectively.
Define the remainder
\begin{align}
{{\mu_{n,k}^{\mathrm a}}}
&:=\sum_{i\in[k]}\big[-\log P_{X|Y}(\tilX_i|\tilY_i)-H_{P_{XY}}(X|Y)-\big(\jW{\tilX_i}-R_{\mathrm{WZ}}(D)\big)\notag\\
&\hspace{42mm}-\imath_{{{\Pi}_{\calC}}({Q_i^{\rm s}})}^{\rm c}(\tilX_i,\tilY_i)-{\lambda^*} d_{{{\Pi}_{\calC}}({Q_i^{\rm s}})}^{\rm c}(\tilX_i,\tilY_i)\big]
-{{\mu}_{n,k}}-{{\mu_{n,k}^{\mathrm d}}}.{\label{eq:prefix-allocation-remainder}}
\end{align}
The prefix information satisfies
\begin{align}
{H_{n,k}^{\rm p}}
&=\sum_{i\in[k]}[-\log P_{X|Y}(\tilX_i|\tilY_i)-\imath_{{Q_i^{\rm s}}}(\tilX_i,\tilY_i)]\label{eq:prefixidentity}\\
&=k{H_*^{\rm r}}-{\Delta_{n,k}}+{\lambda^*}\Bigl(\sum_{i\in[k]}d_i-kD\Bigr)+{{\mu_{n,k}^{\mathrm a}}},
\end{align}
The posterior chain rule leads to the first line; summing~\eqref{eq:pointwise} and the residual definitions leads to the second. For $k=k_n$ and every $\eta>0$,
\begin{align}
\Pr\Bigl\{{\calE_n}\cap\bigl\{\bigl|{{\mu_{n,k_n}^{\mathrm a}}}\bigr|>\eta\sqrt n\bigr\}\Bigr\}&\longrightarrow0,\label{eq:v20-prefix-remainder-1}
\end{align}
Use~\eqref{eq:prefix-residual-rate}--\eqref{eq:v20-prefix-residual-1} and $k_n=o(n)$: the centered baseline/source and reverse martingale sums have variance $O(k_n)$ under $P_n^{\rm o}$, whose relabeled pairs are iid given ${\Upsilon}$.

The prefix distortion decomposes as
\begin{align}
\frac{\sum_{i\in[{k_n}]}d_i-{k_n}D}{\sqrt n}
&=\frac{\sum_{i\in[{k_n}]}d_i-({k_n}/n)\sum_{i\in[n]} d_i}{\sqrt n}
+\frac{{k_n}}{n}\frac{\sum_{i\in[n]} d_i-nD}{\sqrt n},
\end{align}
For every $\eta>0$,
\begin{align}
\Pr\Bigl\{{\calE_n}\cap\bigl\{\bigl|\sum_{i\in[{k_n}]}d_i-{k_n}D\bigr|>\eta\sqrt n\bigr\}\Bigr\}
&\leq\frac{4{{C}}{k_n}}{\eta^2n}
+\Pr\Bigl\{{\calE_n}\cap\bigl\{\bigl|\sum_{i\in[n]} d_i-nD\bigr|>
\tfrac{\eta n\sqrt n}{2{k_n}}\bigr\}\Bigr\}\longrightarrow0,
\end{align}
where the inequality follows from a union bound and Chebyshev's inequality: the original distortion vector is fixed before ${{\Upsilon}}$, so the first term in the decomposition has conditional variance at most ${{C}}{k_n}/n$. The remaining probability vanishes by~\eqref{eq:v20-distortion-tight-1} and ${k_n}/n\to0$. Thus,
\begin{align}
\Pr\Bigl\{{\calE_n}\cap\bigl\{\bigl|{H_{n,k_n}^{\rm p}}-{k_n}{H_*^{\rm r}}+{\Delta_{n,k_n}}\bigr|>\eta\sqrt n\bigr\}\Bigr\}&\longrightarrow0,\label{eq:prefixcancel}
\end{align}
where~\eqref{eq:prefixcancel} follows by substituting the preceding distortion and remainder bounds into~\eqref{eq:prefixidentity}; it retains and cancels the possibly order-$n^{3/4}$ prefix cost.

Since $H_n'-H_n=H_{n,k_n}^{\rm p}$, subtracting the two instances of~\eqref{eq:exactH} leads to
\begin{align}
{\lambda^*}\sqrt n({\tilde t_n}'-{\tilde t_n})
={H_{n,k_n}^{\rm p}}-\sqrt n(J_n'-J_n)-({{\mu}_n}'-{{\mu}_n}),\label{eq:exact-allocation-transfer}
\end{align}
Define the source-measurable $\bart_n:=t_n^{S'}-k_nH_*^{\rm r}/(\lambda^*\sqrt n)$ on $B_n=1$, extending it arbitrarily elsewhere. Then
\begin{align}
\Pr\Bigl\{{\calE_n}\cap\bigl\{\bigl|\bart_n-{\tilde t_n}'+\tfrac{{k_n}{H_*^{\rm r}}}{{\lambda^*}\sqrt n}\bigr|>\eta\bigr\}\Bigr\}&\longrightarrow0\label{eq:v20-allocation-transfer-1}\\
\Pr\Bigl\{{\calE_n}\cap\bigl\{\bigl|\bart_n-{\tilde t_n}-\tfrac{{H_{n,k_n}^{\rm p}}-{k_n}{H_*^{\rm r}}}{{\lambda^*}\sqrt n}\bigr|>\eta\bigr\}\Bigr\}&\longrightarrow0\label{eq:v20-allocation-transfer-2}\\
\Pr\Bigl\{{\calE_n}\cap\bigl\{\bigl|\bart_n-{\tilde t_n}+\tfrac{{\Delta_{n,k_n}}}{{\lambda^*}\sqrt n}\bigr|>\eta\bigr\}\Bigr\}&\longrightarrow0,\label{eq:allocation-transfer-tail}
\end{align}
where~\eqref{eq:v20-allocation-transfer-1} follows from the definition of $\bart_n$ and~\eqref{eq:repairscalar};~\eqref{eq:v20-allocation-transfer-2} follows from~\eqref{eq:exact-allocation-transfer}, the vanishing channel-score difference,~\eqref{eq:repairscalar} for the repaired residual, and~\eqref{eq:v20-gap-residual-limits-2} for the original residual; and~\eqref{eq:allocation-transfer-tail} follows from~\eqref{eq:prefixcancel}. For the channel-score difference, $\bbE[(J_n'-J_n)^2|\tilX^n,{{\Upsilon}}]\leq {{C}}{k_n}/n=o(1)$ follows from~\eqref{eq:suffix} and reverse martingale orthogonality. The omitted flag contributes at most $n^{-1}$ to each probability.
Consequently,
\begin{align}
\Pr\Bigl\{{\calE_n}\cap\bigl\{\dist\Bigl(\bart_n,\bigl[{\tilde t_n}-\tfrac{{\Delta_n}}{{\lambda^*}\sqrt n},\,{\tilde t_n}\bigr]\Bigr)>\eta\bigr\}\Bigr\}\longrightarrow0,\label{eq:alloc-selector}
\end{align}
where~\eqref{eq:alloc-selector} follows from~\eqref{eq:allocation-transfer-tail} and $0\leq {\Delta_{n,k_n}}\leq {\Delta_n}$. Both selectors depend only on $(\tilX^n,{{\Upsilon}})$.

For ${v_-}\leq {v_+}$ and ${v}\in\bbR$,
\begin{align}
\max\{({v}-{v_+})_+,({v_-}-{v})_+\}&\leq\dist({v},[{v_-},{v_+}]),\label{eq:interval-violations}
\end{align}
by the three positions of ${v}$ relative to $[{v_-},{v_+}]$. Apply this with ${v_-}=\tilde t_n-\Delta_n/(\lambda^*\sqrt n)$, ${v_+}=\tilde t_n$, and ${v}=\bart_n$. The rate/distortion events~\eqref{eq:rate-event}--\eqref{eq:distortion-event}, localization~\eqref{eq:initial-information-violation}--\eqref{eq:initial-distortion-violation}, and~\eqref{eq:alloc-selector} yield
\begin{align}
\Pr\{{\calE_n}\cap{\calR_n}(\bart_n;\eta)^c\}&\longrightarrow0\label{eq:successrectangle}\\
\Pr\{{\calE_n}\cap\calD_n(\bart_n;\eta)^c\}&\longrightarrow0,\label{eq:successrectangle-distortion}
\end{align}
using $\lambda^*>0$. Their union bound leads to, for every $\eta>0$,
\begin{align}
\Pcn(D)
&\leq \Pr\{J_n\leq L-{\Lambda_n}-{\lambda^*}\bart_n+{\eta},\quad
G_n\leq\bart_n+{\eta}\}+o(1),
\end{align}
Both thresholds and $\Sigma(\bar{\btheta}_n)$ depend only on $(\tilX^n,{\Upsilon})$. The two violation limits prove Claim~(i); with~\eqref{eq:cov-selector} for Claim~(ii), they complete Lemma~\ref{lem:selectors}.

\section{Proof of Lemma~\ref{lem:completion}}
\label{app:completion}
Recall $\calE_n$, $\bXi_n$, and $\Omega_n$ from~\eqref{eq:success-event} and~\eqref{eq:generic-martingale-quantities}. Under $P_n^{\rm o}$, the positive semidefinite target $\Gamma_n$ and thresholds $\ell_{n,j_2}$ are measurable with respect to the initial sigma-field ${\calF}_n^{\rm 0}$ and satisfy~\eqref{eq:completion-covariance-hypothesis}--\eqref{eq:completion-threshold-hypothesis}. The uniform bounds lead to $\|{\boldsymbol{\zeta}}_{n,i}\|\leq {C}_0$ and $\|\Gamma_n\|_2\leq {C}_\Gamma$ almost surely, with constants independent of $n$. We keep $(\tau,\eta)\in(0,\infty)^2$ fixed as $n\to\infty$. First, predictable stopping and independent completion lead to exact covariance; next, a uniform conditional Gaussian approximation permits transfer to $\calE_n$. The independent signs below extend $P_n^{\rm o}$ without changing the original experiment.

\subsection{Stopping and Completing the Covariance}
Retain original increments until the next accumulated predictable covariance would fail to be dominated by $n(\Gamma_n+\tau \mathrm{Id}_2)$ in positive semidefinite order, then set that increment and all later increments to zero. The next covariance is known before its increment is revealed, so retaining the original filtration preserves conditional centering. Let $\Gamma_n^{\mathrm{r}}$ be the target minus the retained covariance, so $0\preceq\Gamma_n^{\mathrm{r}}\preceq n(\Gamma_n+\tau \mathrm{Id}_2)$. Here, for an arbitrary symmetric matrix $\Sigma_{\rm lim}\in\bbR^{2\times2}$, $\Sigma_{\rm lim}\succeq0$ denotes positive semidefiniteness, $\Sigma_{\rm lim}^{1/2}$ its positive semidefinite square root, and $\tr\Sigma_{\rm lim}$ its trace.

After revealing the original filtration, reveal sign vectors ${\boldsymbol R}_j\in\{-1,1\}^2$, $j\in[n]$, with independent symmetric coordinates, independent across $j$ and of the original experiment. Append $(\Gamma_n^{\mathrm{r}})^{1/2}{\boldsymbol R}_j/\sqrt n$, whose conditional mean is zero and conditional covariance is $\Gamma_n^{\mathrm{r}}/n$. The completed predictable covariance is therefore $n(\Gamma_n+\tau \mathrm{Id}_2)$, and
\begin{align}
\bigl\|(\Gamma_n^{\mathrm{r}})^{1/2}{\boldsymbol R}_j/\sqrt n\bigr\|
\leq\sqrt{2\|\Gamma_n^{\mathrm{r}}\|_2/n}\leq {C}',\label{eq:completion-increment-bound}
\end{align}
by $\|{\boldsymbol R}_j\|=\sqrt2$ and the bound on $\Gamma_n$. Divide the retained increments by $\sqrt n$ to obtain $\overline{{\boldsymbol{\zeta}}}_{n,i}$, $i\in[n]$, and define
\begin{align}
\overline{{\boldsymbol{\zeta}}}_{n,n+j}&:=(\Gamma_n^{\mathrm{r}})^{1/2}{\boldsymbol R}_j/n,\quad j\in[n],\\
\widetilde\bXi_n&:=\sum_{k\in[2n]}\overline{{\boldsymbol{\zeta}}}_{n,k},\label{eq:completed-sum}\\
\calE_n^{\rm v}&:=\{\|\Omega_n/n-\Gamma_n\|_2\leq\tau/2\}.\label{eq:completion-good-covariance}
\end{align}
On $\calE_n^{\rm v}$, each partial original covariance is dominated by $\Omega_n\preceq n(\Gamma_n+\tau \mathrm{Id}_2/2)$, so no stopping occurs and
\begin{align}
\tr(\Gamma_n^{\mathrm{r}}/n)\leq3\tau.\label{eq:completion-trace}
\end{align}
Conditionally on the original experiment, the appended normalized noise has covariance $\Gamma_n^{\mathrm{r}}/n$. Conditional Markov inequality and~\eqref{eq:completion-trace} lead to
\begin{align}
\Pr\{\calE_n\cap\calE_n^{\rm v},\ \|\widetilde\bXi_n-\bXi_n\|>\eta\}
&\leq\eta^{-2}\bbE[\bbI\{\calE_n\cap\calE_n^{\rm v}\}\tr(\Gamma_n^{\mathrm{r}}/n)]\label{eq:addednoise}\\
&\leq3\tau/\eta^2.\label{eq:v20-completion-noise-2}
\end{align}

\subsection{Uniform Conditional Gaussian Limit}
The Gaussian comparison is under $P_n^{\rm o}$, without conditioning on $\calE_n$.
\begin{proposition}\label{prop:uniform-rectangle}
For every $\tau>0$, the completed martingale satisfies
\begin{align}
\sup_{(\ell_1,\ell_2)\in\bbR^2}\bigl|
\Pr\{\widetilde\Xi_{n,1}\leq\ell_1,\widetilde\Xi_{n,2}\leq\ell_2|{\calF}_n^{\rm 0}\}
-\Phi_{\Gamma_n+\tau \mathrm{Id}_2}(\ell_1,\ell_2)\bigr|=o(1),\label{eq:uniformrectangle-prop}
\end{align}
with a deterministic bound outside a null set of initial-information values, uniformly over targets with the stated bounds.
\end{proposition}
\begin{proof}
Define the completed filtration $({\calF}_{n,k}^{\rm c})_{k=0}^{2n}$ to retain the original filtration for $k\leq n$ and then include the first $k-n$ sign vectors, with ${\calF}_{n,0}^{\rm c}={\calF}_n^{\rm 0}$. For ${\boldsymbol v}\in\bbR^2$ and $k\in[2n]$, put $\zeta_{n,k}:={\boldsymbol v}^\top\overline{{\boldsymbol{\zeta}}}_{n,k}$ and $\sigma_{n,k}^2:=\bbE[(\zeta_{n,k})^2|{\calF}_{n,k-1}^{\rm c}]$. Define the total projected variance $\sigma_n^2$ by~\eqref{eq:v20-completion-variance-2}. The original increment bound,~\eqref{eq:completion-increment-bound}, and exact covariance lead to
\begin{align}
|\zeta_{n,k}|&\leq {C}_{{\boldsymbol v}}n^{-1/2},\label{eq:v20-completion-variance-1}\\
\sigma_n^2&:={\boldsymbol v}^\top(\Gamma_n+\tau \mathrm{Id}_2){\boldsymbol v},\label{eq:v20-completion-variance-2}\\
\sum_{k\in[2n]}\sigma_{n,k}^2&=\sigma_n^2\leq({C}_\Gamma+\tau)\|{\boldsymbol v}\|^2.\label{eq:v20-completion-variance-3}
\end{align}
For $k\in\{0,\ldots,2n\}$, let $\chi_{n,k}:=\exp(\rmi\sum_{\ell\in[k]}\zeta_{n,\ell}+\frac12\sum_{\ell\in[k]}\sigma_{n,\ell}^2)$, where $\rmi^2=-1$. For each $k\in[2n]$, define the Taylor remainder ${\mu_{n,k}^{\mathrm T}}$ by~\eqref{eq:v20-conditional-taylor-1}. Conditional centering and Taylor expansion yield
\begin{align}
\bbE[e^{\rmi\zeta_{n,k}}|{\calF}_{n,k-1}^{\rm c}]
&=1-\tfrac12\sigma_{n,k}^2+{\mu_{n,k}^{\mathrm T}},\label{eq:v20-conditional-taylor-1}\\
|{\mu_{n,k}^{\mathrm T}}|&\leq\tfrac16\bbE[|\zeta_{n,k}|^3|{\calF}_{n,k-1}^{\rm c}],\label{eq:v20-conditional-taylor-2}\\
|\bbE[\chi_{n,k}-\chi_{n,k-1}|{\calF}_{n,k-1}^{\rm c}]|
&\leq {C}_{{\boldsymbol v}}\bigl(\bbE[|\zeta_{n,k}|^3|{\calF}_{n,k-1}^{\rm c}]+\sigma_{n,k}^4\bigr),\label{eq:v20-conditional-taylor-3}
\end{align}
where the last bound uses $|\chi_{n,k-1}|\leq {C}_{{\boldsymbol v}}$ and expansion of $e^{\sigma_{n,k}^2/2}$. Almost surely, $\sum_{k\in[2n]}\bbE[|\zeta_{n,k}|^3|{\calF}_{n,k-1}^{\rm c}]\leq {C}_{{\boldsymbol v}}\sigma_n^2/\sqrt n$ and $\sum_{k\in[2n]}\sigma_{n,k}^4\leq {C}_{{\boldsymbol v}}^2\sigma_n^2/n$. Telescoping~\eqref{eq:v20-conditional-taylor-3} from $\chi_{n,0}=1$ and multiplying by the initially measurable $e^{-\sigma_n^2/2}$ leads to
\begin{align}
\Bigl|\bbE[e^{\rmi{\boldsymbol v}^\top\widetilde\bXi_n}|{\calF}_n^{\rm 0}]
-e^{-{\boldsymbol v}^\top(\Gamma_n+\tau \mathrm{Id}_2){\boldsymbol v}/2}\Bigr|
&\leq {C}_{{\boldsymbol v}}(n^{-1/2}+n^{-1})=O_{{\boldsymbol v}}(n^{-1/2}).\label{eq:conditionalcf}
\end{align}
Martingale orthogonality also leads to $\bbE[\|\widetilde\bXi_n\|^2|{\calF}_n^{\rm 0}]=\tr(\Gamma_n+\tau \mathrm{Id}_2)\leq2({C}_\Gamma+\tau)$. Choose conditional versions on a common probability-one set for every $n$ and every vector with rational coordinates. This second-moment bound leads to uniform equicontinuity, extending~\eqref{eq:conditionalcf} to all ${\boldsymbol v}\in\bbR^2$. If uniform weak convergence failed, compactness of the targets $\tau \mathrm{Id}_2\preceq\Gamma_n+\tau \mathrm{Id}_2\preceq({C}_\Gamma+\tau)\mathrm{Id}_2$ would lead to initial-information values and a target subsequence converging to $\Sigma_{\rm lim}\succeq\tau \mathrm{Id}_2$. Equation~\eqref{eq:conditionalcf} and the characteristic-function continuity theorem~\cite{billingsley} then imply weak convergence to $\calN(\bzero_2,\Sigma_{\rm lim})$, a contradiction. Uniform second moments control tails; on a fixed square, a finite rectangle grid, monotonicity, and the Gaussian marginal-density bound $(2\pi\tau)^{-1/2}$ control all thresholds. Thus,~\eqref{eq:uniformrectangle-prop} has a deterministic error depending only on $({C}_0,{C}_\Gamma,\tau)$.
\end{proof}

\subsection{\texorpdfstring{Transfer to the Non-Excess-Distortion Event}{Transfer to the Non-Excess-Distortion Event}}
Let $\calE_n^{\rm t}(\eta):=\bigcap_{j_2\in[2]}\{\Xi_{n,j_2}\leq \ell_{n,j_2}+\eta\}$, $\veps_n^{\rm v}(\tau):=\Pr\{\calE_n\cap\bigl(\calE_n^{\rm v}\bigr)^c\}$, and $\veps_n^{\rm t}(\eta):=\allowbreak\sum_{j_2\in[2]}\Pr\{\calE_n,\allowbreak\Xi_{n,j_2}>\ell_{n,j_2}+\eta\}$. For fixed $(\tau,\eta)$, the hypotheses lead to $\veps_n^{\rm t}(\eta)+\veps_n^{\rm v}(\tau)=o(1)$. Let $\veps_n^{\rm r}(\tau)=o(1)$ be the deterministic error in Proposition~\ref{prop:uniform-rectangle}. Then
\begin{align}
\Pr\{\calE_n\}
&\leq\Pr\{\calE_n\cap\calE_n^{\rm v}\cap\calE_n^{\rm t}(\eta)\}+\veps_n^{\rm v}(\tau)+\veps_n^{\rm t}(\eta)\label{eq:v20-completion-chain-1}\\
&\leq\Pr\{\widetilde\Xi_{n,j_2}\leq \ell_{n,j_2}+2\eta,\ j_2\in[2]\}
+3\tau/\eta^2+o(1)\label{eq:v20-completion-chain-2}\\
&=\bbE\Bigl[\Pr\{\widetilde\Xi_{n,j_2}\leq \ell_{n,j_2}+2\eta,\ j_2\in[2]|{\calF}_n^{\rm 0}\}\Bigr]
+3\tau/\eta^2+o(1)\label{eq:v20-completion-chain-3}\\
&\leq\bbE[\Phi_{\Gamma_n+\tau \mathrm{Id}_2}(\ell_{n,1}+2\eta,\ell_{n,2}+2\eta)]
+3\tau/\eta^2+\veps_n^{\rm r}(\tau)+o(1).\label{eq:completion-explicit-errors}
\end{align}
The steps use the union bound;~\eqref{eq:addednoise}, since displacement at most $\eta$ on $\calE_n^{\rm v}\cap\calE_n^{\rm t}(\eta)$ leads to the enlarged thresholds; the tower property; and Proposition~\ref{prop:uniform-rectangle}, whose uniformity permits ${\calF}_n^{\rm 0}$-measurable thresholds. This proves~\eqref{eq:completionbound}.

\section{Proofs of Lemmas~\ref{lem:gaussian} and~\ref{lem:finite-approximation}}
\label{app:gaussian}
Use the Gaussian cdf $\Phi_\Sigma$ and envelopes $F^{\rm c}_{\btheta},F^{\rm c},F^{\rm c}_{\tau,\eta}$ in~\eqref{eq:gaussian-cdf} and~\eqref{eq:gaussian-envelopes}--\eqref{eq:smoothed-envelope}. Lemma~\ref{lem:geometry},~\eqref{eq:weightset}, and~\eqref{eq:mixture-covariance} lead to $\lambda^*>0$, nonempty compact $\calA$, and continuous bounded $\btheta\mapsto\Sigma(\btheta)$. We first remove smoothing, then approximate the covariance family by finite sets, and finally separate the quantile.

\subsection{Smoothing and Its Removal}
For $(\tau,\eta)\in(0,\infty)^2$, bounding the first Gaussian marginal density uniformly over $(\btheta,t)\in\calA\times\bbR$ and taking suprema leads to
\begin{align}
|F^{\rm c}_{\tau,\eta}(\ell+\ell_{\rm i})-F^{\rm c}_{\tau,\eta}(\ell)|
\leq |\ell_{\rm i}|/\sqrt{2\pi\tau},\qquad(\ell,\ell_{\rm i})\in\bbR^2.\label{eq:v20-envelope-lipschitz-1}
\end{align}
Thus, $0\leq F^{\rm c}_{\tau,\eta}\leq1$ is continuous;~\eqref{eq:pointwise-envelope-bound} follows from its definition. For fixed $\ell$, choose positive $\tau_k\to0$ attaining the limit superior and $1/k$-maximizers $(\btheta_k,t_k)\in\calA\times\bbR$. If the limit superior is positive, allocations are bounded: otherwise one threshold tends to $-\infty$ and uniformly bounded covariances force the probabilities to zero. Compactness leads to a subsequence converging to $(\btheta,t)$. Gaussian weak convergence, closed-set Portmanteau, and continuity from above yield
\begin{align}
\limsup_{k\to\infty}F^{\rm c}_{\tau_k,\tau_k^{1/4}}(\ell)
&\leq\lim_{\eta_{\rm u}\downarrow0}\Phi_{\Sigma(\btheta)}(\ell-\lambda^*t+\eta_{\rm u},t+\eta_{\rm u})\label{eq:v20-remove-smoothing-1}\\
&=\Phi_{\Sigma(\btheta)}(\ell-\lambda^*t,t)\leq F^{\rm c}(\ell).\label{eq:v20-remove-smoothing-2}
\end{align}
A zero limit superior is immediate. These bounds prove Claim~(i) of Lemma~\ref{lem:gaussian}, including singular covariances.

\subsection{Finite Approximation of the Covariance Family}
For $\btheta\in\calA$ and $\ell\in\bbR$, a deterministic first coordinate leads to
\begin{align}
F^{\rm c}_{\btheta}(\ell)=
\begin{cases}
\Phi\bigl(\ell/(\lambda^*\sqrt{\Sigma_{22}(\btheta)})\bigr),
&\Sigma_{11}(\btheta)=0<\Sigma_{22}(\btheta),\\
\bbI\{\ell\geq0\},&\Sigma(\btheta)=0,
\end{cases}\label{eq:degenerate-envelope}
\end{align}
by taking $t=\ell/\lambda^*$; $\Phi$ is the standard Gaussian cdf in~\eqref{eq:standard-gaussian-functions}. If $\Sigma_{11}(\btheta)>0$, the first marginal density bound leads to continuity uniformly over allocations. Hence every nonzero-covariance envelope is continuous. Second countability of $\bbR^2$ supplies a countable subfamily with the same union of open sets $\{(\ell,F_{\rm \ell}):F_{\rm \ell}<F^{\rm c}_{\btheta}(\ell)\}$ over nonzero covariances. Let $\calA_\infty$ contain its weights and one zero-covariance weight, if any; all such weights share~\eqref{eq:degenerate-envelope}. Then $\calA_\infty$ is nonempty and $\sup_{\btheta\in\calA_\infty}F^{\rm c}_{\btheta}=F^{\rm c}$ pointwise. Fix an enumeration of $\calA_\infty$ and, for each $k\in\bbN$, define $\calA_k$ to contain its first $k$ elements, or all elements if fewer exist. These sets are nested, finite, and nonempty, proving~\eqref{eq:finiteapprox} and Lemma~\ref{lem:finite-approximation}.

\subsection{Quantile Separation}
Recall the definition of $F^{\rm c}$ in~\eqref{eq:full-pointwise-envelope}. We now prove that the supremum is attained and that this envelope is right-continuous. For every $\ell\in\bbR$,
\begin{align}
F^{\rm c}(\ell)&=\max_{(\btheta,t)\in\calA\times\bbR}\Phi_{\Sigma(\btheta)}(\ell-\lambda^*t,t),\label{eq:envelope-attainment}\\
F^{\rm c}(\ell)&=\lim_{\ell_{\rm i}\downarrow0}F^{\rm c}(\ell+\ell_{\rm i}).\label{eq:envelope-right-continuity}
\end{align}
If $F^{\rm c}(\ell)=0$, every rectangle probability is zero and any pair attains the maximum. If $F^{\rm c}(\ell)>0$, probabilities vanish uniformly as $|t|\to\infty$, locally uniformly in $\ell$. Bounded allocations, compactness, and closed-set Portmanteau lead to attainment. The same subsequence argument for converging $\ell$ leads to upper semicontinuity; monotonicity then leads to right continuity. With $(J_{\btheta},G_{\btheta})\sim \calN(\bzero_2,\Sigma(\btheta))$, for $\ell<0$,
\begin{align}
F^{\rm c}(\ell)\leq\sup_{\btheta\in\calA}\Pr\{J_{\btheta}+\lambda^*G_{\btheta}\leq\ell\}
\leq {C}/\ell^2\longrightarrow0\quad(\ell\to-\infty),\label{eq:envelope-left-tail}
\end{align}
by rectangle inclusion and bounded variance. Fixing $\btheta$ and taking $t=\ell/(2\lambda^*)$ leads to $F^{\rm c}(\ell)\to1$ as $\ell\to\infty$. If $0<F^{\rm c}(\ell)<1$, its maximizing covariance is nonzero, and for $\ell_{\rm i}>0$,
\begin{align}
F^{\rm c}(\ell+\ell_{\rm i})&\geq\Phi_{\Sigma(\btheta)}(\ell-\lambda^*t+\ell_{\rm i}/2,t+\ell_{\rm i}/(2\lambda^*))
>\Phi_{\Sigma(\btheta)}(\ell-\lambda^*t,t)=F^{\rm c}(\ell).\label{eq:envelope-strict-increase}
\end{align}
The allocation $t+\ell_{\rm i}/(2\lambda^*)$ leads to the first inequality. Enlarging both thresholds leads to strictness by positive Gaussian density in rank two, expansion of the supporting half-line or interval in rank one, or strict increase of the univariate cdf with one random coordinate. Thus, no intermediate constant interval exists.

Recall $\Lambda\sim \calN(0,V)$, $F_*(L)=\bbE[F^{\rm c}(L-\Lambda)]$, and $L_*(\veps)$ from~\eqref{eq:Fstar}--\eqref{eq:quantile}. For $V=0$, $F_*=F^{\rm c}$; otherwise $F^{\rm c}$ is a cdf and Gaussian convolution leads to a continuous strictly positive density. The end limits make $L_*(\veps)$ finite for $\veps\in(0,1)$, and right continuity leads to $F_*(L_*(\veps))\geq1-\veps$. Below this quantile,~\eqref{eq:quantile} leads to $F_*(L)<1-\veps$. Above it, a strict inequality at the quantile suffices by monotonicity; if equality holds,~\eqref{eq:envelope-strict-increase} or the positive convolution density leads to $F_*(L)>1-\veps$. Both cases allow zero-covariance jumps, proving~\eqref{eq:quantile-separation} and Claim~(ii) of Lemma~\ref{lem:gaussian}.
\section{Proof of Lemma~\ref{lem:oneshot}}
\label{app:oneshot}
Fix the inputs $P_n^{\rm a}$, positive $Q_n^{\rm a}$, and $z^n$ of Lemma~\ref{lem:oneshot}; $\imath_n$ and $D_n$ are defined in~\eqref{eq:block-aux-information}--\eqref{eq:block-aux-distortion}. For $M\in\bbN$ and $\gamma>0$, generate a source-independent array of independent unit-rate exponentials ${E}_{s,u^n}$, $(s,u^n)\in[M]\times\calU_{\rm a}^n$, shared during random-code analysis. For $x^n\in\calX^n$, define the encoder output and its selected auxiliary sequence by
\begin{align}
(f_n(x^n),u_{\mathrm e}^n(x^n))
&:=\operatorname*{argmin}_{(s,u^n)\in[M]\times\calU_{\rm a}^n}
\frac{{E}_{s,u^n}}{P_n^{\rm a}(u^n|x^n)/M}.
\label{eq:exponential-race-encoder}
\end{align}
Send $f_n(x^n)$. For $(s,y^n)\in[M]\times\calY^n$, define the decoded auxiliary sequence and reconstruction by
\begin{align}
\widehat u^n(s,y^n)
&:=\operatorname*{argmin}_{u^n\in\calU_{\rm a}^n}
\frac{{E}_{s,u^n}}{Q_n^{\rm a}(u^n|y^n)},
\label{eq:exponential-race-decoder}\\
\phi_n(s,y^n)
&:=z^n(\widehat u^n(s,y^n),y^n).
\label{eq:exponential-race-reconstruction}
\end{align}
Zero encoder weights have score $+\infty$. A fixed ordering resolves ties, which have probability zero for positive weights. Fixing the array makes the code deterministic.

Fix $(x^n,y^n)$, let $\calJ^{\rm r}:=[M]\times\calU_{\rm a}^n$, and write $j=(s,u^n)$ for the encoder winner. For $k=(s_k,u_k^n)\in\calJ^{\rm r}$, define the weights and winning time
\begin{align}
P_{\rm r}(k)&:=P_n^{\rm a}(u_k^n|x^n)/M,\qquad
Q_{\rm r}(k):=\bbI\{s_k=s\}Q_n^{\rm a}(u_k^n|y^n),
\label{eq:exponential-race-weights}\\
{E_*}&:=\min_{k\in\calJ^{\rm r}:\ P_{\rm r}(k)>0}\frac{{E}_{s_k,u_k^n}}{P_{\rm r}(k)}.
\label{eq:exponential-race-winning-time}
\end{align}
Let $\hatj=(s,\widehat u^n(s,y^n))$ be the decoder winner. Over the exponential array, $\Pr\{j=k,{E_*}>{e_*}\mid x^n,y^n\}=P_{\rm r}(k)e^{-{e_*}}$ for ${e_*}\geq0$, so ${E_*}$ is a unit-rate exponential independent of $j$. Conditional on $j$ and ${E_*}={e_*}$, the other exponentials have independent unit-rate residuals above $P_{\rm r}(k){e_*}$. Matching requires these residuals to exceed ${e_*}(P_{\rm r}(j)Q_{\rm r}(k)/Q_{\rm r}(j)-P_{\rm r}(k))_+$. Averaging over ${E_*}$ leads to
\begin{align}
\Pr\{\hatj=j|x^n,y^n,j\}
&=\int_0^\infty\exp\Biggl[-{e_*}\Bigl(1+\sum_{k\in\calJ^{\rm r}\setminus\{j\}}
\bigl(\tfrac{P_{\rm r}(j)Q_{\rm r}(k)}{Q_{\rm r}(j)}-P_{\rm r}(k)\bigr)_+\Bigr)\Biggr]d{e_*}\label{eq:v20-race-1}\\
&=\Bigl(1+\sum_{k\in\calJ^{\rm r}\setminus\{j\}}
\bigl(\tfrac{P_{\rm r}(j)Q_{\rm r}(k)}{Q_{\rm r}(j)}-P_{\rm r}(k)\bigr)_+\Bigr)^{-1}\label{eq:v20-race-2}\\
&\geq\Bigl(1+\tfrac{P_{\rm r}(j)}{Q_{\rm r}(j)}\sum_{k\in\calJ^{\rm r}}Q_{\rm r}(k)\Bigr)^{-1}
=\Bigl(1+\tfrac{\exp\{\imath_n(x^n;u^n|y^n)\}}{M}\Bigr)^{-1},\label{eq:v20-race-3}
\end{align}
where the second line evaluates the integral; the last uses $({v}_1-{v}_2)_+\leq{v}_1\quad({v}_1,{v}_2\geq0)$, $\sum_{k\in\calJ^{\rm r}}Q_{\rm r}(k)=1$, and $P_{\rm r}(j)/Q_{\rm r}(j)=e^{\imath_n(x^n;u^n|y^n)}/M$, with $Q_{\rm r}(j)>0$.
The selected auxiliary has distribution $P_n^{\rm a}(\cdot|X^n)$ independently of $Y^n$ given $X^n$. Correct auxiliary decoding with $D_n\leq nD$ implies $\calE_n$ from~\eqref{eq:success-event}. Averaging~\eqref{eq:v20-race-3} therefore yields
\begin{align}
\Pr\{\calE_n\}
&\geq\Pr\{\hatj=j,\ D_n\leq nD\}\label{eq:v20-oneshot-success-1}\\
&\geq\bbE\left[\frac{\bbI\{D_n\leq nD\}}
{1+\exp\{\imath_n(X^n;U^n|Y^n)\}/M}\right].\label{eq:v20-oneshot-success-2}
\end{align}
The complementary error, averaged over the array, satisfies
\begin{align}
1-\Pr\{\calE_n\}
&\leq1-\bbE\left[\frac{\bbI\{D_n\leq nD\}}
{1+\exp\{\imath_n(X^n;U^n|Y^n)\}/M}\right]
\label{eq:oneshot-average-error}\\
&\leq1-(1-e^{-\gamma})\Pr\{D_n\leq nD,\ \imath_n\leq\log M-\gamma\}
\label{eq:oneshot-threshold-error}\\
&\leq\Pr\{\{D_n>nD\}\cup\{\imath_n>\log M-\gamma\}\}+e^{-\gamma},
\label{eq:oneshot-complement-error}
\end{align}
where the first line uses~\eqref{eq:v20-oneshot-success-2}; the second restricts to the displayed event and uses $(1+e^{-\gamma})^{-1}\geq1-e^{-\gamma}$; and the last takes its complement. Choose one array whose excess-distortion probability does not exceed the average. This same deterministic code satisfies both bounds in~\eqref{eq:oneshot}, proving Lemma~\ref{lem:oneshot}.

\section{Proof of Lemma~\ref{lem:types}}
\label{app:types}
Fix $\btheta\in\calA$ and a compact allocation interval as in Lemma~\ref{lem:types}. Recall $P_{\btheta}$, $g_j$, $\bard_{\btheta}$, ${\widetilde V}_{\btheta}$, and $\dot P_{\btheta}$ from~\eqref{eq:optimal-joint},~\eqref{eq:contact-rule-output}, and~\eqref{eq:weight-mean-distortion}--\eqref{eq:distortion-direction}. Retain $\calU_{\btheta}:=\{j\in[m]:\theta_j>0\}$, extend $P_{\btheta}(\cdot|x)$ and $\dot P_{\btheta}(\cdot|x)$, the target $P_{\btheta,n}(\cdot|x)$, and rounded type $\hatP_n(\cdot|x)$ by zero on $[m]\setminus\calU_{\btheta}$. The source type is $\hatP$, $\bxi=\sqrt n(\hatP-P_X)$, and $P_n^{\rm a}$ is uniform on $\calT_{\hatP_n}(x^n)$. Its reference is $Q_n^{\rm a}(u^n|y^n)=\prod_{i\in[n]}P_{\btheta}(u_i|y_i)$. All conditional statements use $Y^n\sim P_{Y|X}^n(\cdot|x^n)$ and are uniform over~\eqref{eq:type-admissibility}--\eqref{eq:type-admissibility-end}, hence over any fixed finite set of weights.

\subsection{\texorpdfstring{Conditional Type and Mean Expansions}{Conditional Type and Mean Expansions}}
With $\bard(x,j)$ the channel distortion mean from~\eqref{eq:mean-distortion-rule}, centering and~\eqref{eq:existing-direction} lead to
\begin{align}
\sum_{j\in[m]}\dot P_{\btheta}(j|x)&=0,\label{eq:v20-direction-1}\\
\sum_{x\in\calX}\sum_{j\in[m]}P_X(x)\dot P_{\btheta}(j|x)\bard(x,j)
&=\frac{{\widetilde V}_{\btheta}}{{\widetilde V}_{\btheta}}=1.\label{eq:v20-direction-2}
\end{align}
Here ${\widetilde V}_{\btheta}\geq\underline{\widetilde V}>0$ by~\eqref{eq:uniform-direction}. Interior contacts make the retained encoder and decoder probabilities positive. For admissible types, define
\begin{align}
\iota_{\rm t}:=-t-\langle\bxi,\bard_{\btheta}\rangle.\label{eq:type-perturbation-coefficient}
\end{align}
Thus,~\eqref{eq:perturb} is a row-normalized positive kernel for all large $n$, since $|\iota_{\rm t}|/\sqrt n=O(n^{-3/8})$ uniformly. Write $n_x:=n\hatP(x)\geq {C}_{\rm -}n$ for large $n$, using full support and the cutoff. Round each $n_xP_{\btheta,n}(j|x)$ down and define the integer deficit by
\begin{align}
n_x^{\rm r}:=n_x-\sum_{j\in\calU_{\btheta}}\lfloor n_xP_{\btheta,n}(j|x)\rfloor
\in\{0,\ldots,|\calU_{\btheta}|-1\}
\label{eq:type-row-rounding-deficit}
\end{align}
Distribute it among $n_x^{\rm r}$ distinct entries. Each count differs from its target by at most one, proving~\eqref{eq:conditional-type-rounding-error}. Dividing by $n_x$ leads to conditional error $O(1/n)$ and positive retained probabilities for large $n$. Assigning these counts to source-$x$ positions proves that the class is nonempty.
Grouping the conditional distortion mean by joint type and substituting the perturbed, rounded kernel leads to
\begin{align}
\bbE[D_n|X^n,U^n]
&=n\sum_{x\in\calX}\sum_{j\in[m]}\hatP(x)\hatP_n(j|x)\bard(x,j)\label{eq:typedistcalc}\\
&=nD+\sqrt n\langle \bxi,\bard_{\btheta}\rangle
+\sqrt n\iota_{\rm t}\sum_{x\in\calX}\sum_{j\in[m]}P_X(x)\dot P_{\btheta}(j|x)\bard(x,j)
+O(1+\|\bxi\||\iota_{\rm t}|)\label{eq:v20-type-distortion-2}\\
&=nD-\sqrt nt+O(1+\|\bxi\|^2+t^2),\label{eq:v20-type-distortion-3}
\end{align}
where the second line has $O(1)$ rounding error; the last uses~\eqref{eq:v20-direction-2}, the definition of $\iota_{\rm t}$, and $\|\bxi\||\iota_{\rm t}|=O(\|\bxi\|^2+t^2)$.

For $Q_X\in\calP(\calX)$ and $Q_{U|X},P_{U|X}\in\calP(\calU_{\btheta}|\calX)$, write $\operatorname{KL}(Q_{U|X}\|P_{U|X}|Q_X):=\sum_{x\in\calX}Q_X(x)\operatorname{KL}(Q_{U|X}(\cdot|x)\|P_{U|X}(\cdot|x))$. The joint type has entries $\hatP(x)\hatP_n(j|x)$. Standard conditional-type bounds~\cite[Lemma~1.2]{TanBook}, with $|\calU_{\btheta}|\leq m$, imply
\begin{align}
&(n+1)^{-|\calX|m}\exp\{nH_{\hatP\hatP_n}(U|X)\}
\leq|\calT_{\hatP_n}(x^n)|\leq\exp\{nH_{\hatP\hatP_n}(U|X)\},\label{eq:conditional-type-cardinality}\\
&-\log|\calT_{\hatP_n}(x^n)|=-nH_{\hatP\hatP_n}(U|X)+O(\log n).\label{eq:v20-type-divergence-1}
\end{align}
Uniformity of $P_n^{\rm a}$ and the entropy--cross-entropy identity then yield
\begin{align}
\log P_n^{\rm a}(U^n|X^n)-\sum_{i\in[n]}\log P_{\btheta}(U_i|X_i)
&=n\operatorname{KL}(\hatP_n\|P_{\btheta}|\hatP)+O(\log n)\label{eq:v20-type-divergence-2}\\
&=O(1+\|\bxi\|^2+t^2+\log n),\label{eq:typediv}
\end{align}
where $\log{v}\leq{v}-1\quad({v}>0)$ bounds each row divergence by $\sum_{j\in\calU_{\btheta}}(\hatP_n(j|x)-P_{\btheta}(j|x))^2/P_{\btheta}(j|x)$. Retained $P_{\btheta}(j|x)$ are bounded away from zero, and the row difference is $O((1+\|\bxi\|+|t|)/\sqrt n)$. The remainder is therefore $o(\sqrt n)$ uniformly under the cutoff.
For $(x,j)\in\calX\times\calU_{\btheta}$, define the channel information mean as
\begin{align}
\bar\imath(x,j)&:=\sum_{y\in\calY:\ P_{Y|X}(y|x)>0}P_{Y|X}(y|x)\imath_j(x,y).\label{eq:type-contact-information-mean}
\end{align}
It follows from~\eqref{eq:optimal-conditionals} that
\begin{align}
\bbE[\imath_n|X^n,U^n]
&=\sum_{i\in[n]}\bar\imath(X_i,U_i)
+\log P_n^{\rm a}(U^n|X^n)-\sum_{i\in[n]}\log P_{\btheta}(U_i|X_i).\label{eq:type-information-mean-decomposition}
\end{align}
Stationarity~\eqref{eq:stationarity} and the distortion expansion lead to
\begin{align}
\sum_{i\in[n]}\bar\imath(X_i,U_i)
&=n(R_{\mathrm{WZ}}(D)+{\lambda^*} D)+\sum_{i\in[n]} \big(\jW{X_i}-R_{\mathrm{WZ}}(D)\big)-{\lambda^*}\sum_{i\in[n]}\bard(X_i,U_i)\label{eq:v20-type-information-1}\\
&=nR_{\mathrm{WZ}}(D)+\sqrt n(\Lambda_n+{\lambda^*} t)+O(1+\|\bxi\|^2+t^2),\label{eq:v20-type-information-2}
\end{align}
where the last line uses the source score~\eqref{eq:source-block-score}. Together with~\eqref{eq:typediv}, these identities prove~\eqref{eq:type-means};~\eqref{eq:v20-type-distortion-3} proves~\eqref{eq:type-distortion-mean}. Thus, Claim~(i) holds.

\subsection{\texorpdfstring{Conditional Gaussian Limit}{Conditional Gaussian Limit}}
Recall the channel-centered scores and covariance $\Sigma_j$ from~\eqref{eq:centered-scores}--\eqref{eq:contact-covariance}, evaluated only where $P_{Y|X}(y|x)>0$. The block definitions and product reference lead to the exact centered decomposition. For ${\boldsymbol v}\in\bbR^2$, also define its scalar projection:
\begin{align}
\frac1{\sqrt n}
\begin{pmatrix}
\imath_n-\bbE[\imath_n|X^n,U^n]\\
D_n-\bbE[D_n|X^n,U^n]
\end{pmatrix}
&=\frac1{\sqrt n}\sum_{i\in[n]}
\begin{pmatrix}
\imath_{U_i}^{\rm c}(X_i,Y_i)\\
d_{U_i}^{\rm c}(X_i,Y_i)
\end{pmatrix},\label{eq:type-centered-decomposition}\\
\zeta_{n,i}&:=\frac1{\sqrt n}{\boldsymbol v}^\top
\begin{pmatrix}
\imath_{U_i}^{\rm c}(X_i,Y_i)\\
d_{U_i}^{\rm c}(X_i,Y_i)
\end{pmatrix},\qquad i\in[n].\label{eq:type-scalar-projection}
\end{align}
Conditionally on $(X^n,U^n)$, the product channel makes these summands independent and centered. Finite alphabets and positive retained probabilities bound them uniformly. Grouping their covariance by joint type leads to
\begin{align}
\frac1n\sum_{i\in[n]}\Sigma_{U_i}(X_i)
=\sum_{x\in\calX}\sum_{j\in[m]}\hatP(x)\hatP_n(j|x)\Sigma_j(x)
\longrightarrow\sum_{j\in[m]}\theta_j\sum_{x\in\calX}Q_j(x)\Sigma_j(x)=\Sigma(\btheta),\label{eq:type-covariance-limit}
\end{align}
by~\eqref{eq:perturb} and~\eqref{eq:mixture-covariance}. Covariance additivity and $|\zeta_{n,i}|\leq {C}_{{\boldsymbol v}}/\sqrt n$ consequently imply
\begin{align}
\sum_{i\in[n]}\bbE[(\zeta_{n,i})^2|X^n,U^n]
&={\boldsymbol v}^\top\Bigl(\frac1n\sum_{i\in[n]}\Sigma_{U_i}(X_i)\Bigr){\boldsymbol v}
\longrightarrow {\boldsymbol v}^\top\Sigma(\btheta){\boldsymbol v}\label{eq:v20-type-moments-1}\\
\sum_{i\in[n]}\bbE[|\zeta_{n,i}|^3|X^n,U^n]&\leq {C}_{{\boldsymbol v}}n^{-1/2}.\label{eq:v20-type-moments-2}
\end{align}
Here ${C}_{{\boldsymbol v}}>0$ is a finite constant depending only on $\btheta$ and ${\boldsymbol v}$. With $\rmi^2=-1$, conditional independence and Taylor expansion yield
\begin{align}
\bbE\Bigl[e^{\rmi\sum_{i\in[n]}\zeta_{n,i}}\bigm|X^n,U^n\Bigr]
&=\prod_{i=1}^n\Bigl(1-\tfrac12\bbE[(\zeta_{n,i})^2|X^n,U^n]+O(n^{-3/2})\Bigr)\label{eq:v20-type-characteristic-1}\\
&=\exp\Bigl(-\tfrac12\sum_{i\in[n]}\bbE[(\zeta_{n,i})^2|X^n,U^n]+O(n^{-1/2})\Bigr)\label{eq:v20-type-characteristic-2}\\
&\longrightarrow\exp\bigl(-\tfrac12{\boldsymbol v}^\top\Sigma(\btheta){\boldsymbol v}\bigr),\label{eq:v20-type-characteristic-3}
\end{align}
where the second line uses $\log(1+{\digamma})={\digamma}+O(|{\digamma}|^2)$ for sufficiently small complex ${\digamma}$; the last uses~\eqref{eq:v20-type-moments-1}. Errors are uniform over admissible triples. The characteristic-function continuity theorem~\cite{billingsley}, applied to~\eqref{eq:type-centered-decomposition} for every ${\boldsymbol v}$, proves~\eqref{eq:typeclt} along every admissible sequence. Zero-variance projections vanish in mean square by~\eqref{eq:v20-type-moments-1}, so singular covariances are included. This proves Claim~(ii) and Lemma~\ref{lem:types}.

\section{Numerical Calculations}\label{app:numerical}
\label{app:nonunique-example}
Recall $V$, $F_*$, and $L_*$ from~\eqref{eq:source-variance},~\eqref{eq:Fstar}, and~\eqref{eq:quantile}. We compute~\eqref{eq:bac-main-values} and Table~\ref{tab:nonunique}, first identifying the binary optimizer numerically, then deriving the quaternary optimal set analytically and evaluating its Gaussian envelopes. Scripts use unrounded parameters; displayed decimals are rounded.

\subsection{Binary Asymmetric Source}
\label{app:bac}
Let $\calX=\calY=\calZ=\{0,1\}$ and $P_X=(0.65,0.35)$. For $(x,z)\in\calX\times\calZ$, use $d(x,z)=\bbI\{x\ne z\}$, and fix $D=0.05$ and $\veps=0.1$. The side-information channel, with rows and columns indexed by $x$ and $y$, is
\begin{align}
P_{Y|X}=\begin{pmatrix}0.92&0.08\\0.20&0.80\end{pmatrix}.
\end{align}
The source and observation probabilities are positive, as required by (A1). Substituting $q=0.35$, ${p_{01}=0.08}$, and ${p_{10}=0.20}$ into~\eqref{eq:bac-overlap-general}--\eqref{eq:bac-distortion-endpoints} leads to $\kappa_W=31/125\approx0.25$, $D_{\min}=0$, and $D_0=61/500\approx0.12$. Thus, (A2) and (A3) hold at the stated $D$ and $\veps$. For the numerical candidate $P^*=P_{XY}P^*_{{\widetilde U}|X}$ with active labels $j\in[3]$, define $Q_j:=P^*_{X|{\widetilde U}}(\cdot|j)$ and $\theta_j:=P^*_{\widetilde U}(j)$. Choose the enumeration so that $\widetilde g(j,y)=g_j(y)$ for $j\in[3]$, where $g_1(y)=0$, $g_2(y)=y$, and $g_3(y)=1$ for $y\in\{0,1\}$, with
\begin{align}
(Q_j(1))_{j=1}^3\approx(0.01,\ 0.34,\ 0.98),\qquad
(\theta_j)_{j=1}^3\approx(0.43,\ 0.34,\ 0.23).
\end{align}
The contacts and dual parameters solve $\Delta_{g_j}(Q_j)=0$ and~\eqref{eq:contactgradient} for $j\in[3]$, with $\Delta_g$ from~\eqref{eq:rule-gap}. Their weights satisfy $\sum_{j\in[3]}\theta_j=1$, $\sum_{j\in[3]}\theta_jQ_j(1)=0.35$, and $\sum_{j\in[3]}\theta_j\delta(Q_j)=0.05$. The remaining rule $g(y)=1-y$ has numerical minimum gap approximately $2.08>0$. Bayes' rule leads to $P^*_{{\widetilde U}|X}(j|x)=\theta_jQ_j(x)/P_X(x)$.

Evaluate $\lambda^*$, $\Sigma(P)$, and $\jW{\cdot}$ from~\eqref{eq:multiplier-definition},~\eqref{eq:cov}, and~\eqref{eq:common-optimal-density} at $P^*$, and write $\Sigma:=\Sigma(P^*)$. This leads to
\begin{align}
R_{\mathrm{WZ}}(D)&\approx0.21,\qquad {\lambda^*}\approx2.91,\qquad V\approx0.0252,\\
(\jW{0},\jW{1})&\approx(0.09,\ 0.43),\\
\Sigma&\approx\begin{pmatrix}0.235&-5.35\times10^{-5}\\-5.35\times10^{-5}&0.0346\end{pmatrix}.
\end{align}
The conditional variance in~\eqref{eq:existing-direction} is positive and approximately $0.01$.
Set $\sigma_{\rms}=\sqrt{V}$,\allowbreak $\sigma_{\rmI}=\sqrt{\Sigma_{11}}$,\allowbreak $\sigma_{\rmD}=\sqrt{\Sigma_{22}}$,\allowbreak and ${\bar\Sigma_{12}}=\Sigma_{12}/(\sigma_{\rmI}\sigma_{\rmD})$. Let $\Phi_2(\cdot,\cdot;{\bar\Sigma_{12}})$ denote the standard bivariate Gaussian distribution function with correlation ${\bar\Sigma_{12}}$; recall the definitions of $\Phi$ and ${\Phi'}$ in~\eqref{eq:standard-gaussian-functions}. For $(\ell,L)\in\bbR^2$,\allowbreak with $\bar{\ell}_{\rm s}\in\bbR$ the standardized source-Gaussian integration variable, define the Gaussian expressions for the numerically identified covariance by
\begin{align}
F^{\rm c}_{\mathrm{BA}}(\ell)&:=\max_{t\in\bbR}\Phi_2\left(\frac{\ell-{\lambda^*} t}{\sigma_{\rmI}},\frac{t}{\sigma_{\rmD}};{\bar\Sigma_{12}}\right),\label{eq:asymmetricenvelope}\\
F_{\mathrm{BA}}(L)&:=\int_{-\infty}^{\infty}F^{\rm c}_{\mathrm{BA}}(L-\sigma_{\rms}\bar{\ell}_{\rm s}){\Phi'}(\bar{\ell}_{\rm s})\,d\bar{\ell}_{\rm s}.
\end{align}
For each pair $(\ell,t)\in\bbR^2$ of a threshold and a candidate allocation, define
\begin{align}
\bar{\ell}_{\rm I}:=\frac{\ell-{\lambda^*} t}{\sigma_{\rmI}},\qquad \bar{\ell}_{\rm D}:=\frac{t}{\sigma_{\rmD}}.
\end{align}
The maximizing allocation satisfies
\begin{align}
\frac{{\Phi'}(\bar{\ell}_{\rm D})}{\sigma_{\rmD}}\Phi\left(\frac{\bar{\ell}_{\rm I}-{\bar\Sigma_{12}} \bar{\ell}_{\rm D}}{\sqrt{1-{\bar\Sigma_{12}}^2}}\right)
=\frac{{\lambda^*}{\Phi'}(\bar{\ell}_{\rm I})}{\sigma_{\rmI}}\Phi\left(\frac{\bar{\ell}_{\rm D}-{\bar\Sigma_{12}} \bar{\ell}_{\rm I}}{\sqrt{1-{\bar\Sigma_{12}}^2}}\right),
\end{align}
where the equality follows by differentiating the bivariate Gaussian cdf with respect to $t$ at the maximum.
The numerical Gaussian coefficient for this candidate covariance is
\begin{align}
L_*(0.1)\approx1.69,
\end{align}
where the value is obtained from~\eqref{eq:asymmetricenvelope} by a bracketed allocation search, Gaussian quadrature, and a root search for $F_{\mathrm{BA}}(L)=0.9$. Its interpretation as $L_*(0.1)$ uses the optimal covariance family identified by floating-point stationary-equation and inactive-rule-gap calculations; certified enclosures of these calculations are not provided. The Gaussian calculation uses unrounded parameters; Gauss--Hermite quadrature with $80$ and $140$ nodes yields the displayed value in both cases.

\subsection{Quaternary Source: Model and Gaussian Calculation}
\label{app:four-symbol-model}
Use $(p,r,q_{\rm p})=(2/5,1/3,1/10)$ from Section~\ref{sec:nonunique-example}, with $\alpha\in(0,p)$, $w$, and ${\omega}_p$ from~\eqref{eq:example-tangent}--\eqref{eq:example-phi}. The distortion and threshold in~\eqref{eq:example-distortion}--\eqref{eq:example-threshold} have $|\calX|=4$, $|\calY|=2$, and $|\calZ|=6$. Define prediction distributions with coordinates ${x_{\rm p}}\in\{0,1\}$ by
\begin{align}
Q_0^{\rm p}=(1-q_{\rm p},q_{\rm p}),\qquad Q_1^{\rm p}=(q_{\rm p},1-q_{\rm p}),\qquad Q_*^{\rm p}=(1-r,r).\label{eq:prediction-distributions}
\end{align}
For every ${({x_{\rm p}},{z_{\rm p}})\in\{0,1\}\times\{0,1,*\}}$, substitution into~\eqref{eq:prediction-cost-table} leads to
\begin{align}
d_{\rm p}({x_{\rm p}},{z_{\rm p}})=\log\frac{1-q_{\rm p}}{Q_{z_{\rm p}}^{\rm p}({x_{\rm p}})}.\label{eq:prediction-cost-log-form}
\end{align}
The natural-log convention is used in both distortion terms: ${w}$, $D$, and $d_{\max}$ therefore use the same scale as the prediction loss. The prediction loss is nonnegative and bounded because every ${Q_{z_{\rm p}}^{\rm p}({x_{\rm p}})}$ is positive and at most ${1-q_{\rm p}}$. In particular,
\begin{align}
D_{\min}&=0<D<D_0=h_2(r)+\log(1-q_{\rm p})+wp,\\
\kappa_W&=2p(1-p)=0.48,\qquad d_{\max}=\log\frac{1-q_{\rm p}}{q_{\rm p}}+w\approx2.99.
\end{align}
Choosing the correct bit and a prediction assigning probability $1-q_{\rm p}$ to ${X_{\rm p}}$ leads to $D_{\min}=0$. With only $Y$, the prior minimizes the prediction risk of independent ${X_{\rm p}}$, and $Y$ is the optimal bit estimate, giving $D_0$. Since $h_2(q_{\rm p})+\log(1-q_{\rm p})=q_{\rm p}\log((1-q_{\rm p})/q_{\rm p})>0$ and $h_2(r)>h_2(q_{\rm p})$ for $0<q_{\rm p}<r<1/2$, we have $0<D<D_0$. This verifies (A3) for every $\veps\in(0,1)$; (A1) and (A2) were checked in Section~\ref{sec:assumptions}.

To realize the threshold in~\eqref{eq:example-threshold} as a mean distortion, choose ${P_{Z_{\rm p}}(1)=(r-q_{\rm p})/(1-2q_{\rm p})=7/24}$, ${P_{{X_{\rm p}}|{Z_{\rm p}}}(1|0)=q_{\rm p}}$, and ${P_{{X_{\rm p}}|{Z_{\rm p}}}(1|1)=1-q_{\rm p}}$. These probabilities induce ${P_{X_{\rm p}}(1)=r}$ and define ${P_{{Z_{\rm p}}|{X_{\rm p}}}}$ by Bayes' rule. Generate ${{Z_{\rm p}}}$ independently of ${({X_{\rm b}},{N})}$ given ${{X_{\rm p}}}$ and reconstruct ${Z=({Z_{\rm p}},Y)}$. Since ${\Pr\{{X_{\rm p}}\ne {Z_{\rm p}}\}=q_{\rm p}}$ and ${\Pr\{{X_{\rm b}}\ne Y\}=p}$, its mean distortion is ${q_{\rm p}\log((1-q_{\rm p})/q_{\rm p})+wp=D}$.

With $\diag$ denoting the diagonal matrix, define
\begin{align}
R_0:=h_2(r)-h_2(q_{\rm p}),\qquad \varpi:=p+(1-2p)\alpha.\label{eq:example-rate-output}
\end{align}
Define the positive standard deviations $\sigma_{\rmI},\sigma_{\rmD},\sigma_{\rms}$ by
\begin{align}
\sigma_{\rmI}^2&:=p(1-p)\left(\log\frac{1-\varpi}{\varpi}\right)^2,\qquad \sigma_{\rmD}^2:=w^2p(1-p),\label{eq:example-variances}\\
\sigma_{\rms}^2&:=r(1-r)\left(\log\frac{1-r}{r}\right)^2=\frac29(\log2)^2.
\end{align}
For ${(\vartheta_{\rm m},\ell)\in[0,1]\times\bbR}$, define
\begin{align}
\Sigma_{\vartheta_{\rm m}}:=\diag((1-\vartheta_{\rm m})\sigma_{\rmI}^2,\vartheta_{\rm m}\sigma_{\rmD}^2).\label{eq:example-covariance-definition}
\end{align}
Define the corresponding pointwise envelope by
\begin{align}
F^{\rm c}_{\vartheta_{\rm m}}(\ell):=\sup_{t\in\bbR}\Phi_{\Sigma_{\vartheta_{\rm m}}}(\ell-t,t).\label{eq:example-envelope-definition}
\end{align}
The multiplier in Lemma~\ref{lem:nonunique-example} is the overall distortion multiplier ${\lambda^*}$; the constant ${w}$ is only the Hamming weight.

\begin{lemma}\label{lem:nonunique-example}
For the above quaternary example and threshold $D$ in~\eqref{eq:example-threshold}, the following hold for ${({x_{\rm p}},{x_{\rm b}},\ell)\in\{0,1\}^2\times\bbR}$.
\begin{enumerate}
\item The first-order rate and source information satisfy
\begin{align}
R_{\mathrm{WZ}}(D)=R_0,\qquad {\lambda^*}=1.
\end{align}
\item The source information and dispersion satisfy
\begin{align}
\jW{({x_{\rm p}},{x_{\rm b}})}&=-\log P_{X_{\rm p}}({x_{\rm p}})-h_2(q_{\rm p}),\\
\jW{({x_{\rm p}},{x_{\rm b}})}-R_{\mathrm{WZ}}(D)&=({x_{\rm p}}-\tfrac13)\log2,\qquad V=\sigma_{\rms}^2.
\end{align}
\item The optimal covariance family is
\begin{align}
\{\Sigma(P):P\in\calP^*(D)\}&=\{\Sigma_{\vartheta_{\rm m}}:0\leq\vartheta_{\rm m}\leq1\}.\label{eq:example-full-family}
\end{align}
\item $0<\sigma_{\rmI}<\sigma_{\rmD}$, and the two endpoint test channels attain the full pointwise envelope
\begin{align}
\max_{0\leq\vartheta_{\rm m}\leq1}F^{\rm c}_{\vartheta_{\rm m}}(\ell)
=\begin{cases}
\Phi(\ell/\sigma_{\rmD}),&\ell<0,\\
\Phi(\ell/\sigma_{\rmI}),&\ell\geq0
\end{cases}.\label{eq:example-envelopes}
\end{align}
\item For every $L\geq0$, the best fixed-test-channel expression in~\eqref{eq:fixed-versus-adaptive} satisfies
\begin{align}
F_{\mathrm{f}}(L)=\Phi\left(\frac{L}{\sqrt{\sigma_{\rms}^2+\sigma_{\rmI}^2}}\right)<F_*(L).\label{eq:example-fixed}
\end{align}
\end{enumerate}
\end{lemma}
The proof in Section~\ref{app:four-symbol-proof} covers dependent auxiliary labels and hence the entire optimal set. Theorem~\ref{thm:achievability} permits the singular endpoints. With $\lambda^*=1$,~\eqref{eq:example-envelopes} leads to, for every $L\in\bbR$,
\begin{align}
F_*(L)
&=\int_{-\infty}^{L/\sigma_{\rms}}
\Phi\left(\frac{L-\sigma_{\rms}\bar{\ell}_{\rm s}}{\sigma_{\rmI}}\right){\Phi'}(\bar{\ell}_{\rm s})\,d\bar{\ell}_{\rm s}
+\int_{L/\sigma_{\rms}}^\infty
\Phi\left(\frac{L-\sigma_{\rms}\bar{\ell}_{\rm s}}{\sigma_{\rmD}}\right){\Phi'}(\bar{\ell}_{\rm s})\,d\bar{\ell}_{\rm s},\label{eq:example-adaptive-integral}
\end{align}
where~\eqref{eq:example-envelopes} selects its two branches according to the sign of ${L-\sigma_{\rms}\bar{\ell}_{\rm s}}$, and $\Lambda=\sigma_{\rms}\bar{\ell}_{\rm s}$ is averaged over the standardized source-Gaussian coordinate $\bar{\ell}_{\rm s}\in\bbR$.
There is no remaining optimization over intermediate test channels. For a source block ${(X_{\rm p}^n,X_{\rm b}^n)}$, define the source count and write the source score in~\eqref{eq:source-block-score} as
\begin{align}
n_{\rm p}:=\sum_{i\in[n]}\bbI\{X_{{\rm p},i}=1\},\qquad \Lambda_n=\frac{(n_{\rm p}-n/3)\log2}{\sqrt n},\label{eq:quaternary-source-score}
\end{align}
where Claim~(ii) of Lemma~\ref{lem:nonunique-example} reduces each source score to ${(X_{{\rm p},i}-1/3)\log2}$.
At rate offset $L\in\bbR$, the encoder selects the endpoint ${\vartheta_{\rm m}=0}$ when $\Lambda_n\leq L$ and ${\vartheta_{\rm m}=1}$ otherwise, then applies the local conditional-type perturbation in Section~\ref{sec:achievability}. The choice uses no realization of $Y^n$. The numerical parameters are
\begin{align}
\alpha&\approx0.31,\qquad w\approx0.79,\qquad D\approx0.54,\\
R_0&\approx0.31,\qquad \sigma_{\rmI}^2\approx0.00583,\qquad \sigma_{\rmD}^2\approx0.150.
\end{align}
Recall the definition of $L_{\mathrm{f}}(\veps)$ in~\eqref{eq:fixed-quantile-definition}. For any real number $\veps\in(0,1/2)$,~\eqref{eq:example-fixed} leads to
\begin{align}
L_{\mathrm{f}}(\veps)=\sqrt{\sigma_{\rms}^2+\sigma_{\rmI}^2}\,\Phi^{-1}(1-\veps).\label{eq:example-fixed-quantile}
\end{align}
Table~\ref{tab:nonunique} compares this value with the root $F_*(L_*(\veps))=1-\veps$, computed by quadrature split at $L/\sigma_{\rms}$ and a bracketed scalar root search. An independent $50$-digit calculation evaluates the improvement integral in~\eqref{eq:example-strict} and leads to the same displayed coefficients.

\subsection{Proof of Lemma~\ref{lem:nonunique-example}}
\label{app:four-symbol-proof}
Recall the model with $(p,r,q_{\rm p})=(2/5,1/3,1/10)$ and distortion~\eqref{eq:example-distortion}; $w$, $\alpha\in(0,p)$, and the predictions are given in~\eqref{eq:example-tangent} and~\eqref{eq:prediction-distributions}. We first identify the contacts attaining the component support bounds. We next construct the optimal mixtures and compute their source and conditional fluctuations, proving Claims~(i)--(iii). Finally, we optimize their Gaussian envelopes and prove the fixed-channel comparison in Claims~(iv) and~(v).

\subsubsection{\texorpdfstring{Supporting Inequalities and Product Contacts}{Supporting Inequalities and Product Contacts}}
Recall ${\omega}_p$ from~\eqref{eq:example-phi}, and extend it continuously to $[0,1]$. For the binary posterior $\Ber(a)$, $a\in[0,1]$, the constant-$0$, identity, constant-$1$, and complement rules have risks $a,p,1-a,1-p$, respectively. Thus, define
\begin{align}
\delta_p(a):=\min\{a,p,1-a\},\qquad \rho_p(a):={\omega}_p(a)-h_2(p).\label{eq:binary-posterior-functions}
\end{align}
For $a\in(0,1)$, set $\varpi_a:=p+(1-2p)a$. Then
\begin{align}
{\omega}_p''(a)
&=\left(\frac{1}{a(1-a)}-
\frac{(1-2p)^2}{\varpi_a(1-\varpi_a)}\right)
=\frac{p(1-p)}{a(1-a)\varpi_a(1-\varpi_a)}>0,\label{eq:example-binary-curvature}
\end{align}
where~\eqref{eq:example-binary-curvature} follows by differentiating ${{\omega}_p(a)=h_2(\varpi_a)-h_2(a)}$ and using ${\varpi_a(1-\varpi_a)-(1-2p)^2a(1-a)=p(1-p)}$.
For every $a\in[0,1]$,
\begin{align}
{\omega}_p(a)+w\delta_p(a)\geq wp,\label{eq:example-binary-support}
\end{align}
where~\eqref{eq:example-binary-support} follows from~\eqref{eq:example-binary-curvature} and the tangent condition ${{\omega}_p(\alpha)+(p-\alpha){\omega}_p'(\alpha)=0}$, with the endpoint values covered by continuity. The latter has a unique root: its left side has derivative ${(p-a){\omega}_p''(a)>0}$ on $(0,p)$, limit $-\infty$ at zero, and positive value at $p$. Thus, ${w=-{\omega}_p'(\alpha)>0}$. The tangent bound on $[0,p]$, the unique minimum ${{\omega}_p(1/2)=0}$ on $[p,1-p]$, and symmetry imply equality exactly at $\alpha,1/2,1-\alpha$, with respective Bayes rules constant $0$, identity, and constant $1$.

For $R_0$ in~\eqref{eq:example-rate-output}, the elementary bounds $h_2(1/4)>(4/5)\log2$, $\log(11/9)<(1/3)\log2$, and $\log3>(3/2)\log2$ imply
\begin{align}
{\omega}_{2/5}(1/4)&<(1/5)\log2,\\
{\omega}_{2/5}'(1/4)&=\tfrac15\log(11/9)-\log3<-(43/30)\log2,\\
{\omega}_{2/5}(1/4)+(2/5-1/4){\omega}_{2/5}'(1/4)&<-(3/200)\log2<0,
\end{align}
where the first inequality follows from $h_2(p+(1-2p)/4)\leq\log2$, the second follows by differentiation, and the third follows by adding the two bounds.
These inequalities also lead to
\begin{align}
0<{\omega}_p(\alpha)<R_0/2,\label{eq:example-strict-feasibility}
\end{align}
where~\eqref{eq:example-strict-feasibility} follows from $\alpha>1/4$ and ${{\omega}_p(\alpha)<(1/5)\log2}$, by the preceding tangent bound, and from $R_0>(2/5)\log2$, since $h_2(1/3)>(9/10)\log2$ and $h_2(1/10)<(1/2)\log2$.

For the prediction loss in~\eqref{eq:prediction-cost-log-form} and any $Q\in\calP(\{0,1\})$, define
\begin{align}
\delta_{X_{\rm p}}(Q):=\min_{{z_{\rm p}}\in\{0,1,*\}}\sum_{{x_{\rm p}}\in\{0,1\}}Q({x_{\rm p}})d_{\rm p}({x_{\rm p}},{z_{\rm p}}).
\end{align}
Then
\begin{align}
-H(Q)+\delta_{X_{\rm p}}(Q)
=\log(1-q_{\rm p})+\min_{{z_{\rm p}}\in\{0,1,*\}}\operatorname{KL}(Q\|Q_{z_{\rm p}}^{\rm p})
\geq\log(1-q_{\rm p}),\label{eq:example-log-support}
\end{align}
where~\eqref{eq:example-log-support} follows from the entropy--cross-entropy identity and nonnegativity of relative entropy.
Equality holds precisely at ${Q_0^{\rm p}},{Q_1^{\rm p}},{Q_*^{\rm p}}$, each with its unique minimizing prediction.

For any posterior $Q_{X_{\rm p}X_{\rm b}}\in\calP(\{0,1\}^2)$, write $Q_{X_{\rm b}}$ and $Q_{{X_{\rm p}}|{X_{\rm b}}}$ for its marginal and conditional distributions, respectively. Under the joint distribution $Q_{X_{\rm p}X_{\rm b}}P_{Y|{X_{\rm b}}}$, let $Y={X_{\rm b}}\oplus {N}$, where ${N}\sim\Ber(p)$ is independent of $({X_{\rm p}},{X_{\rm b}})$. For this binary channel, $P_{Y|{X_{\rm b}}}(y|{x_{\rm b}})=(1-p)\bbI\{y={x_{\rm b}}\}+p\bbI\{y\ne {x_{\rm b}}\}$ for $({x_{\rm b}},y)\in\{0,1\}^2$. Define
\begin{align}
\delta_{X_{\rm p}}^Y(Q_{X_{\rm p}X_{\rm b}}):=\sum_{y\in\{0,1\}}\min_{{z_{\rm p}}\in\{0,1,*\}}
\sum_{({x_{\rm p}},{x_{\rm b}})\in\{0,1\}^2}Q_{X_{\rm p}X_{\rm b}}({x_{\rm p}},{x_{\rm b}})P_{Y|{X_{\rm b}}}(y|{x_{\rm b}})d_{\rm p}({x_{\rm p}},{z_{\rm p}}).\label{eq:prediction-risk-with-side-information}
\end{align}
Recall the definition of $\delta_p$ in~\eqref{eq:binary-posterior-functions}. The Cartesian reproduction alphabet leads to total risk ${\delta_{X_{\rm p}}^Y(Q_{X_{\rm p}X_{\rm b}})+w\delta_p(Q_{X_{\rm b}}(1))}$. Then
\begin{align}
&-H_Q({X_{\rm p}},{X_{\rm b}}|Y)+\delta_{X_{\rm p}}^Y(Q_{X_{\rm p}X_{\rm b}})+w\delta_p(Q_{X_{\rm b}}(1))\notag\\
&\quad\geq -H_Q({X_{\rm b}}|Y)+w\delta_p(Q_{X_{\rm b}}(1))
+\sum_{{x_{\rm b}}\in\supp(Q_{X_{\rm b}})}Q_{X_{\rm b}}({x_{\rm b}})\big[-H_Q({X_{\rm p}}|{X_{\rm b}}={x_{\rm b}})+\delta_{X_{\rm p}}(Q_{{X_{\rm p}}|{X_{\rm b}}={x_{\rm b}}})\big]\label{eq:v20-product-support-1}\\
&\quad\geq wp-h_2(p)+\log(1-q_{\rm p}),\label{eq:example-product-support}
\end{align}
where~\eqref{eq:v20-product-support-1} follows from the entropy chain rule, ${{X_{\rm p}}-{X_{\rm b}}-Y}$, and the fact that revealing ${{X_{\rm b}}}$ reduces prediction risk;~\eqref{eq:example-product-support} follows from~\eqref{eq:example-binary-support} and~\eqref{eq:example-log-support}, omitting zero-mass terms. Equality requires a binary contact and prediction contacts ${Q_{{X_{\rm p}}|{X_{\rm b}}={x_{\rm b}}}}$. Both binary masses and both channel probabilities are positive; differing prediction contacts would require different unique actions for the same $y$, making the first inequality strict. Thus, ${Q_{{X_{\rm p}}|{X_{\rm b}}=0}=Q_{{X_{\rm p}}|{X_{\rm b}}=1}}$.
The equality set is exactly the nine product contacts; mixture labels may still be dependent.

\subsubsection{\texorpdfstring{Optimal Mixtures and Fluctuations}{Optimal Mixtures and Fluctuations}}
The preceding bounds lead to prediction contacts $Q_0^{\rm p},Q_1^{\rm p},Q_*^{\rm p}$ and binary contacts $\Ber(\alpha),\Ber(1/2),\Ber(1-\alpha)$. Fix an identity-rule weight $\vartheta_{\rm m}\in[0,1]$ and assign binary weights $((1-\vartheta_{\rm m})/2,\vartheta_{\rm m},(1-\vartheta_{\rm m})/2)$. For $R_0$ in~\eqref{eq:example-rate-output}, set
\begin{align}
\vartheta_{\rm p}:=\frac{r-q_{\rm p}}{1-2q_{\rm p}},\qquad \vartheta_*(\vartheta_{\rm m}):=\frac{(1-\vartheta_{\rm m}){\omega}_p(\alpha)}{R_0}.\label{eq:example-face}
\end{align}
Assign prediction weights ${((1-\vartheta_*)(1-\vartheta_{\rm p}),(1-\vartheta_*)\vartheta_{\rm p},\vartheta_*)}$ with ${\vartheta_*=\vartheta_*(\vartheta_{\rm m})}$, and take the product mixture. Then
\begin{align}
\bbE[d(X,g(U,Y))]
&=h_2(q_{\rm p})+\log(1-q_{\rm p})+(1-\vartheta_{\rm m}){\omega}_p(\alpha)
+w[\alpha+(p-\alpha)\vartheta_{\rm m}]=D,\label{eq:example-mixture-values}\\
I(X;U|Y)&=(1-\vartheta_*(\vartheta_{\rm m}))R_0+(1-\vartheta_{\rm m}){\omega}_p(\alpha)=R_0,\label{eq:v20-example-mixture-2}
\end{align}
Equation~\eqref{eq:example-mixture-values} adds the prediction and binary distortions $h_2(q_{\rm p})+\log(1-q_{\rm p})+\vartheta_*R_0$ and $w[\alpha+(p-\alpha)\vartheta_{\rm m}]$; substituting ${\omega}_p(\alpha)=w(p-\alpha)$ leads to $D$. Equation~\eqref{eq:v20-example-mixture-2} adds their information quantities and substitutes~\eqref{eq:example-face}. These mixtures have marginals ${P_{X_{\rm p}}}$ and $\Ber(1/2)$, and ${0\leq \vartheta_*<1/2}$ follows from~\eqref{eq:example-strict-feasibility}, so the product mixture is feasible.
For any finite auxiliary $U$ satisfying $U-X-Y$ and reconstruction function $g$,
\begin{align}
I(X;U|Y)+\bbE[d(X,g(U,Y))]
\geq h_2(r)+wp+\log(1-q_{\rm p})=R_0+D,\label{eq:example-global-support}
\end{align}
where~\eqref{eq:example-global-support} follows by averaging~\eqref{eq:example-product-support}, using $H(X|Y)=h_2(r)+h_2(p)$, and lower bounding each reconstruction risk by its Bayes risk.
Together with~\eqref{eq:example-mixture-values}, we have $R_{\mathrm{WZ}}(D)=R_0$. At ${\vartheta_{\rm m}=0}$, ${0<\vartheta_*(0)<1}$; varying ${\vartheta_*}$ locally preserves the source mean and supporting equality while changing distortion at speed $R_0>0$. Thus, the rate-distortion function is locally affine with slope $-1$, so $\lambda^*=1$. This proves Claim~(i).

Every optimum attains the supporting equality, has distortion $D$, and uses the nine contacts with Bayes reconstruction. The fair binary marginal forces weights ${((1-\vartheta_{\rm m})/2,\vartheta_{\rm m},(1-\vartheta_{\rm m})/2)}$; the distortion constraint forces prediction star weight ${\vartheta_*(\vartheta_{\rm m})}$. Every ${\vartheta_{\rm m}\in[0,1]}$ is attainable, even though the labels may be dependent.

Both endpoints use at most $|\calX|+1=5$ labels. At ${\vartheta_{\rm m}=1}$, two suffice. At ${\vartheta_{\rm m}=0}$, assign each binary constant rule weight $1/2$, with conditional prediction star weights $0$ and ${2\vartheta_*(0)}$. Each prediction mixture has mean ${P_{X_{\rm p}}}$; ${2\vartheta_*(0)<1}$ permits five positive joint labels with the same distortion and rate.

At a product contact, the information density splits into the prediction and binary contributions. For any $({x_{\rm p}},{z_{\rm p}})\in\allowbreak\{0,1\}\times\allowbreak\{0,1,*\}$,
\begin{align}
\log\frac{Q_{z_{\rm p}}^{\rm p}({x_{\rm p}})}{P_{X_{\rm p}}({x_{\rm p}})}+d_{\rm p}({x_{\rm p}},{z_{\rm p}})=-\log P_{X_{\rm p}}({x_{\rm p}})+\log(1-q_{\rm p}),\label{eq:example-source-identity}
\end{align}
where substituting~\eqref{eq:prediction-cost-log-form} cancels ${\log Q_{z_{\rm p}}^{\rm p}({x_{\rm p}})}$.
The common source information is
\begin{align}
\jW{({x_{\rm p}},{x_{\rm b}})}&=-\log P_{X_{\rm p}}({x_{\rm p}})-h_2(q_{\rm p}),\label{eq:example-common-source}\\
\jW{({x_{\rm p}},{x_{\rm b}})}-R_{\mathrm{WZ}}(D)&=-\log P_{X_{\rm p}}({x_{\rm p}})-h_2(r),
\end{align}
where~\eqref{eq:example-common-source} follows from~\eqref{eq:example-source-identity} and~\eqref{eq:stationarity}: at the constant-$0$ binary contact, the conditional means of information plus weighted distortion differ by ${{\omega}_p'(\alpha)+w=0}$ and average to ${{\omega}_p(\alpha)+w\alpha=wp}$. Symmetry handles constant $1$; identity has zero information and mean distortion ${wp}$. The second identity follows by subtracting $R_0$ and equals ${({x_{\rm p}}-1/3)\log2}$ with variance $2(\log2)^2/9$ when $r=1/3$. This proves Claim~(ii).

Using $\varpi$ and the variances from~\eqref{eq:example-rate-output}--\eqref{eq:example-variances}, every optimum with identity weight $\vartheta_{\rm m}$ has
\begin{align}
\Sigma(P)=\diag((1-\vartheta_{\rm m})\sigma_{\rmI}^2,\vartheta_{\rm m}\sigma_{\rmD}^2)=\Sigma_{\vartheta_{\rm m}},\label{eq:example-covariance-calculation}
\end{align}
where~\eqref{eq:example-covariance-calculation} follows from~\eqref{eq:cov} and~\eqref{eq:example-variances}. Given $(X,U)$, prediction information and distortion are fixed. Binary constant contacts have fixed distortion and information variance $\sigma_{\rmI}^2$ from the observation gap ${\log((1-\varpi)/\varpi)}$; identity has zero information and distortion variance $\sigma_{\rmD}^2$. Thus, cross-covariances vanish, regardless of label dependence. Equation~\eqref{eq:example-full-family} follows from~\eqref{eq:example-covariance-calculation} and attainability for every ${\vartheta_{\rm m}\in[0,1]}$. This proves Claim~(iii).

Since ${\vartheta_*(\vartheta_{\rm m})<1}$, both ${Q_0^{\rm p}},{Q_1^{\rm p}}$ have positive weight. Their loss difference varies with ${{x_{\rm p}}}$ and cannot be canceled by a function of ${{x_{\rm b}}}$ on the full-support posteriors. Hence the variance in~\eqref{eq:existing-direction} is positive.

To optimize the covariance family, we next compare its endpoint variances. Since ${\alpha<p<\varpi<1/2}$, we have
\begin{align}
{\omega}_p(\alpha)=\int_\alpha^\varpi h_2'(a)\,da
>(p-\alpha)h_2'(\varpi)>0,\label{eq:v20-variance-order-1}
\end{align}
where~\eqref{eq:v20-variance-order-1} follows from the definition of ${{\omega}_p}$ and $h_2'>0$ decreasing on $(0,1/2)$, since ${\varpi-\alpha>p-\alpha}$.
Dividing by $p-\alpha$ shows ${w>h_2'(\varpi)>0}$, and consequently $0<\sigma_{\rmI}<\sigma_{\rmD}$.

\subsubsection{\texorpdfstring{Gaussian Envelopes and Fixed-Channel Comparison}{Gaussian Envelopes and Fixed-Channel Comparison}}
For the family~\eqref{eq:example-full-family}, fix $(\vartheta_{\rm m},\ell)\in[0,1]\times\bbR$. Let $(J_{\vartheta_{\rm m}},G_{\vartheta_{\rm m}})\sim\allowbreak \calN(\bzero_2,\Sigma_{\vartheta_{\rm m}})$. The envelope~\eqref{eq:example-envelope-definition} satisfies
\begin{align}
F^{\rm c}_{\vartheta_{\rm m}}(\ell)
\leq\Phi\left(\frac{\ell}{\sqrt{(1-\vartheta_{\rm m})\sigma_{\rmI}^2+\vartheta_{\rm m}\sigma_{\rmD}^2}}\right),\label{eq:example-halfspace}
\end{align}
where~\eqref{eq:example-halfspace} follows because each rectangle in~\eqref{eq:example-envelope-definition} lies in ${\{J_{\vartheta_{\rm m}}+G_{\vartheta_{\rm m}}\leq\ell\}}$. The endpoint identities ${F^{\rm c}_0(\ell)=\Phi(\ell/\sigma_{\rmI})}$ and ${F^{\rm c}_1(\ell)=\Phi(\ell/\sigma_{\rmD})}$ follow by choosing $t=0$ and $t=\ell$, respectively. The bound is largest at the smallest variance for $\ell\geq0$ and largest variance for $\ell<0$. Equation~\eqref{eq:example-envelopes} follows from~\eqref{eq:example-halfspace} and these endpoint equalities, and~\eqref{eq:example-adaptive-integral} follows by averaging over the source Gaussian. The positive threshold enlargement in Section~\ref{sec:achievability} covers these singular endpoints. This proves Claim~(iv). We next keep the test channel fixed while averaging over the source fluctuation.

For independent $\Lambda\sim \calN(0,V)$ and any $(\vartheta_{\rm m},L)\in[0,1]\times\bbR$, source averaging leads to
\begin{align}
\bbE[F^{\rm c}_{\vartheta_{\rm m}}(L-\Lambda)]
\leq\Phi\left(\frac{L}{\sqrt{\sigma_{\rms}^2+(1-\vartheta_{\rm m})\sigma_{\rmI}^2+\vartheta_{\rm m}\sigma_{\rmD}^2}}\right),\label{eq:example-fixed-upper}
\end{align}
where~\eqref{eq:example-fixed-upper} follows from~\eqref{eq:example-halfspace} and convolution of the independent Gaussian source and channel fluctuations.
For $L\geq0$, the upper bound is attained at $\vartheta_{\rm m}=0$, giving the equality in~\eqref{eq:example-fixed}; for $L<0$, it is attained at $\vartheta_{\rm m}=1$. Since $F_{\mathrm{f}}(0)=1/2$,~\eqref{eq:example-fixed-quantile} applies when $\veps<1/2$. It remains to prove the strict comparison in Claim~(v).

For every finite $L\geq0$,
\begin{align}
F_*(L)-F_{\mathrm{f}}(L)
=\int_{L/\sigma_{\rms}}^\infty
\Biggl[\Phi\left(\frac{L-\sigma_{\rms}\bar{\ell}_{\rm s}}{\sigma_{\rmD}}\right)
-\Phi\left(\frac{L-\sigma_{\rms}\bar{\ell}_{\rm s}}{\sigma_{\rmI}}\right)\Biggr]{\Phi'}(\bar{\ell}_{\rm s})\,d\bar{\ell}_{\rm s}>0,\label{eq:example-strict}
\end{align}
where~\eqref{eq:example-strict} follows by subtracting~\eqref{eq:example-fixed} from~\eqref{eq:example-adaptive-integral}. Its integrand is positive since ${L-\sigma_{\rms}\bar{\ell}_{\rm s}<0}$ and $\sigma_{\rmD}>\sigma_{\rmI}$. For $L<0$, integrate instead over ${\bar{\ell}_{\rm s}<L/\sigma_{\rms}}$. This proves Claim~(v) and completes Lemma~\ref{lem:nonunique-example}. Continuity and strict monotonicity imply strictly ordered quantiles.

\bibliographystyle{IEEEtran}
\bibliography{IEEEfull_lin}
\end{document}

%% file: preamble.tex
\usepackage{fix-cm}
\usepackage[mathscr]{eucal}
\usepackage[cmex10]{amsmath}
\usepackage{amssymb,amsthm,amsfonts,latexsym,mathtools}
\usepackage[utf8]{inputenc}
\usepackage[T1]{fontenc}
\usepackage{graphicx,bm,xcolor,url,array,verbatim,textcomp,mathrsfs}
\usepackage{booktabs,microtype,cite}
\usepackage{tikz,flafter,placeins,needspace}
\usetikzlibrary{arrows.meta,positioning}
\usepackage[colorlinks=true,linkcolor=blue,urlcolor=blue,citecolor=red,anchorcolor=green]{hyperref}

\newcommand{\openone}{\leavevmode\hbox{\small1\normalsize\kern-.33em1}}

\catcode`~=11 \def\UrlSpecials{\do\~{\kern -.15em\lower .7ex\hbox{~}\kern .04em}} \catcode`~=13

\newcommand{\calA}{\mathcal{A}}
\newcommand{\calB}{\mathcal{B}}
\newcommand{\calC}{\mathcal{C}}
\newcommand{\calD}{\mathcal{D}}
\newcommand{\calE}{\mathcal{E}}
\newcommand{\calF}{\mathcal{F}}
\newcommand{\calG}{\mathcal{G}}

\newcommand{\calI}{\mathcal{I}}
\newcommand{\calJ}{\mathcal{J}}
\newcommand{\calK}{\mathcal{K}}
\newcommand{\calL}{\mathcal{L}}
\newcommand{\calM}{\mathcal{M}}
\newcommand{\calN}{\mathcal{N}}

\newcommand{\calP}{\mathcal{P}}

\newcommand{\calR}{\mathcal{R}}
\newcommand{\calS}{\mathcal{S}}
\newcommand{\calT}{\mathcal{T}}
\newcommand{\calU}{\mathcal{U}}

\newcommand{\calX}{\mathcal{X}}
\newcommand{\calY}{\mathcal{Y}}
\newcommand{\calZ}{\mathcal{Z}}

\newcommand{\bd}{\mathbf{d}}

\newcommand{\be}{\mathbf{e}}

\newcommand{\rmD}{\mathrm{D}}

\newcommand{\rmi}{\mathrm{i}}
\newcommand{\rmI}{\mathrm{I}}

\newcommand{\rms}{\mathrm{s}}

\newcommand{\bbE}{\mathsf{E}}

\newcommand{\bbI}{\openone}

\newcommand{\bbN}{\mathbb{N}}

\newcommand{\bbP}{\mathbb{P}}

\newcommand{\bbR}{\mathbb{R}}

\DeclareMathAlphabet{\mathbsf}{OT1}{cmss}{bx}{n}
\DeclareMathAlphabet{\mathssf}{OT1}{cmss}{m}{sl}

\DeclareSymbolFont{bsfletters}{OT1}{cmss}{bx}{n}  
\DeclareSymbolFont{ssfletters}{OT1}{cmss}{m}{n}
\DeclareMathSymbol{\bsfGamma}{0}{bsfletters}{'000}
\DeclareMathSymbol{\ssfGamma}{0}{ssfletters}{'000}
\DeclareMathSymbol{\bsfDelta}{0}{bsfletters}{'001}
\DeclareMathSymbol{\ssfDelta}{0}{ssfletters}{'001}
\DeclareMathSymbol{\bsfTheta}{0}{bsfletters}{'002}
\DeclareMathSymbol{\ssfTheta}{0}{ssfletters}{'002}
\DeclareMathSymbol{\bsfLambda}{0}{bsfletters}{'003}
\DeclareMathSymbol{\ssfLambda}{0}{ssfletters}{'003}
\DeclareMathSymbol{\bsfXi}{0}{bsfletters}{'004}
\DeclareMathSymbol{\ssfXi}{0}{ssfletters}{'004}
\DeclareMathSymbol{\bsfPi}{0}{bsfletters}{'005}
\DeclareMathSymbol{\ssfPi}{0}{ssfletters}{'005}
\DeclareMathSymbol{\bsfSigma}{0}{bsfletters}{'006}
\DeclareMathSymbol{\ssfSigma}{0}{ssfletters}{'006}
\DeclareMathSymbol{\bsfUpsilon}{0}{bsfletters}{'007}
\DeclareMathSymbol{\ssfUpsilon}{0}{ssfletters}{'007}
\DeclareMathSymbol{\bsfPhi}{0}{bsfletters}{'010}
\DeclareMathSymbol{\ssfPhi}{0}{ssfletters}{'010}
\DeclareMathSymbol{\bsfPsi}{0}{bsfletters}{'011}
\DeclareMathSymbol{\ssfPsi}{0}{ssfletters}{'011}
\DeclareMathSymbol{\bsfOmega}{0}{bsfletters}{'012}
\DeclareMathSymbol{\ssfOmega}{0}{ssfletters}{'012}

\newcommand{\hatj}{\hat{j}}

\newcommand{\hatP}{\hat{P}}

\newcommand{\hatt}{\hat{t}}
\newcommand{\hatT}{\hat{T}}

\newcommand{\tilx}{\tilde{x}}
\newcommand{\tilX}{\tilde{X}}

\newcommand{\tily}{\tilde{y}}
\newcommand{\tilY}{\tilde{Y}}

\newcommand{\tilZ}{\tilde{Z}}

\newcommand{\bard}{\bar{d}}

\newcommand{\bart}{\bar{t}}

\newcommand{\bbeta}{\bm{\beta}}

\newcommand{\btheta}{\bm{\theta}}

\newcommand{\veps}{\varepsilon}

\newcommand{\bxi}{\bm{\xi}}

\newcommand{\bXi}{\bm{\Xi}}

\newcommand{\Pen}{P_{\mathrm{e}}^{(n)}}

\DeclareMathOperator*{\argmin}{arg\,min}

\DeclareMathOperator{\diag}{diag}

\DeclareMathOperator{\tr}{tr}

\DeclareMathOperator{\supp}{supp}

\DeclareMathOperator{\var}{\mathsf{Var}}

\DeclareMathOperator{\cov}{\mathsf{Cov}}

\newcommand{\Ber}{\mathrm{Bern}}

\newcommand{\bzero}{\mathbf{0}}
\newcommand{\bone}{\mathbf{1}}

\newtheorem{theorem}{Theorem} 
\newtheorem{lemma}{Lemma}

\newtheorem{proposition}{Proposition}
\newtheorem{corollary}{Corollary}

\newtheorem{definition}{Definition}

\DeclareMathOperator{\TC}{TC}
\DeclareMathOperator{\dist}{dist}
\newcommand{\jW}[1]{\jmath_{\mathrm{WZ}}(#1\mid D)}